\documentclass[11pt]{article}
\usepackage[top=1in, left=1in, right=1in, bottom=0.85in, headheight=15pt,headsep=10pt]{geometry}
\usepackage{titlesec} 
\usepackage{graphicx}
\usepackage{subfigure}
\usepackage{booktabs} 
\usepackage{fancyhdr}
\usepackage[skip=0pt]{subcaption}
\usepackage{diagbox}
\usepackage[round]{natbib}
\usepackage[backref=page]{hyperref}

\usepackage{amsmath}
\usepackage{amssymb}
\usepackage{mathtools}
\usepackage{amsthm}
\usepackage[mathscr]{euscript}
\usepackage[capitalize]{cleveref}
\usepackage{cancel}    
\usepackage{url}            
\usepackage{booktabs}       
\usepackage{microtype}      
\usepackage{setspace}
\usepackage{amssymb}
\usepackage{url}
\usepackage{lettrine}
\usepackage{framed}
\usepackage{color}
\usepackage{xfrac}
\usepackage{nicefrac}
\usepackage[svgnames, x11names]{xcolor}
\usepackage{algorithm}
\usepackage{algpseudocode}
\usepackage{overpic}
\usepackage{dsfont}
\usepackage{booktabs} 
\usepackage{framed}
\usepackage{graphicx}
\usepackage{dsfont}
\usepackage{framed}
\usepackage{bm}
\usepackage{bbm}
\usepackage{overpic}
\usepackage{tikz}
\usepackage{titletoc}
\usepackage{wrapfig}
\usetikzlibrary{decorations.pathreplacing}
\newlength{\bibitemsep}
\newlength{\bibparskip}
\let\oldthebibliography\thebibliography
\renewcommand\thebibliography[1]{%
  \oldthebibliography{#1}%
  \setlength{\parskip}{\bibitemsep}%
  \setlength{\itemsep}{\bibparskip}%
}
\theoremstyle{plain}
\newtheorem{theorem}{Theorem}[section]

\newtheorem{lemma}[theorem]{Lemma}
\newtheorem{corollary}[theorem]{Corollary}
\newtheorem{example}{Example}

\theoremstyle{definition}
\newtheorem{definition}{Definition}
\newtheorem{assumption}[theorem]{Assumption}
\newtheorem{claim}[theorem]{Claim}
\newtheorem{observation}{Observation}
\newtheorem{fact}{Fact}
\theoremstyle{remark}
\newtheorem{remark}[theorem]{Remark}
\usepackage{enumitem}
\usepackage[textsize=tiny]{todonotes}
\newcommand{\ind}[1]{\mathbb{I}{\left[#1\right]}}

\newcommand{\otherbid}[1]{\boldsymbol{\beta}_{-}^{#1}}
\newcommand{\ordotherbid}[2]{\boldsymbol{\beta}_{-, #1}^{-(#2)}}
\newcommand{\ordotherbidnew}[1]{\boldsymbol{\beta}_{-}^{-(#1)}}

\hypersetup{
	colorlinks=true,
	linkcolor=Blue4,
	urlcolor=black,
	citecolor=Blue4,
	linkbordercolor={0 0 1}
}

\def\R{\mathbb{R}}
\def\P{\mathbb{P}}
\def\N{\mathbb{N}}
\def\E{\mathbb{E}}

\def\v{\mathbf{v}}

\def\prob{\mathbb{P}}
\def\path{\mathfrak{p}}
\newcommand{\negin}[1]{{\color{brown}[\textsc{NG}: \emph{#1}]}}
\newcommand{\sourav}[1]{{\color{blue}{#1}}}

\usepackage{comment}

\allowdisplaybreaks[4]
\newcommand{\argmax}{\mathop{\mathrm{argmax}}}

\newcommand{\inv}[1]{\dfrac{1}{#1}}

\newcommand{\floor}[1]{\left\lfloor #1 \right\rfloor}
\newcommand{\ceil}[1]{\left\lceil #1 \right\rceil}

\def\l{\ell}
\def\hist{\mathcal{H}_{-}}
\def\nature{\mathscr{B}_{c}}

\def\generic{\mathscr{B}}
\def\val{V}
\def\price{P}
\def\numbid{m}
\def\allbids{\boldsymbol{\beta}}
\def\ibid{\mathbf{b}}

\def\maxbid{M}

\newcommand{\optbid}[1]{\mathbf{b}_{#1}^{\mathsf{OPT}}\left(\hist\right)}
\newcommand{\safebid}[1]{\mathbf{b}_{#1}^{\mathsf{SAFE}}\left(\hist\right)}
\newcommand{\psum}[1]{{W}_{#1}}

\newcommand{\optvalue}[1]
{V^\mathsf{OPT}_{#1}\left(\hist\right)}
\newcommand{\safevalue}[1]
{V^\mathsf{SAFE}_{#1}\left(\hist\right)}

\newcommand{\optufclass}[1]{\mathscr{U}_{#1}^\star}
\newcommand{\optoneufclass}[1]{\mathscr{S}_{#1}^\star}
\newcommand{\ufclass}[1]{\mathscr{U}_{#1}}
\newcommand{\feasclass}[1]{\mathscr{F}^{\hist}_{#1}}
\newcommand{\onefeasclass}[1]{\mathscr{G}^{\hist}_{#1}}
\newcommand{\resonefeasclass}[1]{\widetilde{\mathscr{G}}^{\hist}_{#1}}
\newcommand{\oneufclass}[1]{\mathscr{S}_{#1}}

\titleformat{\section}
  {\normalfont\bfseries\sffamily\large}
  {\thesection.}{0.5em}{}

\titleformat{\subsection}
  {\normalfont\bfseries\sffamily\normalsize}
  {\thesubsection.}{0.5em}{}

\titleformat{\subsubsection}
  {\normalfont\bfseries\sffamily\normalsize}
  {\thesubsubsection.}{0.5em}{}
  
\title{\vspace{-1cm}Learning Safe Strategies for Value Maximizing Buyers in Uniform Price Auctions\footnote{A preliminary version appeared at the 42nd International Conference on Machine Learning (ICML) 2025.}}

\author{
  Negin Golrezaei\\
  \footnotesize Sloan School of Management, Massachusetts Institute of Technology, \href{mailto:golrezae@mit.edu}{\textsf{golrezae@mit.edu}}\\[1ex]
  Sourav Sahoo\\
  \footnotesize Operations Research Center, Massachusetts Institute of Technology, \href{mailto:sourav99@mit.edu}{\textsf{sourav99@mit.edu}}
}

\date{}
\begin{document}
\maketitle
\setstretch{1.33}

\begin{abstract}

We study the bidding problem in repeated uniform price multi-unit auctions, a format prevalent in emission permit auctions, Treasury auctions and energy markets, from the perspective of a \textit{value-maximizing} buyer. The buyer aims to maximize their cumulative value over $T$ rounds while adhering to per-round return-on-investment~(RoI) constraints in a strategic~(or adversarial) environment. The buyer uses $\numbid$-\textit{uniform bidding} format, submitting $\numbid$ bid-quantity pairs \((b_i, q_i)\) to demand \( q_i \) units at bid \( b_i \), with $\numbid\ll \maxbid$ in practice, where, $\maxbid$ is the maximum demand of the buyer. 

We introduce the notion of \textit{safe} bidding strategies as those that satisfy the RoI constraint irrespective of competing bids. Despite the stringent requirement, we show that these strategies satisfy a mild no-overbidding condition, depend only on the bidder’s valuation curve, and the bidder can focus on a finite subset without loss of generality. While the number of strategies in this subset is $O(\maxbid^\numbid)$, leveraging a key insight on the decomposition of a strategy's obtained value across its constituent bid-quantity pairs, we develop a polynomial-time algorithm to learn the optimal safe strategy that achieves a regret of $\widetilde{O}(\maxbid\sqrt{\numbid T})$ and $\widetilde{O}(\maxbid^{2}\numbid^{3/2}\sqrt{T})$ in the full-information and bandit settings respectively, where regret is measured against a clairvoyant that selects the fixed hindsight optimal safe strategy. We complement the upper bounds by establishing a $O(\maxbid\sqrt{T})$ regret lower bound.


We then assess the robustness of safe strategies against the hindsight-optimal strategy from a richer class of strategies. We define the \emph{richness ratio} $\alpha \in (0,1]$ as the minimum ratio of the value of the optimal safe strategy to that of the optimal strategy from richer class and construct hard instances showing the tightness of $\alpha$. Our algorithm achieves $\alpha$-approximate sublinear regret against these stronger benchmarks. Simulations on semi-synthetic real-world auction data demonstrate that empirical richness ratios are significantly better than the tight theoretical bounds. Finally, we show that (variants of) safe strategies and the proposed learning algorithm are well-suited to cater to more nuanced behavioral models of the buyer and the competitors. The simplicity, robustness, and adaptability of safe bidding strategies make them highly promising for practitioners.
\end{abstract}
\clearpage
\setcounter{tocdepth}{2}
\tableofcontents

\section{Introduction}
 
In a uniform price multi-unit auction, the auctioneer sells \( K \) identical units of a single good to buyers who may demand multiple units, with the per-unit price set at the \( K^{th} \) highest bid. These auctions are widely used to allocate scarce resources in critical {avenues such as emissions permit auctions, wholesale electricity markets}, and Treasury auctions. To succeed in these markets, bidders must develop effective bidding strategies that balance long-term value maximization with financial risk management, all while operating under limited information, uncertainty in competing bids and strategic~(or adversarial) behavior of other bidders, among other challenges. 

In this paper, we model the bidders as a value-maximizing agent  with per-round return-on-investment (RoI) constraints. This model reflects real-world decision-making in industries where managing financial risks is as {important} as maximizing value. Unlike traditional auction theory, which assumes bidders with quasilinear utility (i.e., utility decreases linearly with payments), many practical settings involve bidders who optimize total value while adhering to financial constraints such as RoI or budgets. These scenarios frequently arise in industries where agents or algorithms bid on behalf of clients, optimizing for high-level objectives within strict constraints--a context akin to the principal-agent framework~\citep{fadaei2016truthfulness, aggarwal2024auto}. 

For example, in online advertising markets, performance marketing agencies often manage autobidding systems for clients who seek to maximize conversions or clicks while ensuring the average cost-per-click or return-on-ad-spend remains within specified limits \citep{lucier2023autobidders, balseiro2021robust, deng2023autobidding}. Similarly, in the EU Emissions Trading System (EU ETS), many {compliant companies} participate in {emission permits} auctions indirectly through third-party {firms such as MK Brokers JSC\footnote{\scriptsize \url{https://mkb.bg/en/services/sustainability-products/}.}, Act Group\footnote{\scriptsize\url{https://www.actgroup.com/products-solutions/compliance/emissions-compliance-products}.}, and AFS Group\footnote{\scriptsize\url{https://www.afsgroup.nl/exchange-traded-energy-instruments}.}}. These intermediaries bid on behalf of clients often under contractual constraints on allowable expenditures or minimum cost-effectiveness of permit acquisition \citep{ellerman2010pricing}. 

Building on this, we study the bidding problem in repeated uniform price auctions over $T$ rounds, from the perspective of a single bidder who seeks to maximize cumulative value while adhering to per-round RoI constraints—that is, ensuring that the value obtained in each round is at least a fixed multiple of the corresponding payment~(see \cref{def:RoI}). To address practical bidding interfaces, we assume bidders adopt \textit{$\numbid$-uniform bidding} format, a generalization of the uniform bidding format~\citep{de2013inefficiency, birmpas2019tight} for some $\numbid\in\N$. In a $\numbid$-uniform strategy $\ibid:=\langle (b_1, q_1), \dots, (b_\numbid, q_\numbid)\rangle$, bidders bid $b_1$ for the first $q_1$ units, $b_2$ for the next $q_2$ units and so on~(see \cref{def:m-unif-bid}).

\subsection{Our Contributions}

\textbf{Safe Bidding Strategies (\cref{ssec:uf-strategies}).}  
To ensure that the bidder satisfies the RoI constraint without knowing the (potentially adversarial) competing bids in the online setting, we introduce the concept of \textit{safe} bidding strategies. These strategies guarantee that the RoI constraint is met regardless of how the other participants bid. We show that these strategies follow a mild no-overbidding condition~(see~\cref{def:under-over_bid}) and the bidder can focus on a \textit{finite} subset of this class without loss of generality. For any $\numbid\in\N$, we characterize this finite subset as the class of $\numbid$-uniform safe \textit{undominated} bidding strategies, denoted by $\optoneufclass{\numbid}$ (\cref{thm:opt-bid}), and demonstrate that the strategies within this class depend solely on the bidder's valuation vector $[v_1, \dots, v_K]$ and exhibit a ``nested'' structure as illustrated in \cref{fig:nested}.

 \begin{figure}[!tbh]
    \centering
    \scalebox{0.7}{
    \begin{tikzpicture}
    \draw (0,0) rectangle (10,1);
    
    \foreach \i in {0,1,...,9}
        \draw (\i,0) -- (\i,1);
    
    \foreach \i in {1,2,...,10}
        \node[font=\Large] at (\i-0.5,0.5) {\(v_{\i}\)};
    \draw[decorate,decoration={brace,amplitude=5pt,mirror},xshift=0pt,yshift=-3pt]
    (0,0) -- (3,0) node[midway,yshift=-11pt] {\({w_3}\)};

    \draw[decorate,decoration={brace,amplitude=5pt,mirror, raise=16pt},xshift=0pt,yshift=-6pt]
    (0,0) -- (7,0) node[midway,yshift=-29pt] {\({w_7}\)};
    
    \draw[decorate,decoration={brace,amplitude=5pt,mirror, raise=34pt},xshift=0pt,yshift=-9pt]
    (0,0) -- (9,0) node[midway,yshift=-47pt] {\({w_9}\)};

\end{tikzpicture} } 
    \caption{Nested structure of the bidding strategies in $\optoneufclass{3}$. Consider the strategy $\ibid=\langle(w_3, 3), (w_7, 4), (w_9, 2)\rangle\in\optoneufclass{3}$, where we note that $Q_1= 3$, $Q_2=3+4=7$ and $Q_3 = 3+4+2=9$.  The $j^{th}$ highest bid (i.e., $b_j$) is the average of the first $Q_j=\sum_{\l\leq j}q_\l$ entries of the valuation vector, i.e., $b_j = w_{Q_j}$, where    $w_{j}=\frac{1}{j}\sum_{\ell\leq j}v_{\ell}$.
    } 
    \label{fig:nested}
\end{figure}

\textbf{Learning Safe Bidding Strategies (\cref{sec:learning-safe}).}  
Designing an algorithm to learn the optimal safe bidding strategy with at most $\numbid$ bid-quantity pairs poses significant challenges as the size of the decision space in $O(\maxbid^\numbid)$, where $\maxbid$ is the maximum number of units demanded by the bidder and in practice, $\numbid\ll\maxbid$. To this end, we first consider an offline setting where the competing bids are known \textit{a priori}. We show that the value obtained by a $\numbid$-uniform safe strategy can be decomposed across its $\numbid$ constituent bid-quantity pairs. Leveraging this decomposition, we construct a directed acyclic graph (DAG), of size $\text{poly}(\numbid, \maxbid)$, with carefully assigned edge weights and show that determining the maximum weight path in the DAG is equivalent to computing the optimal safe bidding strategy with at most $\numbid$ bid-quantity pairs~(\cref{thm:DAG-base-strategy}). This reduction allows us to compute the offline optimal safe strategy in $\text{poly}(\numbid, \maxbid)$ time, despite the exponential size of the original decision space. 

Building on this result, we study the online setting over $T$ rounds where the competing bids are generated by an oblivious adversary and unknown to the bidder when they submit their strategy.\footnote{For a detailed discussion on leveraging offline algorithms for no-regret learning, we refer readers to \citet{roughgarden2019minimizing, niazadeh2022online, branzei2023learning}.} Leveraging \cref{thm:DAG-base-strategy}, one can think of implementing a {na\"ive version of the} Hedge~(exponential weight updates) algorithm~\citep{freund1997decision} in the full information setting by considering each path in the DAG as an expert. However, such an approach is intractable as there are $O(\maxbid^\numbid)$ such paths~(one corresponding to each safe strategy). To tackle this, we leverage our DAG construction and its decomposition, and adopt ideas from the weight-pushing algorithm based on dynamic programming from \citet{takimoto2003path}. We propose an algorithm within $\text{poly}(\numbid, \maxbid)$ space and time complexity that achieves a regret of $\widetilde{O}(\maxbid\sqrt{\numbid T})$ in a full-information setting and $\widetilde{O}(\maxbid^{2}\numbid^{3/2} \sqrt{T})$ regret in a bandit setting, where the clairvoyant benchmark selects the fixed hindsight optimal safe strategy~(\cref{thm:full-info}).\footnote{Here, $\widetilde{O}(\cdot)$ hides the logarithmic factors in $\maxbid$.} Additionally, we establish a regret lower bound of $O(\maxbid\sqrt{T})$, that matches the regret upper bound in $T$~(\cref{thm:regret-LB}).  

The problem of learning to bid in multi-unit uniform price and pay-as-bid auctions has been explored recently by \citet{branzei2023learning, galgana2023learning, potfer2024improved}, assuming the bidders are quasilinear utility maximizers. Our work differs from them in two key ways: (a) we consider value-maximizing buyers with RoI constraints, a fundamentally different behavioral model, and (b) from a technical perspective, unlike prior works that require bid spaces to be discretized, our approach does not, due to the structure of safe strategies. Consequently, the time complexity becomes independent of $T$ in both online and offline settings\footnote{Assuming the discretization level in the offline setting is the same as that in the online setting. In the online setting, time complexity refers to per-round running time.}, and under bandit feedback, the dependence on $T$ in the regret bound improves from $T^{2/3}$ to $\sqrt{T}$ implying that it is easier to learn safe strategies for value maximizers compared to bidding strategies for quasilinear utility maximizers.


\textbf{Richer Classes of Strategies for the Clairvoyant (\cref{sec:learning-rich}).} We then evaluate the performance of the bidder's strategy class, i.e., safe bidding strategies with at most $\numbid$ bid-quantity pairs and the online learning algorithm from \cref{sec:learning-safe} against a clairvoyant benchmark that selects the optimal strategy from a richer class. Here, a richer class refers to any class of strategies that is a superset of the bidder's strategy class~(see examples below). Recall that the bidder chooses to follow the $\numbid$-uniform safe strategies to satisfy the RoI constraints with limited knowledge about the competing bids in any auction in the online setting. 

Comparing against such stronger benchmarks demonstrates the robustness of the bidder's strategy class. We quantify this robustness via the \textit{richness ratio} $\alpha \in (0, 1]$~(\cref{prop:choose-alpha}), analogous to an approximation ratio. We prove that for any richer class with richness ratio $\alpha$, the bidder achieves $\alpha$-approximate regret (as defined in \cref{eq:alpha-approx-regret}) of the same order as in \cref{thm:full-info}~(see \cref{thm:alpha-UB}). One of our main contributions is to compute the richness ratio $\alpha$ for various classes of bidding strategies. We define richness along two dimensions: (i) the safety of strategies (safe vs. only RoI-feasible), and (ii) the number of bid-quantity pairs allowed. A strategy is defined as RoI-feasible~(or feasible in short) if it satisfies the RoI constraint for a \textit{given} sequence of competing bids over $T$ rounds, unlike safe strategies that satisfy the constraint for \textit{every} such sequence. We choose these two dimensions of richness as the bidder follows the class of strategies that are (i) safe and (ii) have at most $\numbid$ bid-quantity pairs. 

When the clairvoyant selects the optimal strategy from the class of 
\begin{itemize}
    \item [(a)] strategies that are RoI-feasible and have at most $\numbid$ bid-quantity pairs, we get $\alpha=\frac{1}{2}$~(\cref{thm:Price_universal}). Notably, for this particular choice of richer class of strategies, $\alpha$ is \textit{independent} of $\numbid$. 
    \item [(b)] safe strategies with at most $\numbid'$ bid-quantity pairs, where $\numbid' \geq \numbid$, we obtain $\alpha=\frac{\numbid}{\numbid'}$~(\cref{thm:m-mbar}). Here, setting $\numbid'=\numbid$ refers to the case considered in \cref{sec:learning-safe} where the bidder and clairvoyant follow the same class of strategies. In this case, \cref{thm:m-mbar} correctly implies the optimal value of $\alpha=1$.
    \item [(c)] strategies that are RoI-feasible and have at most $\numbid'$ bid-quantity pairs, where $\numbid' \geq \numbid$, we get $\alpha=\frac{\numbid}{2\numbid'}$~(\cref{thm:mbar-non-safe}). Setting $\numbid'=\numbid$ in this case recovers the result of \cref{thm:Price_universal}.
\end{itemize}

Computing the richness ratio involves two different challenges: (a) deriving an upper bound on the ratio between the value obtained by the clairvoyant’s optimal strategy from the richer class and that of the optimal safe strategy in the worst case (see \cref{prop:choose-alpha} for details), and (b) proving that this upper bound is tight by constructing a hard problem instance for the given richer class of strategies. The key ideas for deriving the upper bounds are outlined in \cref{sec:learning-rich}. The hard problem instances that show the tightness of the bounds are highly non-trivial where we need to consider an instance of size that is exponential in \(\numbid\), with careful choice of competing bids and the valuation vector~(see details in \crefrange{apx:LB-2-m-gen}{apx:LB-2m}). 


\textbf{Extensions (\cref{sec:extensions}).} Finally, we show that safe strategies and the learning algorithm from \cref{sec:learning-safe} can be adapted to handle more nuanced modeling assumptions. First, as competitors may adapt to the bidder’s past actions in repeated auctions, we consider an adaptive adversary model for generating competing bids (as opposed to an oblivious one). In this setting, the learning algorithm with appropriate modifications still yields a high-probability regret bound of $\text{poly}(\numbid, \maxbid)\cdot\sqrt{T}$~(\cref{ssec:non-oblivious}). Second, when the RoI constraint is enforced over a sliding window of $T_0 \in \mathbb{N}$ rounds instead of each round, we propose a heuristic based on \textit{shifted} safe strategies for this setting that shows good empirical performance on semi-synthetic real-world auction data~(\cref{ssec:cumulative-roi}). Third, in the case of time-varying valuation vectors, that are either sampled from a distribution or generated adversarially from a finite set $\mathcal{V}$, we treat the valuation vector as contexts and show that a contextual version of our learning algorithms achieves $\text{poly}(\numbid, \maxbid, |\mathcal{V}|)\cdot\sqrt{T}$ regret in this setting~(\cref{ssec:time-varying}). These extensions highlight the practical appeal and flexibility of our approach.


\subsection{Managerial Insights}

Our findings offer several key implications for RoI-constrained value maximizing buyers. 
\begin{enumerate}
    \item Safe strategies are an intuitive class of bidding strategies due to their mild no-overbidding conditions. As the bidder can focus on a finite subset of the class without any loss {of generality}, designing and implementing the proposed learning algorithm in repeated auction environments is quite efficient from a practical perspective. Furthermore, the theoretical regret upper and lower bounds with $\sqrt{T}$ dependence indicate that the proposed algorithm is near-optimal, which is desirable to practitioners.
    \item We crucially show the robustness of safe strategies and the learning algorithm by proving that they are approximately optimal against stronger benchmarks. A key highlight here is that for any given $\numbid$, the cost of following safe strategies compared to the instance dependent RoI-feasible strategies, which are impossible to know \textit{a priori}, is uniformly bounded. So, a clairvoyant with prior knowledge of the adversarially generated competing bids in this setting can not be arbitrarily better than the bidder following safe strategies.
    \item Simulations on semi-synthetic data from EU Emission Trading System~(ETS) auctions show that the empirical richness ratios in real-world scenarios are significantly better than tight theoretical bounds, which reinforces the fact that tight bounds are attained under highly pathological settings and suggests that safe strategies can perform near-optimal in practice, even for small values  $\numbid$~(\cref{sec:sims}). 
    \item Finally, we show that variants of safe strategies and the proposed learning algorithm can be appropriately deployed in several practical extensions of the primary model considered in this work. This aspect is particularly interesting as it allows practitioners to tweak safe strategies and/or the proposed algorithm to design novel algorithms or heuristics to adapt to their objectives and risk appetite that may change over time. 
\end{enumerate}


\subsection{Related Work}\label{sec:related}
\textbf{Value Maximizers and RoI Constraints.} The concept of agents as value maximizers within financial constraints is a well-established notion in microeconomic theory \citep{mas1995microeconomic}. In mechanism design literature, one of the earliest explorations of value-maximizing agents was conducted by \citet{wilkens2016mechanism}. Their work primarily delved into the single-parameter setting, characterizing truthful auctions for value maximizers. Similarly, \citet{fadaei2016truthfulness} and \citet{lu2023auction} studied truthful (approximate) revenue-maximizing mechanisms in combinatorial markets tailored for such agents. More recently, \citet{tang2024towards} investigated auction design in a setting where buyers may be either quasilinear utility maximizers or RoI-constrained value maximizers. \citet{xu2024sponsored} proposed the allowance utility maximization model, which generalizes both quasilinear and value-maximizing preferences.

In recent years, there has been growing interest in RoI-constrained value maximizers, particularly in the context of autobidding and online advertising. One of the earliest works in this area is \citet{golrezaei2021auction}, which studied auction design for RoI-constrained buyers and validated the presence of such soft financial constraints using data from online ad auctions. Broadly, the literature in this space can be divided into two categories: (i) works that design optimal auctions under RoI constraints~\citep{balseiro2021robust, balseiro2021landscape, balseiro2022optimal, deng2021towards, deng2023autobidding}, and (ii) works that characterize optimal bidding strategies and/or develop learning algorithms in repeated auction settings~\citep{aggarwal2019autobidding, deng2023multi, golrezaei2023pricing, castiglioni2024online, aggarwal2025noregret, lucier2023autobidders}. Our work primarily aligns with the latter. We study value-maximizing buyers in uniform price auctions under RoI constraints and show that they can employ safe strategies that are efficiently learnable and robust against a variety of strong benchmarks. 

\textbf{Multi-unit Auctions.} In this work, we focus on a subset of combinatorial auctions termed as multi-unit auctions in which multiple identical goods are sold to a group of buyers.\footnote{For a comprehensive survey on combinatorial auctions, we refer readers to several excellent works by \citet{de2003combinatorial, blumrosen2007combinatorial,palacios2022combinatorial}.} These auctions find widespread application in various practical scenarios, including Treasury auctions \citep{nyborg2002bidder, hortaccsu2010mechanism, https://doi.org/10.3982/ECTA8365}, procurement auctions \citep{cramton2006dynamic}, electricity markets \citep{tierney2008uniform, fabra2006designing}, and emissions permit auctions \citep{goulder2013carbon, schmalensee2017lessons, goldner2020reducing}. 

While several works have focused on studying the equilibria properties in these auctions~\citep{ausubel2014demand, engelbrecht1998multi, noussair1995equilibria, markakis2015uniform, de2013inefficiency}, computing equilibrium strategies is usually intractable in general multi-unit auctions due to multi-dimensional valuations~\citep{kasberger2025bidding}. As a result, a growing body of work has emphasized on designing optimal bidding strategies in a prior-free setting, i.e., without any assumptions of the competing bidders' bids or valuations--in a repeated setting~\citep{galgana2023learning, potfer2024improved, branzei2023learning} or under a minimax loss framework~\citep{kasberger2025bidding}. Our work contributes to this literature by learning optimal strategies for value-maximizing buyers in repeated uniform price auctions.


\textbf{Data-Driven Methods in Auctions.} 
Traditionally, the problems of bidding in auctions and determining parameters such as reserve prices have been studied in the Bayesian framework~\citep{myerson1981optimal,riley1981optimal, HR09, beyhaghi2021improved}. However, these approaches rely on full knowledge of bidders’ valuation distributions, which is often unavailable or unreliable in practice. Moreover, they become computationally intractable in multi-unit auctions, as discussed earlier. In practice, historical bidding data is crucial for bidders to calibrate their strategies. For instance, \citet{kasberger2024robust} note that none of the auction consultants surveyed could offer viable solutions without access to sufficient real-world data. These challenges, along with the increasing availability of auction data, have led to a growing interest in data-driven approaches. Beyond learning optimal bidding strategies in a repeated setting, as discussed in the previous section, this line of work includes data-driven reserve price optimization~\citep{roughgarden2019minimizing, derakhshan2022linear, derakhshan2021beating, golrezaei2021dynamic, feng2021reserve}, boost value optimization in second-price auctions~\citep{golrezaei2021boosted} and deep learning methods for auction design~\citep{feng2018deep, dutting2024optimal, wang2024gemnet, wang2025bundleflow, liu2025interpretable}. 


\section{Model}\label{sec:model}
\textbf{Preliminaries.} There are $n$ buyers~(bidders) indexed by $i\in[n]$, and $K$ identical units of a single good. Each bidder $i$ has a fixed, private valuation curve, denoted by $\mathbf{v}_i\in\R^K_{+}$ that has diminishing marginal returns property, i.e., $v_{i, 1}\geq v_{i, 2}\geq\dots\geq v_{i, K}$ which is standard in literature~\citep{branzei2023learning, goldner2020reducing}. As we study optimal bidding strategies from the perspective of a single bidder, we drop the index $i$ when the context is clear. The maximum total demand for bidder $i$, denoted by $\maxbid\in[ K]$, is defined as $\min\{j\in [K-1]: v_{j+1}= 0 \}$. If such an index does not exist, we set $\maxbid$ as $K$. Hence, without loss of generality, we assume $\mathbf{v}\in\R_{+}^{\maxbid}$. For each $\mathbf{v}=[v_{1}, \dots, v_{\maxbid}]$, we define the \textit{average cumulative valuation} vector as $\mathbf{w}=[w_{1}, \dots, w_{\maxbid}]$, where
\begin{align}\label{eq:w}
w_{j}&=\frac{1}{j}\sum_{\ell\leq j}v_{\ell},\forall j\in[\maxbid]\,.
\end{align}
As $v_{1}\geq \dots\geq v_{\maxbid}$, we also have $w_{1}\geq \dots\geq  w_{\maxbid}$. 

\subsection{Auction Format and Bidders' Behavior}\label{ssec:bidder-behaviour}


\textbf{Allocation and Payment Rule.} In a uniform price auction, each bidder $i$ submits a sorted bid vector $\ibid$ using $\numbid$-uniform bidding language~(see \cref{def:m-unif-bid}). The auctioneer collects the bids~(entries of the bid vector) from all the bidders, sorts them in non-increasing order, and allocates units to the bidders with the top $K$ bids~(also termed as `winning' bids). That is, if bidder $i$ has $j$ bids in the top $K$ positions, they are allocated $j$ units. For ease of exposition, we assume there are no ties~(or ties are always broken in the favor of the bidder in consideration). See \cref{apx:sec:ties} for our discussion on handling ties. We assume the auction follows the last-accepted-bid payment rule~(bidders pay the $K^{th}$ highest bid per unit), which is widely used for uniform price auctions in practice~\citep{eu-ets-regulations, garbade2005treasury, burkett2020uniform}.

Let \( \otherbid{} \) denote the vector of bids submitted by all bidders except bidder~\( i \). Let \( \textsc{top}_K(\otherbid{}) \subseteq \otherbid{} \) be the multiset of the top \( K \) highest bids among these competing bids. Furthermore, we define \( \boldsymbol{\beta}_{-}^{-(j)} \) to be the \( j^{\text{th}} \) smallest element in \( \textsc{top}_K(\otherbid{}) \). Suppose bidder \(i\) submits \(\ibid\), such that the bid profile is \(\allbids := (\ibid; \otherbid{})\). Let \(x(\allbids)\) and \(p(\allbids)\) denote the number of units allocated to bidder \(i\) and the clearing price (i.e., the per-unit price which is the $K^{th}$ highest submitted bid), respectively. The total value obtained by the bidder is
$
\val(\allbids) = \sum_{j \leq x(\allbids)} v_{j},
$
while the total payment made is
$
\price(\allbids) = p(\allbids) \cdot x(\allbids).$




\textbf{Bidding Language.} Multi-unit auctions allocate a large number of identical units, requiring efficient ways for the bidders to express preferences. A common approach is \textit{standard bidding}, where bidders submit a vector of bids, one for each unit~\citep{branzei2023learning, galgana2023learning, babaioff2023making, birmpas2019tight, potfer2024improved}. Although expressive, this becomes computationally impractical when the number of units, \(K\), is large, as in EU ETS emission permit auctions and Treasury auctions. To address this, we consider a bidding language called \(\numbid\)-\textit{uniform bidding}, for any \(\numbid \in \mathbb{N}\).\footnote{This format generalizes the \textit{uniform bidding} format~\citep{de2013inefficiency, birmpas2019tight} and aligns with practical languages like those in product-mix auctions~\citep{klemperer2009new} and piecewise-linear bidding~\citep{nisan2015algorithmic}.} In \(\numbid\)-uniform bidding format, bidders submit $\numbid$ bid-quantity pairs \((b_i, q_i)\), where \(b_i\) is the bid value per unit and \(q_i\) is the quantity demanded:

\begin{definition}[$\numbid$-Uniform Bidding]\label{def:m-unif-bid}
For a fixed $\numbid\in\mathbb{N}$, a $\numbid$-uniform bidding strategy is characterized by $\numbid$ bid-quantity pairs, denoted as \[\ibid:=\langle(b_1, q_1),\dots, (b_\numbid, q_\numbid)\rangle\,,\] where $b_1>b_2>\dots>b_\numbid>0$ and $q_j> 0$, $j\in [m]$. This $\numbid$-uniform bidding strategy can be equivalently expressed as a vector~(similar to the standard bidding format) in which the first $q_1$ bids are $b_1$, followed by $q_2$ bids of $b_2$, and so on.
\end{definition}
We define $\ibid[1:\ell]=\langle(b_1, q_1),\dots, (b_\ell, q_\ell)\rangle$, for all $\ell<\numbid$, to represent the first $\ell$ bid-quantity pairs within a $\numbid$-uniform bidding strategy $\ibid=\langle(b_1, q_1),\dots, (b_\numbid, q_\numbid)\rangle$. We further define \[Q_j=\sum_{\l=1}^jq_{\ell},  \qquad\text{ for all $j\in[\numbid]$}\]  as the total quantity demanded in the first $j$ bid-quantity pairs, with $Q_0 = 0$, and we assume, without loss of generality that $Q_\numbid\leq \maxbid$.\footnote{Suppose the bidder bids for, and wins more than $\maxbid$ units. There is no additional value being allocated over $\maxbid$ units, but the total payment increases~(assuming the clearing price is positive), potentially violating the RoI constraint.} If $\numbid=\maxbid$, the bidding format is equivalent to standard bidding, but in practice, bidders often submit only a few bid-quantity pairs. For instance, in the 2023 EU ETS auctions, bidders submitted $\sim4.35$ pairs per auction on average~\citep{eex-eua-primary-auction}. 

\textbf{Bidders' Behavior.} The bidders maximize their total value obtained while adhering to a constraint that ensures the total value obtained in an auction is at least a constant multiple of the payment in that auction. This can be equivalently expressed as a return-on-investment (RoI) constraint:
\begin{align}\label{def:RoI}
    \val(\ibid; \otherbid{}) \geq (1+\gamma)\price(\ibid; \otherbid{})\,.
\end{align}
Here, $\gamma$ is defined as the \textit{target RoI} which is private and fixed. Without loss of generality, we assume $\gamma=0$~(or equivalently, the valuation curve $\mathbf{v}$ is scaled by $\frac{1}{1+\gamma}$) for the rest of this work. For $\gamma=0$, the RoI constraint implies the value obtained in an auction is at least the payment. 



\begin{example}\label{ex:mwe}
   Consider an auction with $n=2$ bidders, and $K=5$ identical units. The valuations are: $\mathbf{v}_1=[6, 4, 3, 1, 1]$ and $\mathbf{v}_2=[5, 3, 1, 1, 0]$. Both the bidders are value maximizing buyers with target RoI, $\gamma_1=\gamma_2=0$. Suppose $m=2$ and the bids submitted by the bidders are $\ibid_1=\langle(5, 2), (3, 3)\rangle$ and $\ibid_2=\langle(4, 2), (2, 2)\rangle$. The bid profile: $\allbids=[5, 5, 4, 4, 3, 3, 3, 2, 2]$ and top $K=5$ winning bids are $[\underline{5}, \underline{5}, 4, 4, \underline{3}]$. Bidder 1 is allocated $3$ units as they have $3$ bids~(underlined) in the winning bids, and bidder 2 gets the remaining $2$ units. The clearing price $p(\allbids)=3$,  $V_1(\allbids)=6+4+3= 13, V_2(\allbids)=5+3=8, P_1(\allbids)=3\cdot 3 =9$, and $P_2(\allbids)=2\cdot 3 =6$. The RoI constraint is satisfied for both the bidders as $13> 9$ and $8>6$. 
\end{example}

\subsection{Learning to Bid in Repeated Settings}

In practice, most multi-unit auctions, such as emission permit auctions and Treasury auctions, are conducted in a repeated setting. Formally, the auction described in the previous section takes place sequentially over $T$ rounds indexed by $t\in[T]$. Throughout this work, we assume that $T$ is known to the bidder.\footnote{In case $T$ is unknown, the bidder can use the standard \textit{doubling trick}~\citep{auer2002nonstochastic, cesa2006prediction} while designing learning algorithms.} We assume that the valuation vector $\v$ is fixed over the $T$ rounds which is a standard assumption in the literature~\citep{branzei2023learning, galgana2023learning, potfer2024improved}. Nonetheless, we discuss the case of time-varying valuation vectors in \cref{ssec:time-varying} and show that the techniques developed for the scenario when the valuation vector is fixed can be leveraged to learn to bid when the valuation vectors change across rounds.

We now extend the notations from the previous section to the repeated setting. Formally, in this setting, $\otherbid{t}$ denotes the bids submitted by all bidders except bidder $i$ in round $t$ and $\ordotherbid{t}{j}$ is the $j^{th}$ smallest among the top $K$ competing bids in round $t$. In round \(t\), if bidder \(i\) submits a bid \(\ibid^t\), the bid profile is \(\allbids^t := (\ibid^t; \otherbid{t})\). Let \(x(\allbids^t)\) and \(p(\allbids^t)\) denote the number of units allocated to bidder \(i\) and the clearing price, respectively, in round \(t\). For the $t^{th}$ auction, the total value obtained by the bidder is
$
\val(\allbids^t) = \sum_{j \leq x(\allbids^t)} v_{j},
$
while the total payment made is
$
\price(\allbids^t) = p(\allbids^t) \cdot x(\allbids^t).$

\textbf{Nature of Competing Bids.} We model the competing bids in each round as being chosen adversarially to capture realistic scenarios where the bidder has partial or no information about their competitors' strategy. We assume that the competing bids are selected by an \textit{oblivious adversary}, i.e., the bids in round~$t$ are chosen adversarially but independently of the bidder’s strategies in earlier rounds $\tau = 1, \dots, t-1$. This is equivalent to the case where all competing bids are fixed in advance, prior to the first round, and without any knowledge of the bidder’s realized strategies. In our setting, an oblivious adversary models a sufficiently large market where the bids submitted by a single bidder do not materially change the allocation, for which the competing bids could be viewed as independent of the strategies submitted by the bidder. In \cref{ssec:non-oblivious}, we extend our analysis for an \textit{adaptive adversary},\footnote{Also called as \textit{non-oblivious} adversary~\citep{cesa2006prediction} or \textit{reactive} adversary~\citep{maillard2011adaptive}.} where the competing bids in any round may depend on the bidder’s strategies in the previous rounds.

\textbf{RoI Constraints in Repeated Setting.} In the repeated setting, we require the bidders to satisfy the RoI constraint described in \cref{def:RoI} in every round. A bidding strategy $\ibid$ is called \textit{feasible} for a sequence of competing bids $[\otherbid{t}]_{t\in[T]}$, if the RoI constraint is satisfied for every round $t\in [T]$. 


\begin{remark} [RoI Constraints]
Our notion of RoI constraints in the repeated setting aligns with the definitions in \citet{wilkens2016mechanism, wilkens2017gsp, lv2023utility}. Similar constraints~(up to scaling factors) are considered by \citet{lucier2023autobidders} and \citet{gaitonde2023budget}, which they term as \textit{maximum bid} constraints.\footnote{An earlier version of \citet{lucier2023autobidders} termed this as \textit{marginal RoI (or value)} constraints.} Unlike the aggregate constraints over \(T\) rounds typically assumed in online ad auction literature~\citep{deng2021towards,deng2022fairness, feng2023online, deng2023multi}, we enforce RoI constraints for each auction individually. This distinction reflects the fundamental differences between the two settings: ad auctions often occur simultaneously and frequently, with values in the order of cents, whereas Treasury and emission permit auctions occur sequentially over longer horizons, with units valued in millions. Furthermore, we consider a prior-free setting that allows for strategic~(possibly adversarial) behavior by the competitors. Hence, bidders in such auctions prioritize profitability in each auction rather than waiting for an indefinite period, leading to per-round RoI constraints.\footnote{Although EU ETS emission permit auctions are scheduled to occur regularly, regulations stipulate that an auction may be canceled if the bidders' demand falls short of the supply of permits or if the clearing price of the auction does not meet the reserve price~\citep{eu-ets-regulations}. Hence, the bidders are more likely to ensure RoI feasibility in each round.} Nonetheless, in \cref{ssec:cumulative-roi}, we present a heuristic for the setting where the RoI constraints are enforced only in aggregate over a sliding window of $T_0 \geq 1$ rounds.
\end{remark}

\subsection{Objective and Performance Metric}
We consider an online setting where a bidder seeks to maximize cumulative value while satisfying RoI constraints and managing uncertainty about competing bids. In this setting, bidders privately submit their bids, and the auctioneer allocates items and sets prices based on these bids. In round \(t\), the learning algorithm maps the feedback information set from the first \(t-1\) rounds, denoted by \(\mathcal{I}^{t-1}\), to a bid \(\ibid^t\), where this mapping may be deterministic or random. Under full information feedback, the information set 
$\mathcal{I}^{t-1} = (\allbids^1, \allbids^2, \ldots, \allbids^{t-1}),$
which includes all bids from previous rounds. In the bandit setting, the information set $\mathcal{I}^{t-1} = (x(\allbids^1), \price(\allbids^1), \ldots, x(\allbids^{t-1}), \price(\allbids^{t-1}))$ includes only the allocations and prices of the previous rounds.

We design learning algorithms to minimize \textit{regret}, the difference in the value obtained by a fixed hindsight-optimal strategy chosen by a clairvoyant with \textit{a priori} knowledge of competing bids and that by the learner over time. Formally,
\begin{align}\label{eq:regret-def}
    \textsf{REG}=\max_{\ibid\in\generic}\sum_{t=1}^T\val(\ibid; \otherbid{t}) - \sum_{t=1}^T\E[\val(\ibid^t; \otherbid{t})],
\end{align}
where the expectation is with respect to any randomness in the learning algorithm. Here,  $\generic$ is the class of bidding strategies~(formally characterized in \cref{ssec:uf-strategies}). We require the bidding strategies in $\generic$ to be RoI feasible for the sequence of competing bids $[\otherbid{t}]_{t\in[T]}$. In \cref{sec:learning-safe}, we consider the case when both the clairvoyant and the learner choose strategies from the same class and aim to obtain sublinear regret. Later, we consider more challenging settings where the clairvoyant can choose the optimal strategy from much richer classes of strategies compared to the learner. In these cases, we obtain sublinear \textit{$\alpha$-approximate} regret, where $\alpha\in(0, 1]$ is the richness ratio~(see \cref{sec:learning-rich}).

\section{Safe Bidding Strategies}\label{ssec:uf-strategies}
Recall that strategies chosen by the learner must be non-anticipating (mapping the history to a bidding strategy) while remaining RoI feasible for all rounds \(t \in [T]\), even under adversarially generated competing bids. This creates an obvious challenge: a bidding strategy that satisfies RoI feasibility in previous rounds may become infeasible even under the slightest change in the competing bids e.g., if $\mathbf{v}=[0.9, 0.5, 0.1]$, the bidding strategy $(0.6, 3)$ is RoI feasible for competing bids $\otherbid{1}=[0.61, 0.59, 0.59]$ but not for $\otherbid{2}=[0.59, 0.59, 0.59]$. To address this, we focus on \textit{safe bidding strategies} that inherently guarantee RoI feasibility, regardless of the adversarial behavior of competing bids:

\begin{definition}[Safe Strategies] \label{def:uni_feasible}
A $\numbid$-uniform  bidding strategy, $\ibid=\langle(b_1, q_1),\dots, (b_\numbid, q_\numbid)\rangle$, is called a \textit{safe} strategy if it is feasible irrespective of the competing bids, i.e.,
\begin{align*}
    \val(\ibid; \otherbid{}) \geq \price({\ibid}; \otherbid{}), \quad \forall \otherbid{}\,.
\end{align*}
The class of all $\numbid$-uniform safe bidding strategies is denoted as $\oneufclass{\numbid}$. The union of classes of safe bidding strategies with at most $\numbid$ bid-quantity pairs is denoted as $\ufclass{\numbid}=\bigcup_{k\in[\numbid]}\oneufclass{k}$.
\end{definition}
Recall that we assume $\gamma=0$ throughout; otherwise the constraint in \cref{def:uni_feasible} is identical to \cref{def:RoI}.
\subsection{Characterizing Safe Bidding Strategies}
We begin by defining underbidding~(overbidding) under the $\numbid$-uniform bidding format in the given context. 
\begin{definition}[Underbid and Overbid]\label{def:under-over_bid}
    A $\numbid$-uniform strategy $\ibid=\langle(b_1, q_1),\dots, (b_\numbid, q_\numbid)\rangle$ is an

    (a) \textit{underbidding} strategy if $b_j \leq w_{Q_j}, \forall j\in[\numbid]$ and $\exists \l\in[\numbid]$ such that $b_\l < w_{Q_\l}$, and

    (b) \textit{overbidding} strategy if $\exists \l\in[\numbid]$ such that $b_\l > w_{Q_\l}$,

    where we recall that $Q_j=\sum_{\l\leq j}q_\l, \forall j\in[\numbid]$.
\end{definition}

\begin{figure}[!tbh]
    \centering
    \includegraphics[width=0.7\linewidth]{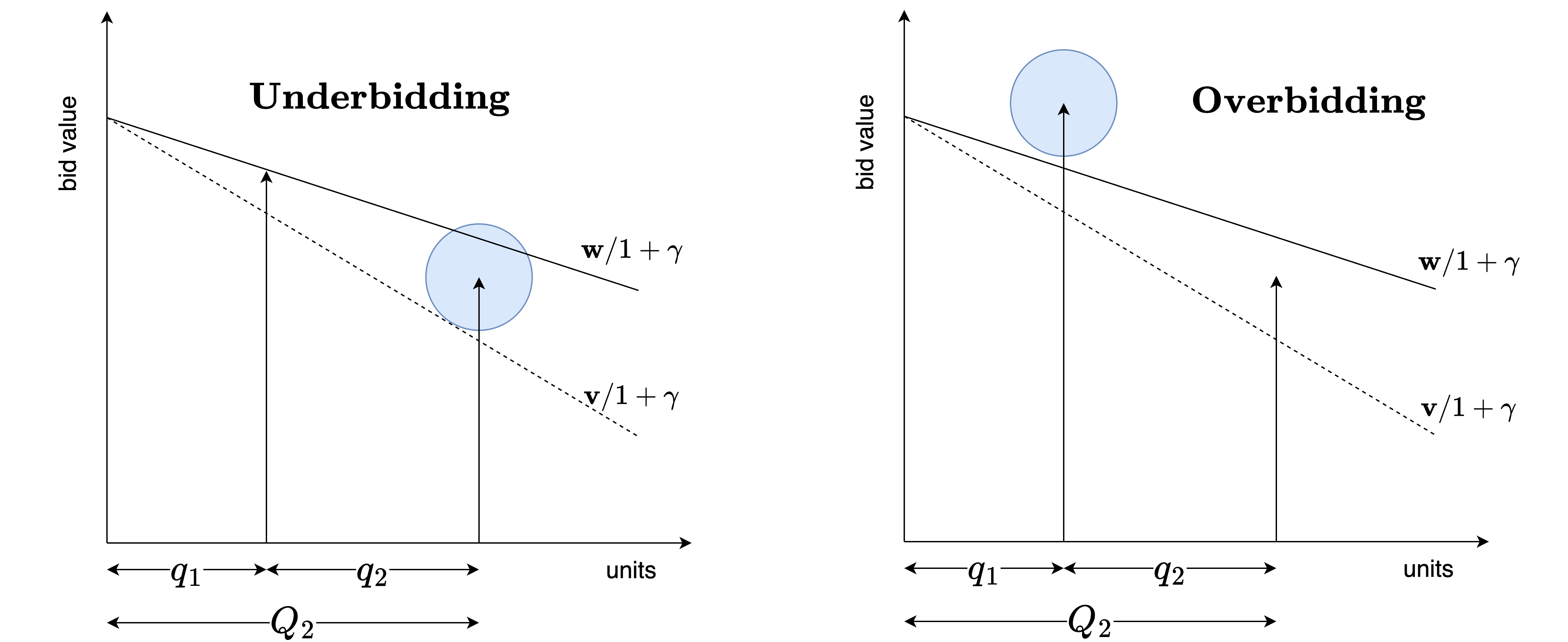}
    \caption{The solid line represents the average cumulative valuation curve, $\mathbf{w}$, and the dotted line represents the valuation curve, $\v$. The figure in the left~(resp. right) illustrates underbidding~(resp. overbidding) for a $2$-uniform bidding strategy. Note that the notions of underbidding and overbidding in \cref{def:under-over_bid} are defined with respect to $\mathbf{w}$ and \textit{not} $\v$.\protect\footnotemark~Here, the plots of $\v$ and $\mathbf{w}$ are shown to be linear for illustrative purposes only.}
    \label{fig:ub-and-ob}
\end{figure}
\footnotetext{In \cref{def:under-over_bid} (and throughout the paper), we assume $\gamma = 0$ without loss of generality for ease of exposition. \cref{fig:ub-and-ob} illustrates how different values of $\gamma$ impacts the results. Specifically, for any $\gamma \geq 0$, underbidding and overbidding are defined with respect to the scaled average cumulative valuation curve, $\frac{\mathbf{w}}{1+\gamma}$.}

In other words, after normalizing $\gamma=0$, $\ibid=\langle(b_1, q_1),\dots, (b_\numbid, q_\numbid)\rangle$ is an underbidding strategy if, for all $j\in [\numbid]$, the bid $b_j$ is at most the average of the first $Q_j$ entries of the valuation vector, denoted as $w_{Q_j}$, and there exists $\ell \in [\numbid]$ where the inequality is strict. Similarly, $\ibid$ is an overbidding strategy if there exists some $\l\in[\numbid]$ such that $b_\l$ is strictly greater than the average of the first $Q_\l$ entries of the valuation vector. Having defined the notions of overbidding and underbidding, we characterize the class of safe bidding strategies, $\oneufclass{\numbid}$:

\begin{theorem} \label{thm:P-m}
For any $\numbid\in\N$, no overbidding is allowed in $\oneufclass{\numbid}$. So, the collection of all $\numbid$-uniform safe strategies is
\begin{align*}
\oneufclass{\numbid} =\Big\{\ibid=\langle(b_1, q_1), \dots, (b_\numbid, q_\numbid)\rangle: b_\ell \leq w_{Q_\l}, \forall \ell\in[\numbid] \Big\}\,,
\end{align*}
where $\bf w$ is defined in Eq. \eqref{eq:w}.

\end{theorem}

Since there are infinitely many safe strategies in $\oneufclass{\numbid}$, we focus on the subset of safe \textit{undominated} strategies, $\optoneufclass{\numbid}$, obtained by eliminating \textit{very weakly dominated} strategies. A safe strategy $\ibid$ is said to be very weakly dominated if there exists another safe strategy $\ibid'$ such that, for any competing bid $\otherbid{}$, $\val(\ibid, \otherbid{}) \leq \val(\ibid', \otherbid{})$~\citep[pp. 79]{shoham2008multiagent}. Keeping these strategies does not improve the bidder’s performance, but removing them yields a \textit{finite} strategy class~(as shown below), which significantly simplifies the design of online learning algorithms. We show that for every underbidding strategy (as per \cref{def:under-over_bid}) in $\oneufclass{\numbid}$, there exists a non-underbidding strategy in the same class that very weakly dominates it, allowing the former to be safely eliminated. Hence,
\begin{theorem}\label{thm:opt-bid}
The class of $\numbid$-uniform safe undominated bidding strategies, $\optoneufclass{\numbid}$, is given as follows:
\begin{align*}
    \optoneufclass{\numbid} =\Big\{\ibid=\langle(b_1, q_1), \dots, (b_\numbid, q_\numbid)\rangle: b_\l = w_{Q_\l}, \l\in[\numbid] \Big\}\,. 
\end{align*}
The union of classes of safe undominated strategies with at most $\numbid$ bid-quantity pairs is denoted by $\optufclass{\numbid}=\bigcup_{k\in[\numbid]}\optoneufclass{k}$. 
\end{theorem}

 Observe that $\optoneufclass{\numbid}$ is a finite class of bidding strategies that depend only the valuation curve and is \textit{independent} of the competing bids. The strategies in $\optoneufclass{\numbid}$ exhibit a ``nested'' structure, where the $j^{th}$ highest bid ($b_j$) is the average of the first $Q_j =\sum_{\ell=1}^j q_\ell$ entries of the valuation curve (illustrated in \cref{fig:nested}). Importantly, fixing $Q_j$'s uniquely determines the bidding strategy, as $b_j = w_{Q_j}$ for strategies in $\optoneufclass{\numbid}$, which will be crucial in learning the optimal safe bidding strategy, as discussed in the following section.
\section{Learning Safe Bidding Strategies}\label{sec:learning-safe}
Here, the clairvoyant and the learner choose strategies from the class of safe bidding strategies with at most $\numbid$ bid-quantity pairs, $\optufclass{\numbid}$ (defined in \Cref{thm:opt-bid}). Thus, in this case, the regret defined in \cref{eq:regret-def} becomes
\begin{align}\label{eq:safe-regret}
\textsf{REG}=\max_{\ibid\in\optufclass{\numbid}}\sum_{t=1}^T\val(\ibid; \otherbid{t}) - \sum_{t=1}^T\E[\val(\ibid^t; \otherbid{t})]\,. 
\end{align}
To design learning algorithms, we first consider the offline setting, where  we aim to solve the following problem for a given bid history, $ [\otherbid{t}]_{t\in[T]}$. 
\begin{align}\label{eq:opt-safe}
    \max_{\ibid\in\optufclass{\numbid}}\sum_{t=1}^T\val(\ibid; \otherbid{t})\tag{$\textsc{Offline}$}
\end{align}
Since $\optufclass{\numbid}$ contains $O(\maxbid^\numbid)$ strategies, evaluating each strategy individually is intractable. To overcome this, we reduce the offline problem to finding the maximum weight path in an appropriate edge-weighted directed acyclic graph (DAG), which allows for the computation of the optimal offline strategy in poly($\numbid, \maxbid$) time.



\subsection{Offline Setting} \label{sec:offline}
Fix a bid history, $\hist=[\otherbid{t}]_{t\in[T]}$, and consider the Problem \eqref{eq:opt-safe}. The following lemma shows that the objective function can be decomposed across the bid-quantity pairs. Formally,

\begin{lemma}\label{lem:decomp}
Problem \eqref{eq:opt-safe} can be formulated as 
\begin{align*}
  \max_{\l\in[\numbid]}\max_{Q_1, \dots, Q_\l} \sum_{j=1}^\l \sum_{t=1}^T \sum_{k=Q_{j-1}+1}^{Q_j}v_k\cdot\ind{w_{Q_j}\geq \ordotherbid{t}{k}}\,.
\end{align*}
\end{lemma}

\begin{proof}
    By \cref{thm:opt-bid}, a bidding strategy $\ibid=\langle (b_1, q_1), \dots, (b_\l, q_\l)\rangle\in\optufclass{\numbid}$ for some $\l\in[\numbid]$ can be uniquely identified by $\{Q_1, \dots, Q_\l\}$ where $Q_j=\sum_{k\leq j}q_k$ because $b_j=w_{Q_j}$ for all $j\in[\l]$. Let the bid history $\hist=[\otherbid{t}]_{t\in[T]}$. Observe that if the bidder wins $r$ units in any auction, then the $r^{th}$ highest bid value in its submitted strategy must be at least the $r^{th}$ smallest bid in the top $K$ competing bids, i.e.,  $\ordotherbid{t}{r}$. Hence, for any round $t$, 
\begin{align*}
    \val(\ibid; \otherbid{t}) &=\sum_{j=1}^\l \sum_{k=Q_{j-1}+1}^{Q_j}v_k\cdot\ind{w_{Q_j}\geq \ordotherbid{t}{k}},
\end{align*}
which implies that the problem \eqref{eq:opt-safe} can be equivalently expressed as
\begin{align*}
    \max_{\ibid\in\optufclass{\numbid}}\sum_{t=1}^T\val(\ibid; \otherbid{t}) = \max_{\l\in[\numbid]}\max_{Q_1, \dots, Q_\l} \sum_{j=1}^\l \sum_{t=1}^T \sum_{k=Q_{j-1}+1}^{Q_j}v_k\cdot\ind{w_{Q_j}\geq \ordotherbid{t}{k}}\,.
\end{align*}
\end{proof}

Building on the decomposition in \Cref{lem:decomp}, we now construct an edge-weighted DAG $\mathcal{G}(V, E)$.

\textbf{Vertices.} The DAG has a `layered' structure with a source node~($s$), destination node~($d$) and $\numbid$ intermediate layers. The $\l^{th}$ intermediate layer contains $\maxbid+1-\l$ nodes, denoted by $(\l, j)$ where $\l\in[\numbid], j\in\{\l, \l+1, \dots, \maxbid\}$ as shown in \cref{fig:dag}. For convenience, we set $s=(0, 0)$.

\textbf{Edges and Edge Weights.} A directed edge exists from vertex \(x\) to \(y\) under the following conditions:

(i) \(x=(\l-1, j)\) and \(y=(\l, j')\), \(\forall \l\in[\numbid]\) and \(j < j'\). This edge connects consecutive layers in the graph, and its weight is given by:
\begin{align}\label{eq:dag-base-weight}
    \textsf{w}(e) = \sum_{t=1}^T\sum_{k=j+1}^{j'}v_k\cdot\ind{w_{j'}\geq \ordotherbid{t}{k}}\,.
\end{align}

(ii) \(x=(\l, j)\) and \(y=d\), \(\forall \l \in [\numbid], j \in [\maxbid]\). This edge connects the current layer to the destination node \(d\), and its weight is \(\textsf{w}(e) = 0\).



\begin{figure}[!tbh]
    \centering
    \includegraphics[width=0.35\linewidth]{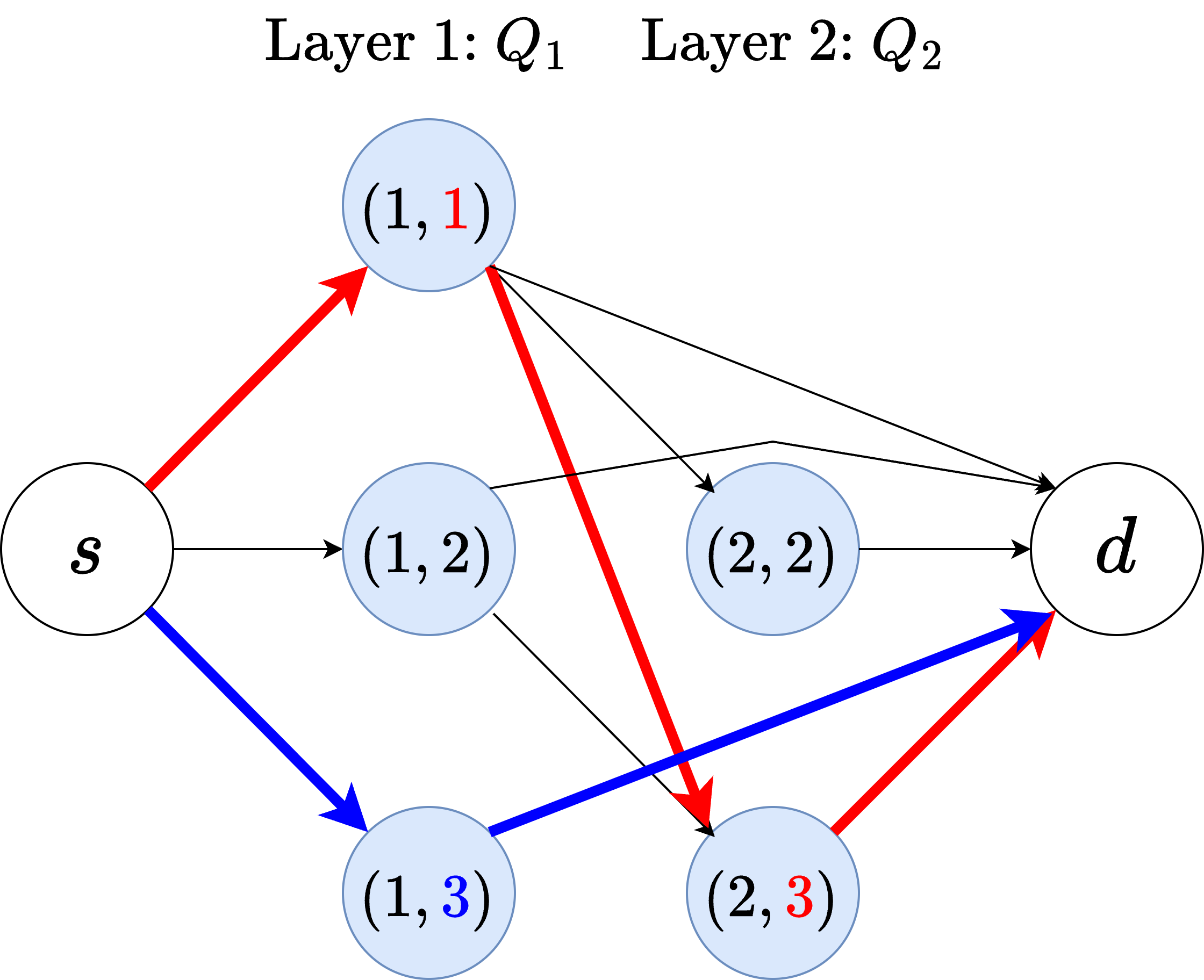}
    \caption{In this DAG, $\maxbid=3$, $\numbid=2$. The red path refers to $\ibid=\langle (w_1, 1), (w_3, 2)\rangle\in\optufclass{2}.$ The values in red are the corresponding $Q_j$'s. Similarly, the blue path refers to $\ibid=(w_3, 3)\in\optufclass{2}$.}
    \label{fig:dag}
\end{figure}

\begin{theorem}\label{thm:DAG-base-strategy}

There exists a bijective mapping between the $s$-$d$ paths in $\mathcal{G}(V, E)$ and strategies in $\optufclass{\numbid}$, i.e., the path $s\to (1, z_1)\to\dots\to(k, z_k)\to d$ for $k\in[\numbid]$ refers to the strategy $\ibid=\langle(b_1, q_1), \dots, (b_k, q_k)\rangle$ where 
\begin{align*}
  b_\l=w_{z_\l} \quad\text{and}\quad q_\l=z_\l-z_{\l-1}, \forall  \l\in [k] \,. 
\end{align*}
where $w_{z_\l}$ is defined as per \cref{eq:w} and $z_0=0$. Conversely, the strategy $\ibid=\langle(w_{Q_1}, q_1), \dots, (w_{Q_k}, q_k)\rangle$ maps to the path $s\to (1, Q_1)\to\dots\to(k, Q_k)\to d$. 

The value obtained by a bidding strategy is the weight of the corresponding $s$-$d$ path.  Thus, the Problem \eqref{eq:opt-safe} is equivalent to finding the maximum weight $s$-$d$ path in $\mathcal{G}(V, E)$ which can be computed in $\text{poly}(\numbid, \maxbid)$ time.

\end{theorem}

\subsection{Online Setting}
In the online setting, in each round \(t \in [T]\), the learner submits a strategy \(\ibid^t \in \optufclass{\numbid}\) and receives feedback from the auction. We consider two feedback models: (a) full information setting and (b) bandit setting. Leveraging the DAG formulation, we design a polynomial-time online learning algorithm such that \(\textsf{REG} = \widetilde{O}(\text{poly}(\numbid, \maxbid)\cdot\sqrt{T})\). For the rest of the work, we assume \(v_j \in [0, 1], \forall j \in [\maxbid],\) without loss of generality.



In \cref{thm:DAG-base-strategy}, we established a bijection mapping between $s$–$d$ paths in $\mathcal{G}(V, E)$ and safe strategies in $\optufclass{\numbid}$. While naïvely applying the Hedge algorithm~\citep{freund1997decision} in the full-information setting by treating each path as an expert yields optimal regret bounds, it is computationally intractable due to the exponential number of paths: $O(\maxbid^\numbid)$. The situation further worsens in the bandit setting, as implementing the EXP3 algorithm~\citep{auer2002nonstochastic}, treating each path as an expert, also yields unsatisfactory regret bounds~(recall that EXP3 obtains $O(\sqrt{NT\log N})$ regret where $N$ is the number of arms and per-round rewards belong to $[0, 1]$. In this setting, $N=O(\maxbid^\numbid)$.). 

Our proposed algorithm retains the core idea of treating paths as experts but does so efficiently. Leveraging the decomposition in \cref{lem:decomp}, it maintains probabilities over edges rather than paths, enabling efficient path sampling as a Markov chain without explicitly maintaining a distribution over all the paths. Specifically, in every round \(t\), a DAG \(\mathcal{G}^t(V, E)\), as described earlier, is constructed. The proposed algorithm consists of three main steps: (1) $\textsf{UPDATE}$, (2) $\textsf{SAMPLE}$, and (3) $\textsf{MAP}$. In the \textsf{UPDATE} step, the bidder uses the feedback from the previous round to update the edge probabilities, $\varphi^t(u\to v)$, for every edge $e=u\to v$ in the DAG. Using the dynamic programming-based weight-pushing method from \cite{takimoto2003path}, this step can be performed in $\text{poly}(\numbid, \maxbid)$ time. In the \textsf{SAMPLE} step, the algorithm samples an $s$–$d$ path in a Markovian manner: starting at the source node $s$, it repeatedly samples the next node $v$ from the set of neighbors of the current node $u$ according to probabilities $\varphi^t(u \to v)$, until the destination node $d$ is reached. Once a $s$-$d$ path is sampled, in the \textsf{MAP} step, the bidder maps it to a safe strategy $\ibid^t\in\optufclass{\numbid}$ and submits $\ibid^t$ in round $t$. After receiving feedback at the end of the round, in the full information setting, the edge weights are set as follows: 

(i) If \(x=(\l-1, j)\) and \(y=(\l, j')\) with \(\l\in[\numbid]\) and \(j < j'\),
\begin{align}\label{eq:full-info-edge-weight}
    \textsf{w}^t(e) = \sum_{k=j+1}^{j'}v_k\cdot\ind{w_{j'}\geq \ordotherbid{t}{k}}\,.
\end{align}

(ii) If \(x=(\l, j)\) and \(y=d\), then \(\textsf{w}^t(e) = 0, \forall \l \in [\numbid], j \in [\maxbid]\).


The proposed algorithm is formally presented in \cref{alg:weight-pushing}. The main result of this section is:

\begin{algorithm}[!tbh]
\caption{Learning Safe Bidding Strategies~(Full Information)}
\label{alg:weight-pushing}
\small{
\begin{algorithmic}[1]
\Require valuation curve $\mathbf{v}$, time horizon $T$, learning rate $\eta\in(0, \frac{1}{\maxbid}]$, $\textsf{w}^0(e)=0,\text{and}~\varphi^0(e)=1, \forall e\in E$.
\For{$t = 1, 2, \dots, T$}
    \State Construct $\mathcal{G}^t(V, E)$ similar to $\mathcal{G}(V, E)$ without weights.
    \State $\textsf{UPDATE}:$ Set $\Gamma^{t-1}(d)=1$ and recursively compute in bottom-to-top fashion for every node $u$ in $\mathcal{G}^t(V, E)$:
    \begin{align*}
        \Gamma^{t-1}(u)=\sum_{v:u\to v=e\ni E}\Gamma^{t-1}(v)\cdot\varphi^{t-1}(e)\cdot\exp(\eta \textsf{w}^{t-1}(e))\,.
    \end{align*}

    \State For edge $e=u\to v$ in $\mathcal{G}^t(V, E)$, update edge probability:
       \begin{align*}
           \varphi^t(e)=\varphi^{t-1}(e)\cdot\exp(\eta \textsf{w}^{t-1}(e))\cdot\frac{\Gamma^{t-1}(v)}{\Gamma^{t-1}(u)}.
       \end{align*}
    
    

    \State $\textsf{SAMPLE}$: Define initial node $u=s$ and path $\path^t=s$. 
    \While{$u\neq d$}\label{line:sample-start}
    \State Sample $v$ with probability $\varphi^t(u\to v)$.
    \State Append $v$ to the path $\path^t$; set $u\gets v$.
    \EndWhile\label{line:sample-end}
    \State $\textsf{MAP}$: If $\path^t=s\to (1, z_1)\to\dots\to(k, z_k)\to d$ for some $k\in[\numbid]$, submit $\ibid^t=\langle(b_1, q_1), \dots, (b_k, q_k)\rangle$ where
    \begin{align*}
      b_\l=w_{z_\l} \quad\text{and}\quad q_\l=z_\l-z_{\l-1}, ~\forall  \l\in [k]\,.
    \end{align*}
    
    where $w_{z_\l}$ is defined as per \cref{eq:w} for all $\l\in[k]$ and $z_0=0$.
    \State The bidder observes $\otherbid{t}$ and sets $\textsf{w}^t(e)$ as per \cref{eq:full-info-edge-weight}.
\EndFor
\end{algorithmic}}
\end{algorithm}

\begin{theorem}\label{thm:full-info}
 In the full information setting, \cref{alg:weight-pushing} achieves \(\textsf{REG} \leq O(\maxbid\sqrt{\numbid T\log \maxbid})\) in \(\text{poly}(\numbid, \maxbid)\) space and time per round.
\end{theorem}

In the bandit setting, as the bidder receives only $x(\allbids^t)$ and $\price(\allbids^t)$ as feedback in round $t$, they can not compute \cref{eq:full-info-edge-weight} for all the edges. Hence, we use an unbiased estimator, $\widehat{\textsf{w}}^t(e)$, of $\textsf{w}^t(e)$ as follows: 
\begin{align}\label{eq:w-hat-estimate-exp3}
    \widehat{\textsf{w}}^t(e) &= \overline{\textsf{w}}(e)-\frac{\overline{\textsf{w}}(e)-\textsf{w}^t(e)}{p^t(e)}\cdot\ind{e\in\path^t}\quad\text{where}\quad p^t(e)=\sum_{\path:e\in\path}\mathbb
    P^t(\path)\,.
\end{align}
Here, $p^t(e)$ is the probability of selecting edge $e$ in round $t$, $\mathbb P^t$ is the distribution over all $s$-$d$ paths in $\mathcal{G}^t(V, E)$. Observe that $\widehat{\textsf{w}}^t(e)$ is well defined as $p^t(e)>0, \forall e\in E, \forall t\in[T]$. Here, $\overline{\textsf{w}}(e)$ is an edge-dependent upper bound of the edge weight $\textsf{w}^t(e)$. Formally, for $\l\in[\numbid]$,
\begin{align*}
    \overline{\textsf{w}}(e)=\begin{cases}
        j'-j, &\text{if } e=(\l-1, j)\to(\l, j')\\
        0, &\text{if } e=(\l, j)\to d\,.
    \end{cases}
\end{align*}

The algorithm followed by the bidder in the bandit setting is identical to \cref{alg:weight-pushing}, with each instance of $\textsf{w}^t(e)$ replaced by $\widehat{\textsf{w}}^t(e), \forall t\in[T], \forall e\in E$. Furthermore, define $\widehat{\textsf{w}}^0(e)=0, \forall e\in E$. Clearly, $\E[\widehat{\textsf{w}}^t(e)]=\textsf{w}^t(e)$. The unbiased estimator $\widehat{\textsf{w}}^t(e)$ is different from the natural importance-weighted estimator $\frac{\textsf{w}^t(e)}{p^t(e)}\cdot\ind{e\in\path^t}$ for technical reasons~(similar to \citet[Eq. 11.6]{lattimore2020bandit}) which we will clarify in the regret analysis~(see details in \cref{apx:thm:bandit}). Hence,

\begin{theorem}\label{thm:bandit}
   In the bandit setting, \cref{alg:weight-pushing} with the aforementioned changes achieves \(\textsf{REG} \leq O(\maxbid^2\sqrt{\numbid^3 T\log \maxbid})\) in \(\text{poly}(\numbid, \maxbid)\) space and time per round.
\end{theorem}


Note that the per-round reward $\val(\ibid^t; \otherbid{t})\in[0, \maxbid]$ and number of experts is $N = O(\maxbid^\numbid)$. In the bandit setting, although we compute unbiased estimators of the edge weights~(equivalently path weights) and perform Hedge-like updates efficiently, the resulting regret bound is tighter than that of a naïve EXP3 implementation. This improvement stems from two key advantages:
(a) after each round, the bidder can observe the rewards of the constituent edges of the submitted path, rather than just the aggregate reward,\footnote{Strictly speaking, this corresponds to the \textit{semi-bandit feedback }model in online combinatorial optimization~\citep{neu2013efficient, jourdan2021efficient, brandt2022finding}, whereas the full bandit model reveals only the total reward. However, in our setting, the full bandit model lacks a natural interpretation. For simplicity, we refer to this partial feedback setting as the bandit feedback model.} and
(b) the structure of the combinatorial action space reveals partial information about unselected paths as well. We complement \cref{thm:full-info} and \cref{thm:bandit} by establishing the following lower bound:

\begin{theorem}\label{thm:regret-LB}
Suppose $\maxbid\geq2$ and $\numbid=1$. Then, there exist competing bids, $[\otherbid{t}]_{t\in[T]}$, such that, under any learning algorithm, $\E[\textsf{REG}]=\Omega(\maxbid\sqrt{T})$ in the full information setting. This implies that the lower bound also holds under bandit feedback. 
\end{theorem}

\section{Competing Against Richer Classes of Bidding Strategies}\label{sec:learning-rich}
In previous sections, we characterized the class of safe bidding strategies and proposed a learning algorithm that obtains sublinear regret, where the regret is computed against a clairvoyant that also selects the optimal safe bidding strategy.  In this section, we consider the cases where the clairvoyant can choose the optimal strategy from richer bidding classes, which we will describe shortly. We show that the bidder who follows safe bidding strategies and runs \cref{alg:weight-pushing} in the full information setting achieves sublinear (approximate) regret when competing against these stronger benchmarks. A similar result also holds in the bandit setting, \textit{mutatis mutandis}. The performance metric we consider in this section is
\begin{align}\label{eq:alpha-approx-regret}
    \alpha\text{-}\textsf{REG}_{\nature}=\alpha\cdot\max_{\ibid\in\nature}\sum_{t=1}^T\val(\ibid; \otherbid{t}) - \sum_{t=1}^T\E[\val(\ibid^t; \otherbid{t})]\,.
\end{align}
Here, $\nature$ is the class of bidding strategies from which the clairvoyant (which serves as our benchmark) chooses the optimal bidding strategy whereas the bidder~(learner) submits bids, $\ibid^t\in\optufclass{\numbid}$, in each round unless stated otherwise. Here, $\alpha\in(0, 1]$ is defined as the \textit{richness ratio}.

\begin{remark}[Richness Ratio $\alpha$]  
In prior works on approximate regret, the parameter \(\alpha\) typically measures the hardness of the offline problem, such as \(\max_{\ibid \in \nature} \sum_{t=1}^T \val(\ibid; \otherbid{t})\) (see \citet{streeter2008online, niazadeh2022online}). In contrast, our work reinterprets \(\alpha\) as \emph{richness ratio}, capturing the relative richness of the clairvoyant's strategy class compared to the learner's. Specifically, it quantifies how closely the value achieved by the optimal safe strategy in \(\optufclass{\numbid}\) approximates that by the optimal strategy in \(\nature\). If \(\nature = \optufclass{\numbid}\), then \(\alpha = 1\)~(see \cref{sec:learning-safe}). By considering a richer class of strategies for the clairvoyant, we aim to quantify the robustness and performance of safe strategies against stronger benchmarks.
\end{remark}
Our main contribution in this section is to compute the richness ratio $\alpha$ for different classes of bidding strategies. To this end, for any $\nature$, define:
\begin{align}\label{eq:def-lambda}
   \Lambda_{\nature, \optufclass{\numbid}}&:=\max_{\v, \hist}\Lambda_{\nature, \optufclass{\numbid}}(\hist, \v), \text{where} \nonumber\\
\Lambda_{\nature, \optufclass{\numbid}}(\hist, \v)&=\frac{\max_{\ibid\in\nature}\sum_{t=1}^T\val(\ibid; \otherbid{t})}{\max_{\ibid'\in\optufclass{\numbid}}\sum_{t=1}^T\val(\ibid'; \otherbid{t})}
    \end{align}
    for $\hist=[\otherbid{t}]_{t\in[T]}$. In words, $\Lambda_{\nature, \optufclass{\numbid}}$ is the maximum ratio of the value obtained by the optimal bidding strategy in \(\nature\) and that by the optimal strategy in $\optufclass{\numbid}$ over all valuation curves and all bid histories of arbitrary length. 
    

\begin{definition}[Richness Ratio]\label{prop:choose-alpha}
Suppose the clairvoyant chooses the optimal strategy from \(\nature\) (which \emph{may depend on the bid history}, \(\hist\)) such that \(\Lambda_{\nature, \optufclass{\numbid}} \leq \lambda\), and the bound is tight, i.e., there exist a bid history \(\hist\), a valuation curve \(\v\), and \(\overline{\delta} > 0\) such that \(\Lambda_{\nature, \optufclass{\numbid}}(\hist, \v) \geq \lambda - \delta\) for any \(\delta \in (0, \overline{\delta}]\). Then, the richness ratio of $\nature$ is \(\alpha = \sfrac{1}{\lambda}\).

\end{definition}
The key idea is that if, in the worst-case scenario for the offline problem, the optimal safe strategy achieves no more than a \(\frac{1}{\lambda}\)-fraction of the value obtained by the optimal bidding strategy in \(\nature\), then the learner can achieve at most a \(\frac{1}{\lambda}\)-fraction of that value in the online setting as well. In the subsequent sections, we consider different classes of bidding strategies for the clairvoyant and compute its richness ratio $\alpha$.
An immediate corollary of \cref{prop:choose-alpha} is:
\begin{corollary}\label{thm:alpha-UB}
  Suppose the clairvoyant chooses the optimal strategy from the class $\nature$ such that $\Lambda_{\nature, \optufclass{\numbid}}\leq \lambda$ and this bound is tight implying the richness ratio of $\nature$ is $\alpha=\frac{1}{\lambda}$. Then, \cref{alg:weight-pushing} obtains $\alpha$-$\textsf{REG}_{\nature}\le O(\maxbid\sqrt{\numbid T\log \maxbid})$ in the full information setting.
\end{corollary}
\begin{proof}
    As $\Lambda_{\nature, \optufclass{\numbid}}\leq \lambda$, 
\begin{align*}
    \alpha\text{-}\textsf{REG}_{\nature}&=\frac{1}{\lambda}\cdot\max_{\ibid\in\nature}\sum_{t=1}^T\val(\ibid; \otherbid{t}) - \sum_{t=1}^T\E[\val(\ibid^t; \otherbid{t})]\\
    &\leq \max_{\ibid\in\optufclass{\numbid}}\sum_{t=1}^T\val(\ibid; \otherbid{t}) - \sum_{t=1}^T\E[\val(\ibid^t; \otherbid{t})]=\textsf{REG}, 
\end{align*}
where the first equality holds as $\alpha=\frac{1}{\lambda}$. By \cref{thm:full-info}, we get the stated regret upper bound.
\end{proof}

Recall that we defined richness of a class of bidding strategies along two dimensions: (i) the safety of strategies (safe vs. only RoI-feasible), and (ii) expressiveness, i.e., the number of bid-quantity pairs allowed. In the subsequent sections, we consider classes of strategies that are richer along either one or both of the said dimensions and compute the corresponding richness ratios, $\alpha$.

\subsection{The Class of RoI-feasible Strategies with $\leq\numbid$ Bid-Quantity Pairs }
First, we consider the clairvoyant class \(\nature\) to be the class of strategies containing at most \(\numbid\) bid-quantity pairs, which are RoI-feasible only for the given bid history $\hist$, rather than for every possible sequence of competing bids as in the case of safe strategies. We denote this class by $\feasclass{\numbid}$. Then,

\begin{theorem}\label{thm:Price_universal}
    For any $\numbid\in\N$, $\Lambda_{\feasclass{\numbid}, \optufclass{\numbid}}\leq2$.
    Furthermore, there exist a bid history and valuation curve such that for any $\delta\in(0, \frac{1}{2}]$, $\Lambda_{\feasclass{\numbid}, \optufclass{\numbid}}(\hist, \v)\geq 2-\delta$. Thus, the richness ratio of the class $\feasclass{\numbid}$ is $\alpha=1/2$.
\end{theorem}
We present the proof of \cref{thm:Price_universal} in \cref{apx:thm:Price_universal}. {\cref{thm:Price_universal} implies that restricting the learner to safe bidding strategies has a bounded cost and does not lead to an arbitrary loss in the obtained value as the upper bound is \textit{independent of $\numbid$}. Building on this, \cref{thm:alpha-UB} shows that  \cref{alg:weight-pushing} results in $\frac{1}{2}$-approximate sublinear regret in the online setting when the safe strategies compete against $\feasclass{\numbid}$. Additionally, the factor-of-two loss  represents a worst-case scenario, occurring in a highly non-trivial setting (see \cref{apx:LB-2-m-gen}). In practice, we expect these strategies to perform near-optimally, a claim further supported by the numerical simulations presented in \cref{sec:sims}.}

\subsection{Classes of Strategies with $\leq \numbid'$ Bid-Quantity Pairs For $\numbid'\geq \numbid$}
We now consider the setting where the clairvoyant selects the optimal strategy from the class of safe strategies with at most $\numbid’$ bid-quantity pairs, where $\numbid’ \geq \numbid$. This strategy class, $\optufclass{\numbid’}$, is richer than the bidder’s class $\optufclass{\numbid}$, although, along the expressiveness dimension compared to the class $\feasclass{\numbid}$ discussed in the previous section, which is richer along the safety dimension. Then,

\begin{theorem}\label{thm:m-mbar}
    For any $\numbid, \numbid'\in\N$ such that $\numbid'\geq \numbid$, $\Lambda_{\optufclass{\numbid'}, \optufclass{\numbid}}\leq\frac{\numbid'}{\numbid}$. Additionally, there exist $\hist$ and $\v$ such that for any $\delta\in(0, \frac{1}{2}]$, 
    \begin{align*}
        \Lambda_{\optufclass{\numbid'}, \optufclass{\numbid}}(\hist, \v)\geq\frac{\numbid'}{\numbid}-\delta\,.
    \end{align*}
    Thus, the richness ratio of the class $\optufclass{\numbid'}$ is $\alpha=\numbid/\numbid'$.
\end{theorem}


Finally, suppose the clairvoyant can now select the optimal strategy from \(\feasclass{\numbid'}\), the class of bidding strategies with at most \(\numbid'\) bid-quantity pairs (\(\numbid' \geq \numbid\)) that are only RoI-feasible for the bid history \(\hist = [\otherbid{t}]_{t \in [T]}\). This strategy class is richer than the bidder's strategy class along both dimensions. Then,
\begin{theorem}\label{thm:mbar-non-safe}
For any $\numbid, \numbid'\in\N$ such that $\numbid'\geq \numbid$, $\Lambda_{\feasclass{\numbid'}, \optufclass{\numbid}}\leq\frac{2\numbid'}{\numbid}$. Moreover, there exist $\hist$ and $\v$ such that for any $\delta\in(0, \frac{1}{2}]$, we have
\begin{align*}
    \Lambda_{\feasclass{\numbid'}, \optufclass{\numbid}}(\hist, \v)\geq\frac{2\numbid'}{\numbid}-\delta\,.
\end{align*}
Thus, the richness ratio of the class $\feasclass{\numbid'}$ is $\alpha=\numbid/2\numbid'$.
\end{theorem}
By \cref{thm:alpha-UB}, for $\numbid'\geq \numbid$,  \cref{alg:weight-pushing} obtains $\frac{\numbid}{\numbid'}$-approximate~(resp. $\frac{\numbid}{2\numbid'}$-approximate) sublinear regret~(same as \cref{thm:full-info}) in the online setting when the clairvoyant draws the optimal strategy from the class of safe~(resp. RoI-feasible) strategies with at most $\numbid'$ bid-quantity pairs.

The richness ratio in \cref{thm:m-mbar} and \cref{thm:mbar-non-safe} can theoretically become arbitrarily small as \(\numbid'\) increases, \textit{for a fixed \(\numbid\)}. However, the bidder retains the flexibility to choose \(\numbid\): increasing \(\numbid\), though, incurs higher space and time complexity for \cref{alg:weight-pushing}. Moreover, for any given pair \((\numbid, \numbid')\), the bounds are tight for carefully constructed non-trivial hard instances (see details in \crefrange{apx:LB-m}{apx:LB-2m}). In practice, as shown in \cref{sec:sims}, the empirical richness ratios are substantially better than the worst-case bounds, suggesting that such pathological scenarios are unlikely to arise in real-world settings.

We conclude the section with two key observations. First, when $\numbid = \numbid’$, \cref{thm:m-mbar} implies $\alpha = 1$, which corresponds to the setting where both the bidder and the clairvoyant follow the class of safe strategies with at most $\numbid$ bid-quantity pairs, the case considered in \cref{sec:learning-safe}. Similarly, when $\numbid = \numbid’$, \cref{thm:mbar-non-safe} yields $\alpha = \frac{1}{2}$, recovering the result of \cref{thm:Price_universal}. Second, recall that we defined the richness of a strategy class along two dimensions: safety and expressiveness. The results in this section imply that when a benchmark strategy class is richer along both dimensions~(\cref{thm:mbar-non-safe}), its overall richness ratio is the product of the richness ratios along each individual dimension~(\cref{thm:Price_universal} and \cref{thm:m-mbar}).

\subsection{Computing Empirical Richness Ratio: A Case Study}\label{sec:sims}
As discussed in the previous section, constructing worst-case instances for showing tight bounds are significantly intricate that may not reflect real world settings. This motivates us to empirically estimate how well do the class of safe strategies approximate the richer strategy classes in practice. To this end, with slight abuse of notation, we define the \textit{empirical richness ratio} as 
\begin{align*}
    \alpha_{\nature, \optufclass{\numbid}}(\hist, \v)= \frac{1}{\Lambda_{\nature, \optufclass{\numbid}}(\hist, \v)},
\end{align*}
where $\Lambda_{\nature, \optufclass{\numbid}}(\hist, \v)$ is defined in \cref{eq:def-lambda}. We assess the tightness of the theoretical bounds in practice by computing $\alpha_{\nature, \optufclass{\numbid}}(\hist, \v)$ using the EU ETS emission permit auction data from 2022 and 2023. However, only aggregate statistics of the submitted bids is publicly available for privacy reasons. Hence, we synthesize individual level bid data from these available statistics. The exact procedure to reconstruct the bid data is presented in \cref{apx:data}.




\begin{figure}[ht]
\begin{minipage}[b]{0.45\linewidth}
\centering
\includegraphics[width=\textwidth]{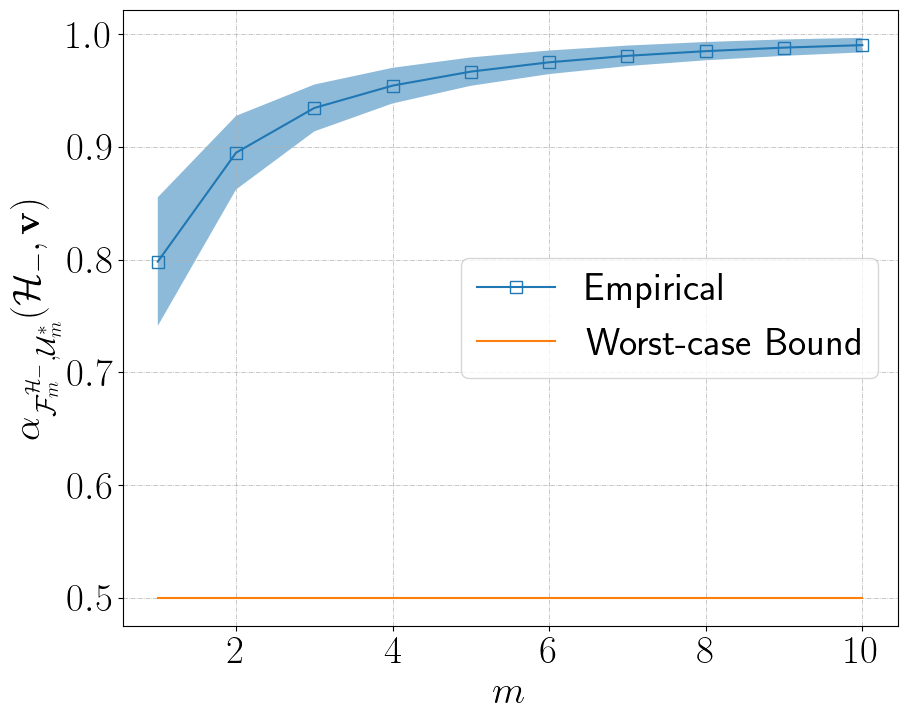}
\label{fig:figure1}
\end{minipage}
\hspace{0.5cm}
\begin{minipage}[b]{0.45\linewidth}
\centering
\includegraphics[width=\textwidth]{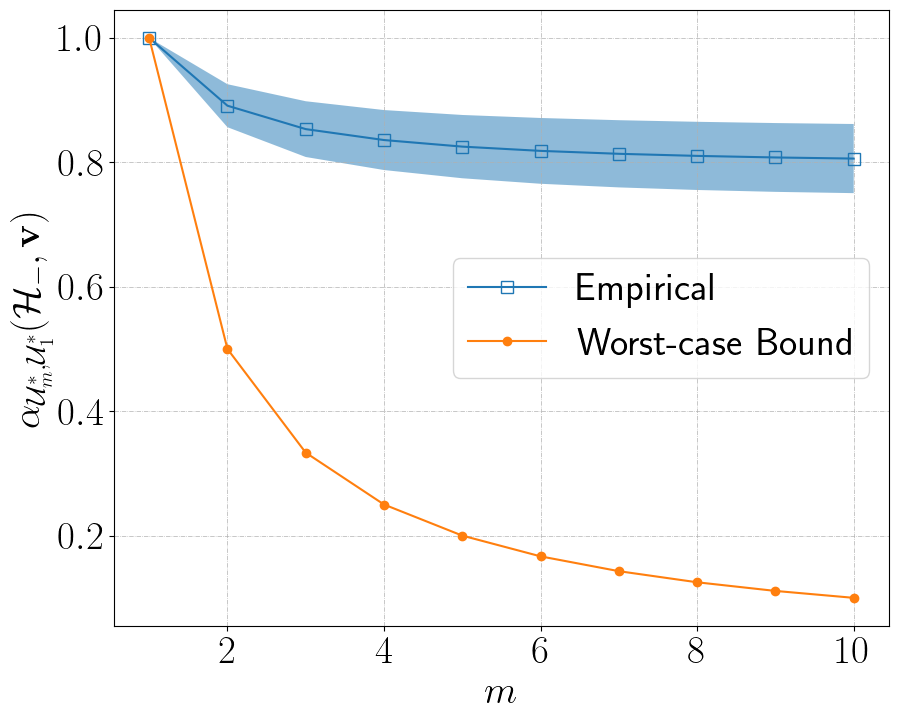}
\label{fig:figure2}
\end{minipage}
\vspace{-0.5cm}
\caption{The left~(resp. right) figure refers to (the lower bound on) $\alpha_{\feasclass{\numbid},\optufclass{\numbid}}(\hist, \v)$~(resp. $\alpha_{\optufclass{\numbid}, \optufclass{1}}(\hist, \v)$) as a function of $\numbid$. The shaded area corresponds to one standard deviation.}
\label{fig:sims}
\end{figure}
\textbf{Estimating $\alpha_{\feasclass{\numbid},\optufclass{\numbid}}(\hist, \v)$.} Since computing the exact value attained by the optimal strategy in $\feasclass{\numbid}$ is non-trivial, we instead compute a uniform upper bound that is \textit{independent of $\numbid$} (see \cref{apx:ilp}). The left plot of \cref{fig:sims} shows that the empirical lower bound of $\alpha_{\feasclass{\numbid},\optufclass{\numbid}}(\hist, \v)$ is substantially better than the theoretical worst-case bound of $1/2$. In particular, for $\numbid \geq 4$, even the lower bound exceeds $0.95$, indicating that safe bidding strategies are nearly optimal in practice.

\textbf{Estimating $\alpha_{\optufclass{\numbid},\optufclass{1}}(\hist, \v)$.} As shown in the right plot of \cref{fig:sims}, the empirical values of $\alpha_{\optufclass{\numbid},\optufclass{1}}(\hist, \v)$ significantly exceed the worst-case theoretical bound of $1/\numbid$ (\cref{thm:m-mbar}). The marginal gain from increasing the number of bid-quantity pairs plateaus beyond $\numbid = 4$; even at $\numbid = 10$, the ratio of the value achieved by the optimal 1-uniform safe strategy to that of the optimal safe strategy with at most 10 bid-quantity pairs is $\sim0.8$, in contrast to the worst-case bound of $0.1$. 

Both these results suggest that empirical richness ratios are substantially better than the worst-case guarantees, implying that safe strategies with small $\numbid$ offer near-optimal performance in realistic settings.

\subsection{Proof of \cref{thm:Price_universal}}\label{apx:thm:Price_universal}

We present the proof of \cref{thm:Price_universal} below. While \cref{thm:m-mbar} follows a similar high-level approach, it requires different techniques, which we detail in \cref{apx:thm:m-mbar}. Finally, \cref{thm:mbar-non-safe} follows directly from the proofs of \cref{thm:Price_universal} and \cref{thm:m-mbar}; see \cref{apx:thm:mbar-non-safe} for details.


To prove that $\Lambda_{\feasclass{\numbid}, \optufclass{\numbid}}\leq2$ for any $\numbid\in\N$, consider the metric presented in \cref{eq:def-lambda} for any bid history, $\hist=[\otherbid{t}]_{t\in[T]}$,
\begin{align}\label{eq:pouf}
    \Lambda_{\feasclass{\numbid}, \optufclass{\numbid}}(\hist, \v) = \frac{\max_{\ibid\in\feasclass{\numbid}}\sum_{t=1}^T\val(\ibid; \otherbid{t})}{\max_{\ibid\in\optufclass{\numbid}}\sum_{t=1}^T\val(\ibid; \otherbid{t})}\,.
\end{align}
To prove the upper bound in \cref{thm:Price_universal}, we show that for any $\hist$ and $\v$, $\Lambda_{\feasclass{\numbid}, \optufclass{\numbid}}(\hist, \v)\leq 2$. Maximizing over all bid histories and valuation vectors, we get the desired result.

Let $\optbid{\numbid}=\langle(b_1^*, q_1^*), \dots, (b_\numbid^*, q_\numbid^*)\rangle$ where $\optbid{\numbid}:=\argmax_{\ibid\in\feasclass{\numbid}} \sum_{t=1}^T\val(\ibid; \otherbid{t})$.\footnote{We assumed $\optbid{\numbid}=\langle(b_1^*, q_1^*), \dots, (b_\numbid^*, q_\numbid^*)\rangle$ without loss of generality. The proof also follows if we considered $\optbid{\numbid}=\langle(b_1^*, q_1^*), \dots, (b_k^*, q_k^*)\rangle$ for some $k\in[\numbid]$ instead.} Define  $Q_\l^*=\sum_{j \leq \l}q_j^*$. Suppose $\optbid{\numbid}$ is allocated $r_t^*$ units in round $t$. For any $j\in[\numbid]$, let $T_j$ be the set of rounds in which the least winning bid is $b_j^*$, i.e.,
\begin{align}
\label{eq:Tj}T_j= \big\{t\in [T]: Q_{j-1}^*<r_t^*\leq Q_j^* \big\}\,.
\end{align}
For any $j\in[\numbid]$, partition $T_j$ into $T_{j, 0}$ and $T_{j, 1}$ such that $T_{j, 0}$ is the set of rounds where the bidder gets \emph{strictly less than} $Q_j^*$ units and $T_{j, 1}$ is the set of rounds when they get \emph{exactly} $Q_j^*$ units:
\begin{align}\label{eq:Tj01}
T_{j, 0}= \{t\in T_j: r_t^*< Q_j^* \}, \qquad  T_{j, 1}= \{t\in T_j: r_t^* = Q_j^* \}\,.
\end{align}
For any $j\in [m]$, let 
\begin{align*}
    \widehat{Q}_j=\begin{cases}
       \max\{r_t^*: t\in T_{j, 0}\}, ~&\text{ if } T_{j, 0}\neq \emptyset\\
       Q_j^*~&\text{ if } T_{j, 0}= \emptyset\,.
    \end{cases}
\end{align*}
In other words, $\widehat{Q}_j$ is the $2^{nd}$ highest number of units allocated to $\optbid{\numbid}$ over the rounds in $T_j$~(or the highest if $T_{j, 0}=\emptyset$ or $T_{j, 1}=\emptyset$). 

\textbf{Step 1. Constructing a Restricted Class of Safe Bidding Strategies.} 
Here, as a crucial part of the proof, we construct a restricted class of safe bidding strategies, denoted by $\optufclass{\numbid}(\hist)$, where $\optufclass{\numbid}(\hist) \subset \optufclass{\numbid}$. This construction serves two purposes. First, it reduces the search space for the optimal safe bidding strategy. Second, and more importantly, it enables us to establish a connection between the optimal safe bidding strategy and $\optbid{\numbid}$.

In defining the restricted class, we use the quantities, $\{\widehat{Q}_j, Q_j^*\}_{j\in [m]}$ as follows:
\begin{align}\label{eq:restricted-safe}
 \optufclass{\numbid}(\hist) = \Big\{\ibid=\langle(b_1, q_1), \dots, (b_\numbid, q_\numbid)\rangle: b_\l =  w_{Q_\l} ,~~ Q_\l\in\{\widehat{Q}_\l, Q_\l^*\}, \quad \forall \l\in[\numbid] \Big\}\,.
\end{align}
Recall that any strategy in $\optufclass{\numbid}$ with $\numbid$ bid-quantity pairs takes the form of $\ibid=\langle(b_1, q_1), \dots, (b_\numbid, q_\numbid)\rangle: b_\l = w_{Q_\l}, \forall \l\in[\numbid]$. The strategies in $\optufclass{\numbid}(\hist)$ also have the same structure, but as an important difference, for any $\ell\in [\numbid]$, we enforce $Q_\l\in\{\widehat{Q}_\l, Q_\l^*\}$. {Observe that for any $\l\in[\numbid-1]$, 
\begin{align*}
    Q_{\l}\leq Q_\l^*< \widehat{Q}_{\l+1}\leq Q_{\l+1},
\end{align*}
where the first and third inequalities follow directly from the definition of $Q_\l$ and the second one follows from the definition of $\widehat{Q}_{\l+1}$. So, $Q_\l$'s are distinct and ordered.} Further observe that the number of bidding strategies in $ \optufclass{\numbid}(\hist) $ is $O(2^{\numbid})$; significantly smaller than the number of strategies in $\optufclass{\numbid}$, which is $O(\maxbid^\numbid)$. 

With the definition of the restricted class of safe bidding strategies, we have 
\begin{align}\label{eq:pouf-1}
    \Lambda_{\feasclass{\numbid}, \optufclass{\numbid}}(\hist, \v) =\frac{\max_{\ibid\in\feasclass{\numbid}}\sum_{t=1}^T\val(\ibid; \otherbid{t})}{\max_{\ibid\in \optufclass{\numbid}}\sum_{t=1}^T\val(\ibid; \otherbid{t})}\leq \frac{\max_{\ibid\in\feasclass{\numbid}}\sum_{t=1}^T\val(\ibid; \otherbid{t})}{\max_{\ibid\in\optufclass{\numbid}(\hist)}\sum_{t=1}^T\val(\ibid; \otherbid{t})}\,.
\end{align}

\textbf{Step 2. Value Decomposition for $\numbid$-uniform Strategies.}  We now present a crucial result that expresses the value obtained by a $\numbid$-uniform strategy as a function of the value obtained by $\numbid$ 1-uniform strategies. 
\begin{lemma} \label{lem:sum-to-max-equivalence-UF} Let $\ibid=\langle(w_{Q_1}, q_1),\dots, (w_{Q_\numbid}, q_\numbid)\rangle$ be a $\numbid$-uniform safe bidding strategy for some $\numbid\in\N$. Then, for any competing bids $\otherbid{}$,
    \begin{align*}  \val(\ibid;\otherbid{})=\max_{\l\in[\numbid]}\val((w_{Q_\l}, Q_\l);\otherbid{}), 
    \end{align*}
    where we recall that $Q_\l=\sum_{j=1}^\l q_j, \forall \l\in[\numbid]$.
\end{lemma}
A key consequence of \cref{lem:sum-to-max-equivalence-UF} is that knowing the value obtained by the strategies in $\optufclass{1}$ is sufficient to compute the value obtained \textit{any} $\numbid$-uniform safe bidding strategy. 

For any $\ibid\in\optufclass{\numbid}(\hist)$, where $\ibid=\langle(b_1, q_1), \dots, (b_\numbid, q_\numbid)\rangle: b_\l = w_{Q_\l}$, and $Q_{\l} \in \{Q_\l^*, \widehat Q_\l\}$  for any $\l\in[\numbid]$, partition the set $[\numbid]$ into $\mathcal{S}^*$ and $\widehat{\mathcal{S}}$ as follows:
\begin{align*}
    \mathcal{S}^*=\{j:Q_j=Q_j^*, j\in[\numbid]\} \quad\text{ and } \quad\widehat{\mathcal{S}}=\{j:Q_j=\widehat{Q}_j< Q_j^*, j\in[\numbid]\}\,.
\end{align*}

Hence, by \cref{lem:sum-to-max-equivalence-UF}, for any round $t\in T_j$, and $\ibid\in\optufclass{\numbid}(\hist)$, we get
\begin{align}
    \val(\ibid; \otherbid{t}) &= \max_{\ell \in[\numbid]}\val((w_{Q_{\l}}, Q_{\l}); \otherbid{t}) \geq \val((w_{Q_j}, Q_j); \otherbid{t})\,,\label{eq:sum-max-thm1-proof}\\
    \implies \sum_{t=1}^T \val(\ibid; \otherbid{t}) &\geq  \sum_{j\in \mathcal{S}^*}\sum_{t\in T_j} \val((w_{Q_j}, Q_j); \otherbid{t}) + \sum_{j\in \widehat{\mathcal{S}}}\sum_{t\in T_j} \val((w_{Q_j}, Q_j); \otherbid{t})\,.\label{eq:lem5-6-main}
\end{align}

\textbf{Step 3. Allocation Lower Bounds for 1-uniform Strategies.} Let $V_{j, k}$ be the time-average  value obtained by $\optbid{\numbid}$ in the set of rounds in $T_{j, k}$~(as defined in \cref{eq:Tj01}). Formally, $\forall j\in[\numbid], k\in\{0, 1\}$, 
\begin{align*}
    V_{j, k} = \inv{|T_{j, k}|}\sum_{t\in T_{j, k}} \val(\optbid{\numbid}, \otherbid{t})\,.
\end{align*}
Define $\psum{k}=\sum_{j=1}^k v_j$. Note that $V_{j, 1}=\psum{Q_j^*}$ because for any $t\in T_{j, 1}$, we have $r_t^*= Q_j^*$. 

For any $j\in[\numbid]$, we present the following lemma to establish lower bound on $\sum_{t\in T_j}\val((w_{Q_j}, Q_j); \otherbid{t})$. 

\begin{lemma} \label{lem:value_restricted}
Let $\ibid\in\optufclass{\numbid}(\hist)$, where $\ibid=\langle(b_1, q_1), \dots, (b_\numbid, q_\numbid)\rangle: b_j = w_{Q_j}$, and $Q_{j} \in \{Q_j^*, \widehat Q_j\}, \forall j\in[\numbid]$.  
Then, for any $j\in [\numbid]$,
\begin{itemize}
    \item[(a)] if $Q_j =  \widehat Q_j< Q_j^*$ (i.e., $j\in \widehat{\mathcal{S}}$), we have 
\begin{align}\label{eq:xj-is-qbar}
    \sum_{t\in T_j}\val((w_{Q_j}, Q_j); \otherbid{t})\geq V_{j, 0}|T_{j, 0}|+\psum{\widehat{Q}_j}|T_{j, 1}|\,.
\end{align}
\item[(b)] If $Q_j =  Q_j^*$ (i.e., $j\in  \mathcal{S}^*$), we have 
 \begin{align}\label{eq:xj-is-qstar}
    \sum_{t\in T_j}\val((w_{Q_j}, Q_j); \otherbid{t})\geq \psum{Q_j^*}|T_{j, 1}|\,.
\end{align}
\end{itemize}
\end{lemma}

Note that the right hand side of \cref{eq:lem5-6-main} depends only on the choice of the partitions $\mathcal{S}^*$ and $\widehat{\mathcal{S}}$. Substituting the lower bound from \cref{lem:value_restricted} in \cref{eq:lem5-6-main}, we establish that,
\begin{align}
    \Lambda_{\feasclass{\numbid}, \optufclass{\numbid}}(\hist, \v)&\stackrel{\eqref{eq:pouf-1}}{\leq} \frac{\max_{\ibid\in\feasclass{\numbid}}\sum_{t=1}^T\val(\ibid; \otherbid{t})}{\max_{\ibid\in\optufclass{\numbid}(\hist)}\sum_{t=1}^T\val(\ibid; \otherbid{t})}\nonumber\\
    &\leq \frac{\max_{\ibid\in\feasclass{\numbid}}\sum_{t=1}^T\val(\ibid; \otherbid{t})}{\max_{(\mathcal{S}^*, \widehat{\mathcal{S}})}\Big\{\sum_{j\in \mathcal{S}^*}\Big(\psum{Q_j^*}|T_{j, 1}|\Big)+\sum_{j\in\widehat{\mathcal{S}}}\Big(V_{j, 0}|T_{j, 0}| + \psum{\widehat{Q}_j}|T_{j, 1}|\Big)\Big\}}\,.\label{eq:pouf-helper-1}
\end{align}
We complete the proof by invoking the following result:
\begin{lemma}\label{lem:tedious-algebra}
   Let $(\mathcal{S}_0^*, \widehat{\mathcal{S}}_0)=\argmax_{(\mathcal{S}^*, \widehat{\mathcal{S}})}\Big\{\sum_{j\in \mathcal{S}^*}\Big(\psum{Q_j^*}|T_{j, 1}|\Big)+\sum_{j\in\widehat{\mathcal{S}}}\Big(V_{j, 0}|T_{j, 0}| + \psum{\widehat{Q}_j}|T_{j, 1}|\Big)\Big\}$. Then,
   \begin{align*}
       \frac{\max_{\ibid\in\feasclass{\numbid}}\sum_{t=1}^T\val(\ibid; \otherbid{t})}{\sum_{j\in \mathcal{S}_0^*}\Big(\psum{Q_j^*}|T_{j, 1}|\Big)+\sum_{j\in\widehat{\mathcal{S}}_0}\Big(V_{j, 0}|T_{j, 0}| + \psum{\widehat{Q}_j}|T_{j, 1}|\Big)} \leq 2\,.
   \end{align*}
\end{lemma}
From \cref{eq:pouf-helper-1} and \cref{lem:tedious-algebra}, we get $\Lambda_{\feasclass{\numbid}, \optufclass{\numbid}}(\hist, \v)\leq 2$ as desired.

\subsubsection{Tight Lower Bound for $\numbid=1$.} 
We complement the proof of the upper bound of \cref{thm:Price_universal} by constructing a tight lower bound instance for $\numbid = 1$. The case for general $\numbid \geq 2$ is deferred to \cref{apx:LB-2-m-gen}. While the key ideas behind the hard instance construction remain the same for both cases, the $\numbid \geq 2$ case is substantially more intricate and is therefore presented separately.


 Formally, for any $\delta\in(0, \frac{1}{2}]$, we construct a bid history, $\hist$, and valuation vector, $\mathbf{v}$ for which $\Lambda_{\feasclass{1}, \optufclass{1}}(\hist, \v)\geq 2-\delta$. Let $\maxbid=T= 2\ceil{\frac{1}{\delta}}$ and $K=\maxbid+1$. Suppose $\mathbf{v}=[1, v, \cdots, v]\in\R^{\maxbid}$, target RoI $\gamma=0$  and $v=1-4\epsilon$ where $\epsilon=\frac{\delta-1/\maxbid}{4(1-1/\maxbid)}< \frac{\delta}{4}\leq \frac{1}{8}$. 
 

 The bid history is defined as:
\begin{align*}
    \ordotherbid{t}{j}=\begin{cases}
        1-\epsilon, ~&\text{ if $t\leq \maxbid-1$ and $j=1$}\\
        C, ~&\text{ if $t\leq \maxbid-1$ and $2\leq j\leq K$}\\
        \epsilon, ~&\text{ if $t= \maxbid$ and $1\leq j\leq K$}\\
    \end{cases}\,,
\end{align*}
\normalsize
where $C\gg 1$. The bid history is presented in \cref{tab:qmaxq}.

\begin{table}[!tbh]
\centering
\caption{\small Bid history for tight lower bound for $\numbid=1$. Here, $C\gg1$ and $\epsilon>0$ is a small real number.}
\label{tab:qmaxq}
\small
    \begin{tabular}{ccccc}
    \toprule
         $t=1$&  $t=2$ &$\cdots$&  $t=\maxbid-1$&  $t=\maxbid$\\\midrule
         $C$&  $C$ &$\cdots$&  $C$ &  $\epsilon$\\
         $\vdots$&  $\vdots$ &$\cdots$&  $\vdots$ &  $\vdots$\\
         $C$&  $C$ &$\cdots$&  $C$ &  $\vdots$\\
 $1-\epsilon$& $1-\epsilon$ &$\cdots$& $1-\epsilon$ & $\epsilon$\\
 \bottomrule
    \end{tabular} 
\end{table}
\normalsize
Notice that the constructed $\hist$ contains $\maxbid-1$ rounds with high competition (where the bidder can acquire at most 1 unit), while there is one round with minimal competition, allowing the bidder to obtain any desired number of units. The 1-uniform safe bidding strategies of the form $(w_q, q)$ do not perform well universally on all rounds because with increase in $q$, despite the increasing demand (i.e., $q$), the bid value $w_q$ decreases, thereby reducing the likelihood of acquiring any units.


For a given $\hist$, recall that for any $\numbid\in\N$, $\optbid{\numbid}=\argmax_{\ibid\in\feasclass{\numbid}} \sum_{t=1}^T\val(\ibid; \otherbid{t})$. We define \[\optvalue{\numbid}=\sum_{t=1}^T\val(\optbid{\numbid}; \otherbid{t}),\] 

the value obtained by $\optbid{\numbid}$ over the bid history $\hist$. Similarly, define
\begin{align*}
\safebid{\numbid}=\argmax_{\ibid\in\optufclass{\numbid}} \sum_{t=1}^T\val(\ibid; \otherbid{t})~\quad\text{and}\quad \safevalue{\numbid}=\sum_{t=1}^T\val(\safebid{\numbid}; \otherbid{t})\,.
\end{align*}


\textbf{Computing $\optbid{1}$.} Observe that no strategy in $\feasclass{1}$ can obtain more than 1 unit in the first $\maxbid-1$ rounds and $\maxbid$ units in the final round. It is easy to verify that $\optbid{1}=(1, \maxbid)$ as it obtains the maximum number of units possible for the constructed $\hist$. So, $\optvalue{1}=\maxbid+v(\maxbid-1)$.


\textbf{Computing $\safebid{1}$.} We claim that $\safebid{1}= (w_1, 1)= (1,1)\implies \safevalue{1}=\maxbid$. As $\ordotherbid{t}{1}=1-\epsilon >1-2\epsilon=w_2$ for $t\in[\maxbid-1]$, the strategy $\left(w_{q}, q\right)$ for $2\leq q\leq \maxbid$ does not get any value in the first $\maxbid-1$ rounds. Bidding $\left(w_{q}, q\right)$ gets exactly $q$ units in round $\maxbid$ as $w_q>\epsilon$ for any $q\geq 2$. So, for $2\leq q\leq \maxbid$, the total value obtained by bidding $(w_q, q)$ is $1+(q-1)v<q$ which proves that $\safevalue{1}=\maxbid$. So, 
\begin{align*}
    \Lambda_{\feasclass{1}, \optufclass{1}}(\hist, \v)= \frac{\optvalue{1}}{\safevalue{1}}=\frac{\maxbid+v(\maxbid-1)}{\maxbid}= \frac{2\maxbid-1-4\epsilon(\maxbid-1)}{\maxbid}=2-\delta\,.
\end{align*}
\section{Extensions}\label{sec:extensions}
 In this section, we consider several important extensions of the model introduced in \cref{sec:model} by relaxing some of its assumptions. We first consider the setting where the competing bids are generated by an adaptive adversary instead of an oblivious adversary. Next, we present an efficient heuristic for the bidder to follow when RoI constraints are enforced over a sliding window of $T_0\in\N$ rounds. Finally, we examine the case where the valuation vectors vary across rounds—either stochastically or adversarially. In each case, we build upon the safe bidding strategies and learning algorithm introduced in \cref{sec:learning-safe}.

\subsection{Adaptive Adversary}\label{ssec:non-oblivious}
In many practical marketplaces, particularly in repeated auctions, the actions of a bidder may influence the behavior of competitors over time. Competing bidders can observe past auction outcomes, infer bidding patterns, and adapt their strategies accordingly. This motivates extending our model to an \emph{adaptive adversary} setting, where the competing bids in each round may depend on the bidder's past submitted strategies. Formally, we consider the model presented in \cref{sec:model} with the crucial change that the competing bids are now generated by an adaptive adversary, i.e., the competing bids in any round $t$ can depend on the strategies submitted by the bidder in previous rounds, $\ibid^1, \dots, \ibid^{t-1}$. For an adaptive adversary, we consider the following performance metric:
\begin{align}\label{def:adaptive-reg}
\textsf{REG}_{\text{adap}}=\E\left[\max_{\ibid\in\optufclass{\numbid}}\sum_{t=1}^T\val(\ibid; \otherbid{t}) - \sum_{t=1}^T\val(\ibid^t; \otherbid{t})\right]\,,
\end{align}
where the expectation is with respect to the randomness of the learning algorithm of the bidder.\footnote{Without loss of generality, we assume that the mapping from past bidding strategies, $\ibid^1, \dots, \ibid^{t-1}$, to current competing bids, $\otherbid{t}$, is deterministic~\cite[Remark 3.1]{dani2006robbing}.} Observe that unlike the case for an oblivious adversary~(see \cref{eq:safe-regret}), the first term in \cref{def:adaptive-reg} is not fixed in advance for an adaptive adversary as the competing bids $\otherbid{t}$ can adaptively change as a function of historical bids. 

In the remaining portion of this section, we propose learning algorithms that yield $\textsf{REG}_{\text{adap}}=o(T)$ both under the full-information and bandit feedback model, when the competing bids are generated by an adaptive adversary. All the remaining modeling assumptions from \cref{sec:model} continue to hold without any change. 

\textbf{Full Information Setting.} In our setting, a learning algorithm is said to be \textit{forgetful}~\citep{dani2006robbing}, if the probability distribution for choosing a strategy $\ibid^t$ in round $t$ is fully determined by the past competing bids $\otherbid{1}, \dots, \otherbid{t-1}$, and does not explicitly depend on the bidder's actions $\ibid^1, \dots, \ibid^{t-1}$. In particular,
\begin{align}\label{eq:forgetful}
    \P[\ibid^t=\ibid|\otherbid{1}, \dots, \otherbid{t-1}, \ibid^1, \dots, \ibid^{t-1}] = \P[\ibid^t=\ibid|\otherbid{1}, \dots, \otherbid{t-1}],\quad \forall \ibid\in\optufclass{\numbid}\,.
\end{align}

Now, we state the following result that rephrases \cite[Lemma 4.1]{cesa2006prediction}:
\begin{lemma}\label{lem:forgetful}
Suppose a learning algorithm satisfies \cref{eq:forgetful} and achieves $\textsf{REG}\leq B$ against an \emph{oblivious} adversary. Then, the same algorithm also yields $\textsf{REG}_{\text{adap}}\leq B$ when deployed against an adaptive adversary.
\end{lemma}
From \cref{sec:learning-safe}, recall that in the full-information setting, \cref{alg:weight-pushing} is an equivalent implementation of the Hedge algorithm. Hence, for any learning rate $\eta>0$, 
\begin{align*}
    \P[\ibid^t=\ibid|\otherbid{1}, \dots, \otherbid{t-1}, \ibid^1, \dots, \ibid^{t-1}] = \frac{\exp(\eta\sum_{\tau=1}^{t-1}\val(\ibid; \otherbid{\tau}))}{\sum_{\ibid'\in\optufclass{\numbid}}\exp(\eta\sum_{\tau=1}^{t-1}\val(\ibid'; \otherbid{\tau}))} =\P[\ibid^t=\ibid|\otherbid{1}, \dots, \otherbid{t-1}]\,.
\end{align*}

This implies \cref{alg:weight-pushing} is forgetful and hence  yields $\textsf{REG}_{\text{adap}} \leq O(\maxbid\sqrt{\numbid T \log \maxbid})$ in the full-information setting even when the competing bids are generated by an adaptive adversary.

\textbf{Bandit Setting.}  
In the bandit feedback model, the forgetfulness property of \cref{alg:weight-pushing} no longer holds, because the unbiased estimator of the value of a strategy \( \ibid \) depends on the bidder’s past actions. This dependency necessitates a different approach for controlling regret.

To bound \( \textsf{REG}_{\text{adap}} \), we begin by analyzing the \emph{random regret}, defined as
\[
    \widehat{\textsf{REG}}_{\text{adap}} = \max_{\ibid \in \optufclass{\numbid}} \sum_{t=1}^T \val(\ibid; \otherbid{t}) - \sum_{t=1}^T \val(\ibid^t; \otherbid{t}),
\]
and we show that it is bounded with high probability. We then relate this high-probability bound to the expected regret using the following lemma:

\begin{lemma}\label{lem:whp-to-expected}
    For any \( \delta \in (0, 1] \), if \( \widehat{\textsf{REG}}_{\text{adap}} \le B(\delta) \) with probability at least \( 1 - \delta \), then 
    \[
    \textsf{REG}_{\text{adap}} \leq B(\delta) + \maxbid T \delta.
    \]
\end{lemma}

To address the challenge posed by adaptive adversaries, we develop a new algorithm inspired by the EXP3.P algorithm~\citep{auer2002nonstochastic}, which provides high-probability regret guarantees. Our method builds on \cref{alg:weight-pushing} and incorporates ideas from~\citet{gyorgy2007line}.

The algorithm maintains a \emph{covering path set} \( \mathcal{C} \), which is a collection of \( s \)--\( d \) paths in the DAG \( \mathcal{G}(V, E) \) such that every edge \( e \in E \) appears in at least one path \( \path \in \mathcal{C} \). At each round \( t \), the algorithm samples a path using a mixture strategy: with probability \( \lambda > 0 \), it samples uniformly at random from \( \mathcal{C} \); with probability \( 1 - \lambda \), it updates edge weights and samples a path as a Markov chain, following the procedure of \cref{alg:weight-pushing}. The resulting path distribution is thus a convex combination of the uniform distribution over \( \mathcal{C} \) and the distribution induced by Hedge-like updates. See \cref{alg:weight-pushing-adaptive-whp} in \cref{apx:bandit-feedback-adaptive} for implementation details. This yields the following regret bound:

\begin{theorem}\label{thm:bandit-adaptive}
    Against an adaptive adversary, there exists a learning algorithm that requires \( \text{poly}(\numbid, \maxbid) \) space, runs in \( \text{poly}(\numbid, \maxbid) \) time per round and, for any \( \delta \in (0, 1] \) and sufficiently large \( T \), guarantees
    \[
    \widehat{\textsf{REG}}_{\text{adap}} \leq O\left(\numbid^{3/2} \maxbid^2 \sqrt{T \log \maxbid} + \numbid \maxbid^2 \sqrt{T \log (\maxbid/\delta)}\right)
    \]
    with probability at least \( 1 - \delta \) under bandit feedback. Setting \( \delta = 1/T \) gives the bound
    \[
    \textsf{REG}_{\text{adap}} \leq O\left(\numbid^{3/2} \maxbid^2 \sqrt{T \log (\maxbid T)}\right).
    \]
\end{theorem}

\subsection{Cumulative RoI Constraints}\label{ssec:cumulative-roi}
In this section, we consider a natural relaxation of the per-round RoI constraints introduced in \cref{sec:model}. Specifically, we assume that bidders aim to satisfy the RoI constraint in aggregate over the entire horizon of $T$ rounds, which is a common assumption in online ad auction literature~\citep{lucier2023autobidders, feng2023online, deng2021towards, deng2023multi}. In this setting, we define regret as 
\begin{align}\label{eq:regret-def-cumulative}
    \textsf{REG}_{\text{cumul}}=\max_{\ibid\in\generic}\sum_{t=1}^T\val(\ibid; \otherbid{t}) - \sum_{t=1}^T\E[\val(\ibid^t; \otherbid{t})],
\end{align}
where the expectation is with respect to any randomness in the learning algorithm. Here, $\generic$ is the class of strategies that satisfy cumulative RoI constraints over $T$ rounds for the sequence of competing bids $[\otherbid{t}]_{t\in[T]}$. In this setting, we consider the following assumption on the bidder's behavior:

\begin{assumption}\label{assumpt:1}
     For adversarial competing bids, a value maximizing bidder aiming to satisfy cumulative RoI constraints over $T$ rounds does not submit non-safe strategies, i.e., overbidding strategies per \cref{def:under-over_bid}, in any round until they have accrued a non-zero value from previous auctions. Formally, for any round $t\in[T]$, if $\sum_{\tau=1}^{t-1}\val(\ibid^\tau; \otherbid{\tau})=0$, then the bidder submits $\ibid^t\in\optufclass{\numbid}$.

\end{assumption}



\cref{assumpt:1} is a natural assumption in adversarial settings. We argue that if a bidder does not adhere to \cref{assumpt:1}, then there exists a sequence of competing bids under which their cumulative RoI constraint will necessarily be violated. To illustrate this, consider a bidder, not having won any units in rounds $1, \dots, t-1$, who overbids in round~$t$ and wins at least one unit violating the round-wise RoI constraint.\footnote{For any overbidding strategy per \cref{def:under-over_bid}, there exists a competing bid profile for which the RoI constraint in that round is violated. See proof of \cref{thm:P-m} for details.} Suppose now that in all subsequent rounds $t+1, \dots, T$, the competing bids are fixed at $v_1 + \epsilon$, irrespective of the bidder's strategy, the initial RoI violation cannot be made up for in subsequent rounds, and the cumulative RoI constraint is violated. This highlights the risk of early overbidding in an adversarial environment and motivates the practicality of \cref{assumpt:1}. Under \cref{assumpt:1}, we have the following strong impossibility result:
\begin{theorem}\label{thm:cumulative-impossible}
    There exists a sequence of competing bids for which there exists no constant $c\in(0, 1]$ for which the $c$-approximate regret is sublinear in $T$, i.e., $c$-$\textsf{REG}_{\text{cumul}}=\Omega(T)$ for every $c\in(0, 1]$.
\end{theorem}
Despite this impossibility result, we propose a flexible heuristic based on safe strategies and \cref{alg:weight-pushing} for practitioners seeking to satisfy cumulative RoI constraints in practice.

To design our heuristic, we adopt a general behavioral model in which bidders aim to satisfy the RoI constraint over a sliding window of $T_0\geq 1$ rounds. Formally, for any sequence of competing bids $[\otherbid{t}]_{t \in [T]}$, the bidder seeks to ensure that
\begin{align}\label{eq:cumulative-roi-const}
\sum_{\tau = \max(1, t - T_0 + 1)}^t \left( \val(\ibid^\tau; \otherbid{\tau}) - \price(\ibid^\tau; \otherbid{\tau}) \right) \geq 0, \quad \forall t \in [T]\,.
\end{align}

An immediate consequence of the condition in \cref{eq:cumulative-roi-const} is that it guarantees the RoI constraint is satisfied in aggregate over the interval $[1, t]$ for all $t\in[T]$. To see this, let $n = \floor{t / T_0}$. Then, the interval $[1, t]$ can be partitioned into the following subintervals: $[1, t - nT_0]$, $[t - nT_0 + 1, t - (n - 1)T_0]$, $\dots$, $[t - T_0 + 1, t]$. As \cref{eq:cumulative-roi-const} ensures that the constraint holds in each disjoint interval, summing the RoI constraints over these intervals yields the desired result. 

In this setting, the per-round RoI constraint is relaxed and evaluated over sliding windows of length $T_0$. This allows the bidder some flexibility to temporarily overbid in individual auctions, as long as the cumulative value collected over the preceding $T_0-1$ rounds sufficiently exceeds the total payment incurred. Intuitively, this creates a temporary ``surplus" that can be used to justify more aggressive bidding in the current round.

To capture this effect, we introduce the notion of \textit{$\delta$-shifted} $\numbid$-uniform undominated safe strategies for $\delta \geq 0$. Specifically, we define:
\begin{align*}
\optoneufclass{\numbid}(\delta) = \left\{ \ibid = \langle(b_1, q_1), \dots, (b_\numbid, q_\numbid)\rangle : b_\l = w_{Q_\l} + \delta, , \l \in [\numbid] \right\}.
\end{align*}
We denote the union of these classes over all $k \in [\numbid]$ as $\optufclass{\numbid}(\delta) = \bigcup_{k \in [\numbid]} \optoneufclass{k}(\delta)$. This class corresponds to the set of safe bidding strategies with at most $\numbid$ bid-quantity pairs under a shifted valuation vector $\v(\delta) = \v + \delta \cdot \boldsymbol{1}$, where $\boldsymbol{1}$ is the all-ones vector. The tunable parameter $\delta$ captures the bidder’s willingness to overbid relative to their valuation, reflecting the additional flexibility afforded by the relaxed RoI constraint.

\textbf{The Heuristic.} When the bidder aims to satisfy the RoI constraints according to \cref{eq:cumulative-roi-const}, we propose a heuristic in the online setting that builds on \cref{alg:weight-pushing}. Specifically, in round $t$, the bidder constructs a DAG, $\mathcal{G}^t(V, E)$ similar to $\mathcal{G}(V, E)$, and follows the \textsf{UPDATE} and \textsf{SAMPLE} step from \cref{alg:weight-pushing} to select a path. Then, the selected path is mapped to a safe bidding strategy within the class $\optufclass{\numbid}(\delta^t)$, where
\begin{align}
    \delta^t := \frac{1}{\maxbid} \sum_{\tau=\max(1, t-T_0+1)}^{t-1}(\val(\ibid^\tau; \otherbid{\tau}) - \price(\ibid^\tau; \otherbid{\tau})), \quad\text{and}\quad \delta^1=0\,.
\end{align}
Note that $\delta^t$ can be computed in both the full-information and bandit settings. If the sampled path $\path^t=s\to (1, z_1)\to\dots\to(k, z_k)\to d$ for some $k\in[\numbid]$, the bidder submits $\ibid^t=\langle(b_1, q_1), \dots, (b_k, q_k)\rangle$ where
    \begin{align*}
      b_\l=w_{z_\l} + \delta^t \quad\text{and}\quad q_\l=z_\l-z_{\l-1}, ~\forall  \l\in [k]\,,
    \end{align*}
    where $w_{z_\l}$ is defined as per \cref{eq:w} for all $\l\in[k]$. Finally, the bidder sets the edge weights in the DAG. In the full-information setting, the edge weights are:

(i) If \(x=(\l-1, j)\) and \(y=(\l, j')\) with \(\l\in[\numbid]\) and \(j < j'\),
\begin{align*}
    \textsf{w}^t(e) = \sum_{k=j+1}^{j'}v_k\cdot\ind{w_{j'} + \delta^t\geq \ordotherbid{t}{k}}\,.
\end{align*}

(ii) If \(x=(\l, j)\) and \(y=d\), then \(\textsf{w}^t(e) = 0, \forall \l \in [\numbid], j \in [\maxbid]\).

In the bandit setting, $\textsf{w}^t(\cdot)$ is replaced by its unbiased estimator, $\widehat{\textsf{w}}^t(\cdot)$, as stated in \cref{eq:w-hat-estimate-exp3}.

We claim that a bidder following the $\delta^t$-shifted safe strategies satisfies the RoI constraints described in \cref{eq:cumulative-roi-const}. To see this, fix any round~$t$. The maximum RoI violation in round~$t$ under a $\delta^t$-shifted safe strategy is at most $\maxbid \delta^t$. Thus, the maximum violation of the RoI constraints over the window of $T_0$ rounds in the interval $[\max(1, t - T_0 + 1), t]$ is $\sum_{\tau=\max(1, t - T_0 + 1)}^{t-1} \big(\price(\ibid^\tau; \otherbid{\tau}) - \val(\ibid^\tau; \otherbid{\tau})\big) + \maxbid \delta^t, $ which equals zero by construction. Hence, the RoI constraints are satisfied as desired.

\textbf{Experiments.} We conduct experiments using EU ETS emission permit data to evaluate the proposed heuristic. First, we quantify the improvement in the bidder’s cumulative value over $T$ rounds when RoI constraints are enforced over a sliding window of $T_0$ rounds, compared to the stricter case of per-round RoI enforcement. We measure this via the \textit{relative gain} as a function of $\numbid$. Next, we compute the degree of per-round RoI feasibility by computing the \textit{relative feasibility} in each round $t \in [T]$, defined as $\frac{\val(\ibid^t; \otherbid{t})}{\price(\ibid^t; \otherbid{t})} - 1.$ When per-round RoI constraints are enforced, this quantity is guaranteed to be at least the target RoI $\gamma$ for all $t$.\footnote{In our experiments, we ensure that the smallest competing bid in each auction is positive, so that the relative feasibility is well-defined whenever the bidder wins at least one unit. Thus, $\price(\ibid^t; \otherbid{t}) = 0$ if and only if $\val(\ibid^t; \otherbid{t}) = 0$. In such cases, we define the relative feasibility to be the target RoI $\gamma$.} In our experiments, we set $T = 200$ and $T_0 \in \{1, 8, 16, 50, \infty\}$. These values are chosen to reflect real-world scenarios: approximately 200 EU ETS auctions are conducted per year, and $T_0$ values of 8, 16, and 50 correspond to the bidder enforcing RoI constraints roughly biweekly, monthly, and quarterly, respectively. In the second experiment, we fix $\numbid=4$, which reflects the average number of bid-quantity pairs submitted by bidders in EU ETS auctions~\citep{eex-eua-primary-auction}. Each experiment is repeated over 50 simulations. 

\begin{figure}[ht]
\begin{minipage}[b]{0.45\linewidth}
\centering
\includegraphics[width=\textwidth]{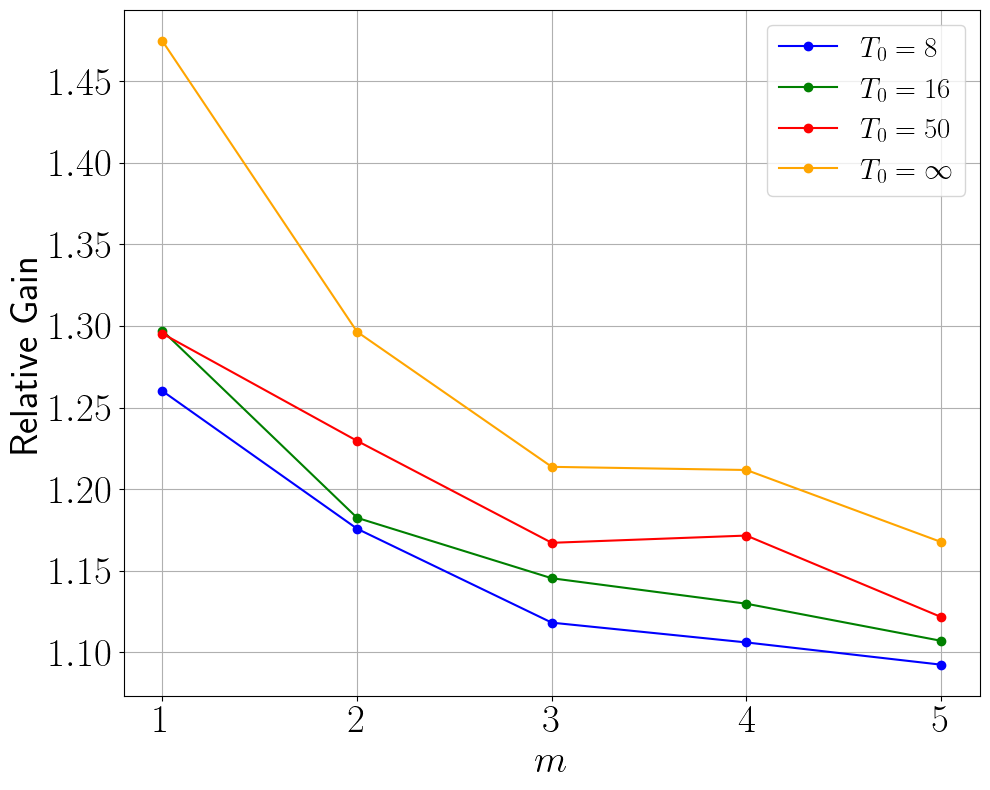}
\label{fig:relative_value}
\end{minipage}
\hspace{0.5cm}
\begin{minipage}[b]{0.45\linewidth}
\centering
\includegraphics[width=\textwidth]{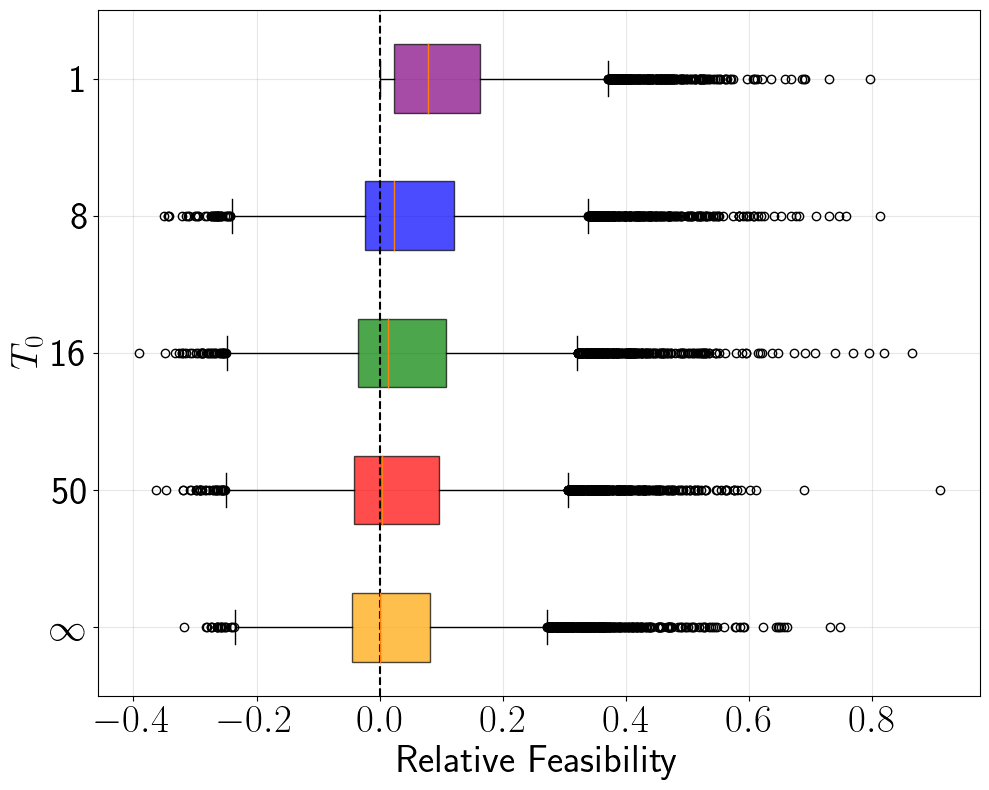}
\label{fig:relative_violations}
\end{minipage}
\vspace{-0.5cm}
\caption{The left figure shows the relative gain obtained as a function of $\numbid$. Here, relative gain is the ratio of the value obtained by the bidder over $T$ rounds when RoI constraints are enforced over $T_0$ rounds to that when RoI constraints are enforced in each round, i.e., $T_0=1$. The box plot in the right figure shows the distribution of per-round relative feasibility $\frac{\val(\ibid^t; \otherbid{t})}{\price(\ibid^t; \otherbid{t})} - 1.$ }
\label{fig:cumulative-roi}
\end{figure}


\begin{table}[h!]
\centering
\caption{Performance metrics of shifted safe strategies as a function of $T_0$. Here, $\gamma=0$ and $\numbid=4$. Here, RF stands for relative feasibility.}
\small
\label{tab:transposed-metrics}
\begin{tabular}{lccccc}
\toprule
 & $T_0=1$ & $T_0=8$ & $T_0=16$ & $T_0=50$ & $T_0=\infty$ \\
\midrule
Fraction of rounds with RoI violations $(\text{RF} \leq \gamma)$ &  
$0.0$ & $0.3225$ & $0.353$ & $0.3815$ & $0.4235$
\\
Median RF $(\frac{\val(\ibid^t; \otherbid{t})}{\price(\ibid^t; \otherbid{t})} - 1)$ & $0.0786$ & $0.0231$ & $0.0127$ & $0.0042$ & $0.0$ \\
\bottomrule
\end{tabular}
\end{table}
\normalsize
\textbf{Discussion.} Our experiments reveal a few interesting trends. In \cref{fig:cumulative-roi}, we observe that while the relative gain increases with $T_0$ for a fixed $\numbid$ as expected, for a fixed $T_0$, it decreases with $\numbid$. The fraction of rounds with RoI violations steadily increases as $T_0$ increases~(see \cref{tab:transposed-metrics}). Hence, the buyers face an inherent trade- off between the increase in obtained value and per-round RoI feasibility. For example, fixing $\numbid=4$, for $T_0=\infty$, the shifted safe strategies improve the obtained value by $\sim21\%$ compared to the safe strategies, but the RoI constraint is violated in $\sim42\%$ of the rounds. Similarly, for $T_0=8$, the same metrics are $\sim10\%$ and $\sim32\%$ respectively. Thus, based on their requirements and risk profile, bidders can employ shifted safe strategies with appropriate values of $T_0$.  The median relative feasibility also steadily decreases with larger \( T_0 \), approaching zero as the strategy becomes more aggressive. This indicates that while occasional high-value rounds may drive up cumulative gains, most rounds are closer to the feasibility barrier or even violate it.  

\subsection{Time Varying Valuations}\label{ssec:time-varying}
Recall that we had assumed the valuation vector remains fixed over the $T$ rounds. Here, we consider the case where the valuation vector in each round is chosen from a \textit{finite} set $\mathcal{V}$ in each round $t\in[T]$ and revealed to the bidder before they play their bidding strategy. As earlier, the competing bids are chosen by an oblivious adversary and the RoI constraints are enforced for each round. To satisfy the RoI constraint in round $t$, the bidder submits a safe bidding strategy corresponding to the valuation vector $\v^t$, i.e., $\ibid^t \in \optufclass{\numbid}(\v^t)$, where
\begin{align*}
    \optufclass{\numbid}(\v^t)=\bigcup_{k\in[\numbid]}\optoneufclass{\numbid}(\v^t)\quad\text{and}\quad \optoneufclass{\numbid}(\v^t)=\Big\{\ibid=\langle(b_1, q_1), \dots, (b_\numbid, q_\numbid)\rangle: b_\l = w_{Q_\l}^t,\forall \l\in[\numbid] \Big\}\,.
\end{align*}
Here, $\v^t=[v^t_1, \dots, v^t_\maxbid]$ and $w_j^t=\frac{1}{j}\sum_{\l\leq j}v_\l^t$ for all $j\in[\maxbid], t\in[T]$. 

For time varying valuations, we frame the learning task as an instance of \textit{contextual online learning} problem with a finite number of contexts, where the contexts are the valuation vectors. These contexts can be either sampled in each round from an (unknown) distribution $F_\v$ or they can be chosen adversarially at the beginning of the first round. The case of adversarially generated contexts is discussed in \cref{apx:ssec:adv-contexts}. In both cases, we assume full information feedback for ease of exposition.  




Suppose the contexts are sampled from an (unknown) distribution $F_\v$ with a finite support. The regret in this setting is defined as
\begin{align*}
    \textsf{REG}_{sto} = \max_{\pi^*\in\Pi}\sum_{t=1}^T\E_{\v^t\sim F_\v}[\val(\pi^*({\v^t}); \otherbid{t})] - \sum_{t=1}^T\E[\val(\ibid^t; \otherbid{t})]\,,
\end{align*}
where $\ibid^t\in \optufclass{\numbid}(\v^t), \forall t\in[T]$. Here, $\pi^*$ is a stationary policy belonging to the class
\begin{align}\label{eq:Pi}
  \Pi = \left\{\pi:\mathcal{V}\to\bigcup_{\v\in\mathcal{V}}\optufclass{\numbid}(\v),~~\text{such that}~~\pi(\v)\in\optufclass{\numbid}(\v), \forall \v\in\mathcal{V}\right\} \,. 
\end{align}
The expectation in the second term in the regret is with respect to the randomness in the contexts as well as any internal randomness of the learning algorithm. Our benchmark in this setting is consistent with existing literature on adversarial rewards with stochastic contexts~\citep{balseiro2019contextual,schneider2023optimal, galgana2023learning}. Suppose $\P[\v^t=\v]=p_\v, \forall \v\in\mathcal{V}, \forall t\in[T]$. As $\sum_{t=1}^T\E_{\v^t\sim F_\v}[\val(\pi({\v^t}); \otherbid{t})] = \sum_{\v\in\mathcal{V}}p_\v \cdot\sum_{t=1}^T \val(\pi(\v); \otherbid{t})$, the optimal stationary policy is 
\begin{align*}
   \pi^*(\v) = \argmax_{\ibid\in\optufclass{\numbid}(\v)}\sum_{t=1}^T\val(\ibid; \otherbid{t}), \quad\forall \v\in\mathcal{V}.
\end{align*}

 
In this setting, the bidder follows an algorithm that generalizes \cref{alg:weight-pushing} by maintaining $|\mathcal{V}|$ copies of the DAG, one for each context in $\mathcal{V}$. In each round $t \in [T]$, the context $\v^t \sim F_\v$ is observed, and the edge probabilities for every DAG copy are updated. The bidder then samples a bidding strategy $\ibid^t \in \optufclass{\numbid}(\v^t)$ and observes the competing bids $\otherbid{t}$ ex post. Finally, edge weights are assigned for all edges in all DAGs, which are used to compute edge probabilities for the next round (see \cref{alg:weight-pushing-stoc} in \cref{apx:ssec:time-varying} for details). 

\begin{theorem}\label{thm:regret-UB-stoc}
    When the contexts are stochastic, there exists an algorithm that runs in \(\text{poly}(\numbid, \maxbid, |\mathcal{V}|)\) space and per-round time and achieves \(\textsf{REG}_{sto} \leq O(\maxbid\sqrt{\numbid T\log \maxbid})\) under full information feedback. 
\end{theorem}

\section{Conclusion and Open Problems}\label{sec:conclusion}
In this work, we studied the bidding problem in repeated uniform-price auctions for a value-maximizing buyer subject to per-round return-on-investment (RoI) constraints. To address this, we introduced the notion of safe bidding strategies as those that satisfy RoI constraints regardless of competing bids and characterized the class of $\numbid$-uniform safe strategies. These strategies depend solely on the buyer’s valuation curve, and the buyer can focus on a finite subset of this class without loss of generality. We proposed a polynomial-time online learning algorithm to efficiently learn optimal strategies within this class. We then analyzed the robustness of safe strategies and the proposed algorithm against stronger benchmarks, showing that our approach achieves $\alpha$-approximate sublinear regret, where $\alpha$ is the richness ratio. We complemented the theoretical results by simulations on semi-synthetic EU ETS auction data which demonstrate that empirical richness ratios are substantially better than the worst-case bounds. We also showed the flexibility and adaptability of safe strategies and our learning algorithm across several model extensions such as adaptive adversaries, sliding-window RoI constraints, and time-varying valuations highlighting their practical relevance.

From a practitioner’s perspective, safe bidding strategies are intuitive, robust, and adaptable to a wide range of risk preferences. The proposed algorithm is simple to implement and backed by strong theoretical guarantees, making it well-suited for real-world deployment. This study opens up several exciting avenues for future investigation. A natural extension is to characterize safe strategies for value-maximizing buyers in other multi-unit auction formats, such as discriminatory-price (pay-as-bid) auctions. From a modeling standpoint, incorporating both RoI and budget constraints could lead to richer and more realistic representations of bidder behavior. While our analysis focuses on a single buyer, it would be valuable to explore market dynamics, equilibrium outcomes, and convergence properties when all participants adopt safe strategies—paralleling recent developments under different market conditions and bidders' behavior~\citep{conitzer2022pacing, li2022auto, lucier2023autobidders, fikioris2023liquid, paes2024complex}. Finally, extending this framework to combinatorial auctions with heterogeneous items presents a host of new algorithmic and design challenges. 

\section*{Acknowledgements}
N.G. and S.S. are partially supported by the MIT Junior Faculty Research Assistance Grant and the Young Investigator Award from the Office of Naval Research (ONR), Award No. N00014-21-1-2776. 

\small
\bibliography{icml2024/references}
\bibliographystyle{plainnat}
\normalsize
\newpage
\appendix



\section{Omitted Proofs from \cref{ssec:uf-strategies}}\label{apx:sec:sbs}
\subsection{Proof of \cref{thm:P-m}}
Before proving \cref{thm:P-m}, we present the following lemma about the bidder's per unit payments, which we will also use to prove several other results.

\begin{lemma}[Per-unit Payments]\label{lem:price-multiple}
Suppose the bidder bids $\ibid=\langle(b_1, q_1), \dots, (b_\numbid, q_\numbid)\rangle$. Recall that $Q_j=\sum_{\l=1}^jq_\l, \forall j\in[\numbid]$. Recall that, $\ordotherbid{t}{j}$ is the $j^{th}$ 
smallest winning bid in the absence of bids from bidder $i$ for round $t$. If $j=0$, $\ordotherbid{t}{j}=0$ and $j> K$, $\ordotherbid{t}{j}=\infty$. Then, the per-unit payments by the bidder in round $t$ is
\begin{align}
        p(\allbids^t)=\begin{cases}
            0,&~\text{if } x(\allbids^t)=0\\
            b_\l,&~\text{if } Q_{\l-1}<x(\allbids^t)<Q_\l\\
            \min(b_\l, \ordotherbid{t}{Q_\l+1}), &~\text{if } x(\allbids^t)=Q_\l
        \end{cases}\,.
    \end{align}
\end{lemma}



    Define the following class of no-overbidding (NOB) strategies with $\numbid$ bid-quantity pairs:
\begin{align*}
\oneufclass{\numbid}^{\mathbf{NOB}}:=\Big\{\ibid=\langle(b_1, q_1), \dots, (b_\numbid, q_\numbid)\rangle:
    b_\l \leq w_{Q_\l}, \forall \l\in[\numbid] \Big\}\,.
\end{align*}
We first show that $\oneufclass{\numbid} \subseteq\oneufclass{\numbid}^{\mathbf{NOB}}$, and then complete the proof by showing $ \oneufclass{\numbid}^{\mathbf{NOB}}\subseteq \oneufclass{\numbid} $. 

\textbf{Proof of $\oneufclass{\numbid} \subseteq\oneufclass{\numbid}^{\mathbf{NOB}}$.} 

\begin{observation}
    Overbidding is not a safe bidding strategy. To see this, let 
    \begin{align*}
      \ibid=\langle(b_1, q_1), \dots, (b_\l, q_\l), \dots, (b_\numbid, q_\numbid)\rangle  
    \end{align*}
    
    be an overbid such that $b_\l>w_{Q_\l}$. Now, consider an auction in which the competing bids are:
    \begin{align*}
        \ordotherbid{}{j}=\begin{cases}
           \epsilon, &~\text{if } 1\leq j\leq Q_\l\\
           2b_1, &~\text{if } Q_\l < j \leq K
        \end{cases}\,,
    \end{align*}
    where $\epsilon < \frac{b_\numbid}{2}$. Submitting $\ibid$ gets allocated $Q_\l$ units and from \cref{lem:price-multiple}, we conclude that the clearing price of auction is $b_\l$. As $b_\l>w_{Q_\l}$ by assumption, the RoI constraint is violated.
\end{observation}
As overbidding is not a safe strategy, every safe bidding strategy is a NOB strategy, i.e., $\oneufclass{\numbid} \subseteq\oneufclass{\numbid}^{\mathbf{NOB}}$.

\textbf{Proof of $\oneufclass{\numbid}^{\mathbf{NOB}}\subseteq\oneufclass{\numbid}$.} We now prove that the converse is also true, i.e., every NOB strategy is also a safe strategy. To show this, fix any competing bid, $\otherbid{}$, and consider any $\ibid=\langle(b_1, q_1), \dots, (b_\numbid, q_\numbid)\rangle\in\oneufclass{\numbid}^{\mathbf{NOB}}$. Suppose bidding $\ibid$ wins $x(\allbids)$ units. So, from \cref{lem:price-multiple}, if $x(\allbids)=0$, trivially, $\price(\allbids)=0=\val(\allbids)$. If $Q_{\l-1}<x(\allbids)\leq Q_\l$, for some $\l\in[\numbid]$,
\begin{align*}
    \price(\allbids)&= x(\allbids)\cdot p(\allbids)\leq x(\allbids)\cdot b_\l \leq x(\allbids)\cdot w_{Q_\l}\leq x(\allbids)\cdot w_{x(\allbids)}=\val(\allbids)\,.
\end{align*}
The first inequality holds true because  bidders' per-unit payment is at most their least winning bid~(individual rationality of the auction format), the second is true by definition, and the third is true because the $w_j$ is non-decreasing in $j$ and $x(\allbids)\leq Q_\l$. As the choice of competing bids and bidding strategy $\ibid$ was arbitrary, we conclude that every strategy in $\oneufclass{\numbid}^{\mathbf{NOB}}$ is safe, i.e., $\oneufclass{\numbid}^{\mathbf{NOB}}\subseteq\oneufclass{\numbid}$, which completes the proof.

\subsubsection{Proof of \cref{lem:price-multiple}}\label{apx:lem:price-multiple}
   Consider the following three cases:

   (1) If $x(\allbids^t)=0$, trivially, $p(\allbids^t)=0$.
   
   (2) Let $Q_{\l-1}<x(\allbids^t)<Q_\l$. Let $b$ be the last accepted bid after including $\ibid$, i.e., the smallest bid in $\allbids^t=(\ibid, \otherbid{t})$. Then $b_\l\stackrel{(a)}{\geq} b \stackrel{(b)}{\geq} b_\l$. 
    \begin{enumerate}
        \item[I.] $(a)$ holds true because the bidder is allocated at least one unit for bid $b_\l$ and
        \item[II.] $(b)$ is correct because they do not acquire at least one unit for bid $b_\l$. Hence, $p(\allbids^t)=b_\l$.
    \end{enumerate}

   (3) Suppose $x(\allbids^t)=Q_\l$. If $b_\l>\ordotherbid{t}{Q_\l+1}$, then $p(\allbids^t)=\ordotherbid{t}{Q_\l+1}$. However, if $b_\l\leq \ordotherbid{t}{Q_\l+1}$, $p(\allbids^t)=b_\l$. So, $p(\allbids^t)=\min(b_\l, \ordotherbid{t}{Q_\l+1})$. 

\subsection{Proof of \cref{thm:opt-bid}}
We show that for any $\numbid\in\N$, most safe bidding strategies are very weakly dominated for which they can eliminated from $\oneufclass{\numbid}$, resulting in $\optoneufclass{\numbid}$. We begin by establishing a general result regarding the monotonocity of feasible bid vectors~(not necessarily $\numbid$-uniform strategies) for value maximizing bidders. As the result holds for any bid vector, it is also true for $\numbid$-uniform bidding strategies.

\begin{lemma}[Monotonocity of feasible bids]\label{lem:monotone}
    Consider two \emph{sorted} bid vectors: $\ibid=[b_1, b_2, \dots, b_{k}]$ and $\ibid'=[b_1', b_2', \dots, b_{k}']$ such that $b_j\geq b_j', \forall j\in[k]$. Suppose $\ibid$ is RoI feasible for some competing bid, $\otherbid{}$. Then, $\ibid'$ is also feasible for $\otherbid{}$ and $\val(\ibid; \otherbid{})\geq \val(\ibid'; \otherbid{})$.
\end{lemma}

Suppose $\ibid=\langle(b_1, q_1),\dots, (b_\numbid, q_\numbid)\rangle$ is an underbidding strategy per \cref{def:under-over_bid}. Consider $\ibid'=\langle(b_1', q_1'),\dots, (b_\numbid', q_\numbid')\rangle$ such that $q_j'=q_j$ and $b_j'=w_{Q_j}, \forall j\in[\numbid]$. By \cref{thm:P-m}, we establish that $\ibid, \ibid'\in\oneufclass{\numbid}$. Invoking \cref{lem:monotone} gives us that the underbidding strategy is very weakly dominated, and hence can be removed from $\oneufclass{\numbid}$ to obtain $\optoneufclass{\numbid}$.

\subsubsection{Proof of \cref{lem:monotone}}
   Let $\allbids=(\ibid, \otherbid{})$ 
and $\allbids'=(\ibid', \otherbid{})$. We first prove that $\ibid'$ is also feasible for $\otherbid{}$. Contrary to our claim, suppose $\ibid'$ is infeasible, i.e., the value obtained by $\ibid'$ when the competing bids are $\otherbid{}$ is strictly less than the payment. Suppose $\ibid'$ is allocated $r'$ units with clearing price $p(\allbids')$ such that the RoI constraint is violated: 
\begin{align}\label{eq:mono1}
    p(\allbids') > w_{r'}.
\end{align}

Suppose $\ibid$ is allocated $r$ units when the competing bids are $\otherbid{}$. By definition of allocation and payment rule, the units allocated and the clearing price in an auction are weakly increasing in bids, so $r'\leq r$ and $p(\allbids')\leq p(\allbids)$. As $\ibid$ is feasible, 
\begin{align}\label{eq:mono2}
   p(\allbids) \leq w_{r}. 
\end{align}
Combining Equations \eqref{eq:mono1} and \eqref{eq:mono2}, we have
\begin{align*}
    p(\allbids) \leq w_{r} \stackrel{(a)}{\leq} w_{r'} < p(\allbids') \implies p(\allbids)<p(\allbids'),
\end{align*}
which is a contradiction. Here $(a)$ is true as $w_j$ is non-increasing in $j$ and $r'\leq r$. So, $\ibid'$ is feasible.

By definition of the allocation rule, the value obtained in an auction is weakly increasing in the bids submitted. As $\ibid$ and $\ibid'$ are both feasible, $\val(\ibid; \otherbid{})\geq \val(\ibid'; \otherbid{})$.
\section{Omitted Proofs From \cref{sec:learning-safe}}\label{apx:sec:proofs}

\subsection{Proof of \cref{thm:DAG-base-strategy}}\label{apx:thm:DAG-base-strategy}

\textbf{$\mathcal{G}(V, E)$ is a DAG.} For convenience, let $d=(\numbid+1, \infty)$. In the constructed graph, any node $(\l, j)$ where $\l\in[\numbid]$ has edges edges either to the next layer, i.e. nodes of the form $(\l+1, j')$ where $j'>j$ or to the destination node $(\numbid+1, \infty)$. Hence, the directed graph has a topological sorting of the nodes implying that $\mathcal{G}(V, E)$ is a DAG.

\textbf{Bijection between $s$-$d$ paths and strategies in $\optufclass{\numbid}$.}
Consider a path\footnote{Here, we assume the path has $\numbid+1$ edges without loss of generality. The same argument follows if $\path=s\to (1, z_1)\to\dots\to (k, z_k)\to d$ with $k+1$ edges for some $k\in[\numbid]$ is considered instead.} 
\begin{align*}
   \path=s\to (1, z_1)\to\dots\to (\numbid, z_\numbid)\to d\,. 
\end{align*}

By construction of $\mathcal{G}(V, E)$, edges $e=s\to(1, z_1)$ and $e=(\numbid, z_\numbid)\to d$ always exist. For any $\l\in[\numbid-1]$, the edge $(\l, z_\l)\to(\l+1, z_{\l+1})$ exists if $z_\l< z_{\l+1}$. With this path $\path$, we associate the strategy
\begin{align*}
   \ibid=\langle (b_1, q_1), \dots, (b_\numbid, q_\numbid)\rangle\quad\text{where}\quad b_\l=w_{z_\l}\quad\text{and}\quad q_\l=z_\l-z_{\l-1}, &\forall\l\in[\numbid]\,. 
\end{align*}
where $z_0=0$. By definition, $Q_j=\sum_{\l=1}^jq_\l=z_j$ and  $ b_j=w_{z_j}=w_{Q_j}, \forall j\in[\numbid]$. Hence, $\ibid\in\optufclass{\numbid}$.

Conversely, consider any safe strategy $\ibid=\langle (b_1, q_1), \dots, (b_\numbid, q_\numbid)\rangle$ with $b_j=w_{Q_j}$. With this strategy, we associate the $s$-$d$ path $\path'=s\to (1, Q_1)\to\dots\to (\numbid, Q_\numbid)\to d$ where $Q_j=\sum_{\l=1}^j q_\l, \forall j\in[\numbid]$. We claim that $\path'$ is a \textit{valid} path, i.e., all the edges exist, because by definition $e=s\to (1, Q_1)$ and $e=(\numbid, Q_\numbid)\to d$ always exist. Furthermore, for any $\l\in[\numbid-1]$, the edge $(\l, Q_l)\to (\l+1, Q_{\l+1})$ also exists as $Q_{\l+1}-Q_\l=q_{\l+1}>0$.

\textbf{Weight of $s$-$d$ paths.} By assumption, $s=(0, 0)$ and $z_0=0$. The weight of $\path=s\to (1, z_1)\to\dots\to (\numbid, z_\numbid)\to d$ is:
\begin{align*}
    \textsf{w}(\path) &= \sum_{e\in \path} \textsf{w}(e)=\sum_{\l=1}^\numbid \textsf{w}((\l-1, z_{\l-1})\to (\l, z_{\l})) \\
    &\stackrel{\eqref{eq:dag-base-weight}}{=}\sum_{\l=1}^\numbid \sum_{t=1}^T\sum_{k=z_{\l-1}+1}^{z_\l}v_k\cdot\ind{w_{z_\l}\geq \ordotherbid{t}{k}}\\
    &=\sum_{t=1}^T\sum_{\l=1}^\numbid \sum_{k=z_{\l-1}+1}^{z_\l}v_k\cdot\ind{w_{z_\l}\geq \ordotherbid{t}{k}},
\end{align*}
which is the value obtained by the safe strategy $\ibid=\langle (b_1, q_1), \dots, (b_\numbid, q_\numbid)\rangle\in\optufclass{\numbid}$ where $b_\l=w_{z_\l}, \forall\l\in[\numbid]$.

\textbf{Computing maximum weight path.} As $\mathcal{G}(V, E)$ is a DAG, the edge weights can be negated and the maximum~(resp. minimum) weight path problem in the original~(resp. `negated') DAG can be solved in space and time complexity of $O(|V|+|E|)=O(\numbid\maxbid^2)$ which is polynomial in the parameters of the problem.

\subsection{Proof of \cref{thm:full-info}}\label{apx:thm:full-info}
We begin by observing that there is a bijective mapping between $s$-$d$ paths in the DAG $\mathcal{G}^t(V, E)$ and bidding strategies $\ibid\in\optufclass{\numbid}$~(see proof of \cref{thm:DAG-base-strategy}). The core idea of the proof is to show that \cref{alg:weight-pushing} is an \textit{equivalent and efficient} implementation of the Hedge algorithm~\citep{freund1997decision} where every $s$-$d$ path is treated as an expert. 


In \cref{alg:weight-pushing}, $\varphi^t(u\to\cdot)$ denotes the probability distribution over the outgoing neighbors of node $u$. By the recursive sampling of nodes, we get that the probability of selecting $s$-$d$ path $\path$ is 
\begin{align}\label{eq:edge-recursion}
    \prob^t(\path) = \prod_{e\in\path} \varphi^t(e),
\end{align}
and for any edge $e=u\to v$ in $\mathcal{G}^t(V, E)$, the edge probabilities, $\varphi^t(e)$, are updated as
\begin{align}\label{eq:edge-prob-update}
    \varphi^t(e)=\varphi^{t-1}(e)\cdot\exp(\eta \textsf{w}^{t-1}(e))\cdot\frac{\Gamma^{t-1}(v)}{\Gamma^{t-1}(u)},
\end{align}
where $\Gamma^{t-1}(d)=1$ and $\Gamma^{t-1}(\cdot)$ is computed recursively in bottom-to-top fashion for every node $u\in V$ as follows:
    \begin{align}\label{eq:gamma-recursion}
        \Gamma^{t-1}(u)=\sum_{v:u\to v\in E}\Gamma^{t-1}(v)\cdot\varphi^{t-1}(u\to v)\cdot\exp(\eta \textsf{w}^{t-1}(u\to v))\,.
    \end{align}
    
 Now, consider a na\"ive version of the Hedge algorithm with learning rate $\eta$ in which each $s$-$d$ path is treated as an individual expert. For $t=1$, the algorithm samples a path uniformly at random from $\mathscr{P}$, the collection of all $s$-$d$ paths in the DAG. For $t\geq 2$, the probability of selecting path $\path$ in round $t$ is 
\begin{align*}
   \prob_{\textsc{hedge}}^t(\path) = \frac{\prob_{\textsc{hedge}}^{t-1}(\path)\exp(\eta\sum_{e\in\path}\textsf{w}^{t-1}(e) )}{\sum_{\path'}\prob_{\textsc{hedge}}^{t-1}(\path')\exp(\eta \sum_{e\in\path'}\textsf{w}^{t-1}(e))}\,.
\end{align*}

 We prove \cref{alg:weight-pushing} is equivalent to the Hedge algorithm by showing that $\prob^t(\path)=\prob_{\textsc{hedge}}^t(\path), \forall \path\in\mathscr{P}, \forall t\in[T]$.
 To this end, we first present the following result:
\begin{claim}\label{cl:hedge}
    For any node $u$ in the graph, let $\mathscr{P}(u)$ be the set of paths starting at $u$ and terminating at $d$. Then,
    \begin{align*}
        \Gamma^{t-1}(u)=\sum_{\path\in\mathscr{P}(u)}\prod_{e\in\path}\varphi^{t-1}(e)\cdot\exp(\eta \textsf{w}^{t-1}(e))
    \end{align*}
\end{claim}
\begin{proof}
    We prove the result by strong induction in a bottom-to-top order. For the base case, $u=d$, $\Gamma^{t-1}(d)=1$. Suppose the result is true for all the nodes in layer $\l'$ for some $\l< \l'$. By the recursion in \cref{eq:gamma-recursion}, for any node $u$ in layer $\l$,
    \begin{align*}
        \Gamma^{t-1}(u)&=\sum_{v:u\to v\in E}\Gamma^{t-1}(v)\cdot\varphi^{t-1}(u\to v)\cdot\exp(\eta \textsf{w}^{t-1}(u \to v))\\
        &=\sum_{v:u\to v\in E}\Big(\sum_{\path\in\mathscr{P}(v)}\prod_{e\in\path}\varphi^{t-1}(e)\cdot\exp(\eta \textsf{w}^{t-1}(e))\Big)\varphi^{t-1}(u \to v)\cdot\exp(\eta \textsf{w}^{t-1}(u \to v))\\
        &=\sum_{\path\in\mathscr{P}(u)}\prod_{e\in\path}\varphi^{t-1}(e)\cdot\exp(\eta \textsf{w}^{t-1}(e)),
    \end{align*}
    where the first equality follows from induction hypothesis.
\end{proof}
 \begin{claim}\label{cl:hedge-1}
     $\prob_{\textsc{hedge}}^1(\path)= \frac{1}{|\mathscr{P}|} =\prod_{e\in\path} \varphi^1(e) =\prob^1(\path)$.
 \end{claim}
 
\begin{proof}
     The first equality holds from the definition of the na\"ive Hedge algorithm above and the last equality follows from \cref{eq:edge-recursion}. To complete the proof, we need to show that $\prod_{e\in\path} \varphi^1(e)=\frac{1}{|\mathscr{P}|}$. As $\varphi^0(e)=1$ and $\textsf{w}^0(e)=0$ for all $e\in E$ by initialization in \cref{alg:weight-pushing}, for any edge $e=u\to v$, we get that 
     \begin{align*}
         \varphi^1(e)\stackrel{\eqref{eq:edge-prob-update}}{=}\frac{\Gamma^0(v)}{\Gamma^0(u)} \implies \prob^1(\path) \stackrel{\eqref{eq:edge-recursion}}{=} \prod_{e\in\path}\varphi^1(e) = \frac{\Gamma^0(d)}{\Gamma^0(s)}
     \end{align*}
     for any $s$-$d$ path $\path\in\mathscr{P}$. For any node $u\in V$, by \cref{cl:hedge}, $\Gamma^0(u)=\sum_{\path\in\mathscr{P}(u)}\prod_{e\in \path}1=|\mathscr{P}(u)|$, the number of paths that start at $u$ and terminate at $d$. By initialization, $\Gamma^0(d)=1$ and $\Gamma^0(s)=|\mathscr{P}(s)|=|\mathscr{P}|$, which proves the result.
 \end{proof}

We will show by induction on $t\in[T]$ that for any path $\path\in\mathscr{P}$, $\prob^t(\path)=\prob_{\textsc{hedge}}^t(\path)$. For $t=1$ and any path $\path$, the result is true due to \cref{cl:hedge-1}. 
Suppose the result holds for round $t-1$, i.e., $\prob^{t-1}(\path)=\prob_{\textsc{hedge}}^{t-1}(\path)=\prod_{e\in\path} \varphi^{t-1}(e)$. For round $t$,
\begin{align*}
    \prob^{t}(\path)&=\prod_{e\in\path} \varphi^{t}(e)\\
    &\stackrel{\eqref{eq:edge-prob-update}}{=}\prod_{e=u\to v \in\path} \varphi^{t-1}(e)\cdot\exp(\eta \textsf{w}^{t-1}(e))\cdot\frac{\Gamma^{t-1}(v)}{\Gamma^{t-1}(u)}\\
    &=\prob_{\textsc{hedge}}^{t-1}(\path)\cdot\exp\Big(\eta\sum_{e\in\path} \textsf{w}^{t-1}(e)\Big)\cdot\frac{\Gamma^{t-1}(d)}{\Gamma^{t-1}(s)},
\end{align*}
where the last equality follows by telescoping product and invoking the induction hypothesis. By definition, $\Gamma^{t-1}(d)=1$. By \cref{cl:hedge},
\begin{align*}
    \Gamma^{t-1}(s) = \sum_{\path'}\prod_{e\in\path'}\varphi^{t-1}(e)\cdot\exp(\eta \textsf{w}^{t-1}(e)) = \sum_{\path'}\prob_{\textsc{hedge}}^{t-1}(\path')\cdot\exp\Big(\eta\sum_{e\in\path'} \textsf{w}^{t-1}(e)\Big)
\end{align*}
Substituting the values gives $\prob^{t}(\path)=\prob_{\textsc{hedge}}^{t}(\path)$. So, \cref{alg:weight-pushing} is a correct implementation of the Hedge algorithm.

\textbf{Regret Upper Bound and Time Complexity.} Having shown that \cref{alg:weight-pushing} is equivalent to the Hedge algorithm, we now recall its regret bound from \citet[Theorem 2.2]{cesa2006prediction} which states that for learning rate $\eta$:
\begin{align*}
   \textsf{REG} \le \frac{\log |\mathscr{P}|}{\eta} + \frac{\eta T\maxbid^2}{8} \,.
\end{align*}

As $|\mathscr{P}|\leq \maxbid^{\numbid}$, 
\begin{align*}
    \textsf{REG} \leq \frac{\numbid\log \maxbid}{\eta}+ \frac{\eta T\maxbid^2}{8}\,.
\end{align*}

Setting $\eta=\frac{1}{\maxbid }\sqrt{\frac{8\numbid\log \maxbid}{T}}$, we get $\textsf{REG} \le O(\maxbid\sqrt{\numbid T\log \maxbid})$. In the full information setting, the running time bottleneck is computing \cref{eq:gamma-recursion}. Thus, \cref{alg:weight-pushing} runs in $|E|=O(\numbid\maxbid^2)$ time per round.

\subsection{Proof of \cref{thm:bandit}}\label{apx:thm:bandit}
\textbf{Regret Analysis.} In the bandit setting, the bidder follows \cref{alg:weight-pushing} by substituting $\textsf{w}^t(e)$ with its unbiased estimator $\widehat{\textsf{w}}^t(e)$. Thus, in this setting, the updates is similar to that of the EXP3 algorithm~\citep{auer2002nonstochastic}. Using the standard analysis for EXP3 algorithm~(see \citet[Chapter 11]{lattimore2020bandit}), we obtain
 \begin{align*}
     \textsf{REG} = \frac{\log |\mathscr{P}|}{\eta} +\sum_{t=1}^T\E\left[\frac{1}{\eta}\log\left(\sum_{\path}\prob^t(\path)e^{\eta \widehat{\textsf{w}}^t(\path)}\right)-\sum_{\path}\prob^t(\path)\widehat{\textsf{w}}^t(\path)\right],
 \end{align*}
where $\widehat{\textsf{w}}^t(\path)=\sum_{e\in \path}\widehat{\textsf{w}}^t(e)$. Note that $\widehat{\textsf{w}}^t(\path)=\sum_{e\in \path}\widehat{\textsf{w}}^t(e)\leq \sum_{e\in \path}\overline{\textsf{w}}(e)\le \maxbid$ and $|\mathscr{P}|\leq \maxbid^{\numbid}$. Note that if the natural importance-weighted estimator had been used instead, we would not be able to guarantee that $\widehat{\textsf{w}}^t(\path) \leq \maxbid$. For $\eta\leq \frac{1}{\maxbid }$, $\eta \widehat{\textsf{w}}^t(\path) \leq 1$. Hence,
\begin{align*}
    \textsf{REG} &\leq \frac{\numbid \log \maxbid}{\eta} + \eta \sum_{t=1}^T \sum_{\path}\prob^t(\path)\E[\widehat{\textsf{w}}^t(\path)^2]\\
    &\le\frac{\numbid \log \maxbid}{\eta} + \eta\numbid \sum_{t=1}^T \sum_{\path}\prob^t(\path)\sum_{e\in \path}\E[\widehat{\textsf{w}}^t(e)^2],
\end{align*}
where in the first inequality, we used $e^x\leq 1+x+x^2$ for $x\leq 1$ and $\log(1+x)\leq x$ for $x\ge0$. The second inequality follows from Cauchy-Schwarz. Observe that
\begin{align*}
    \sum_{\path}\prob^t(\path)\sum_{e\in \path}\E[\widehat{\textsf{w}}^t(e)^2] =\sum_{e\in E}\E[\widehat{\textsf{w}}^t(e)^2] \sum_{\path:e\in\path}\prob^t(\path)=\sum_{e\in E}p^t(e)\E[\widehat{\textsf{w}}^t(e)^2],
\end{align*}
where the last equality holds as $p^t(e)=\sum_{\path:e\in\path}\prob^t(\path)$.
\begin{align*}
   \sum_{e\in E}p^t(e)\E[\widehat{\textsf{w}}^t(e)^2] &= \sum_{e\in E}p^t(e)\Big[\overline{\textsf{w}}(e)^2(1-p^t(e))+\Big(\overline{\textsf{w}}(e)-\frac{\overline{\textsf{w}}(e)-\textsf{w}^t(e)}{p^t(e)}\Big)^2p^t(e) \Big]\\
   &=\sum_{e\in E}\textsf{w}^t(e)^2p^t(e)+(\overline{\textsf{w}}(e)-\textsf{w}^t(e))^2(1-p^t(e))\leq \sum_{e\in E}\overline{\textsf{w}}(e)^2\,.
\end{align*}
Observe that $\sum_{e}\overline{\mathsf{w}}(e)^2\leq |E|M^2\lesssim\numbid\maxbid^4$. Hence, 
\begin{align*}
   \mathsf{REG} &\leq \frac{\numbid \log \maxbid}{\eta} + \eta\numbid \sum_{t=1}^T \sum_{\path}\prob^t(\path)\sum_{e\in \path}\E[\widehat{\mathsf{w}}^t(e)^2] \\
   &\lesssim \frac{\numbid \log \maxbid}{\eta} + \eta T\numbid^2\maxbid^4\,.
\end{align*}
Setting $\eta = \frac{1}{\maxbid^2}\sqrt{\frac{\log \maxbid}{\numbid T}}$, we get $\mathsf{REG} \leq O(\maxbid^2\sqrt{\numbid^3 T\log \maxbid})$.
\textbf{Time Complexity.}  The bottleneck of running \cref{alg:weight-pushing} under the bandit feedback is the efficient computation of the marginal distribution, $p^t(e)$. We claim that this can be done in $O(|E|)=O(\numbid\maxbid^2)$ time per-round.

We begin by defining a quantity analogous to $\Gamma^{t-1}(\cdot)$. Define $\widetilde{\Gamma}^{t-1}(s)=1$ and $\widetilde{\Gamma}^{t-1}(\cdot)$ is computed recursively for every node $u\in V$ as follows:
    \begin{align}\label{eq:gamma-tilde-recursion}
        \widetilde{\Gamma}^{t-1}(u)=\sum_{v:v\to u\in E}\widetilde{\Gamma}^{t-1}(v)\cdot\varphi^{t-1}(v\to u)\cdot\exp(\eta \widehat{\textsf{w}}^{t-1}(v\to u))\,.
    \end{align}
\begin{claim}\label{cl:hedge-tilde}
    For any node $u$ in the graph, let $\widetilde{\mathscr{P}}(u)$ be the set of paths starting at $s$ and terminating at $u$. Then,
    \begin{align*}
        \widetilde{\Gamma}^{t-1}(u)=\sum_{\path\in\widetilde{\mathscr{P}}(u)}\prod_{e\in\path}\varphi^{t-1}(e)\cdot\exp(\eta \widehat{\textsf{w}}^{t-1}(e))
    \end{align*}
\end{claim}
\begin{proof}
The proof follows the same structure as \cref{cl:hedge}, but instead of proceeding bottom-up through the nodes, we use top-down induction.
\end{proof}
Let $e=u\to v$. Recall that
\begin{align*}
    p^t(e)&\stackrel{\eqref{eq:edge-recursion}, \eqref{eq:w-hat-estimate-exp3}}{=}\sum_{\path:e\in\path}\prod_{e'\in\path}\varphi^t(e')\\
    &= \varphi^t(e)\left(\sum_{\path\in\widetilde{\mathscr{P}}(u)}\prod_{e'\in\path}\varphi^t(e')\right)\left(\sum_{\path\in\mathscr{P}(v)}\prod_{e''\in\path}\varphi^t(e'')\right)\\
    &\stackrel{\eqref{eq:edge-prob-update}}{=}\varphi^t(e)\cdot\frac{\Gamma^{t-1}(u)}{\Gamma^{t-1}(s)}\cdot\frac{\Gamma^{t-1}(d)}{\Gamma^{t-1}(v)}\times\left(\sum_{\path\in\widetilde{\mathscr{P}}(u)}\prod_{e'\in\path}\varphi^{t-1}(e')\exp(\eta\widehat{\textsf{w}}^{t-1}(e'))\right)\\
    &\quad\times\left(\sum_{\path\in\mathscr{P}(v)}\prod_{e''\in\path}\varphi^{t-1}(e'')\exp(\eta\widehat{\textsf{w}}^{t-1}(e''))\right)\,.
\end{align*}
By \cref{cl:hedge}, \cref{cl:hedge-tilde}, \cref{eq:edge-prob-update} and the fact that $\Gamma^{t-1}(d)=1$, we get
\begin{align*}
   p^t(e)&= \varphi^{t-1}(e)\cdot\exp(\eta \widehat{\textsf{w}}^{t-1}(e))\cdot\frac{\Gamma^{t-1}(v)}{\Gamma^{t-1}(u)}\cdot\frac{\Gamma^{t-1}(u)}{\Gamma^{t-1}(s)}\cdot\frac{1}{\Gamma^{t-1}(v)}\cdot\widetilde{\Gamma}^{t-1}(u)\cdot\Gamma^{t-1}(v) \\
   &= \varphi^{t-1}(e)\cdot\exp(\eta \widehat{\textsf{w}}^{t-1}(e))\cdot\frac{\widetilde{\Gamma}^{t-1}(u)\Gamma^{t-1}(v)}{\Gamma^{t-1}(s)}
\end{align*}
Observe that $\Gamma^{t-1}(\cdot)$~(resp. $\widetilde{\Gamma}^{t-1}(\cdot)$) can be computed for all nodes recursively using \cref{eq:gamma-recursion}~(resp. \cref{eq:gamma-tilde-recursion}) in $O(|E|)$ time. Hence, $p^t(e)$ can be computed for all edges $e\in E$ in $O(|E|)=O(\numbid\maxbid^2)$ time in any round.
\subsection{Proof of \cref{thm:regret-LB}}\label{apx:thm:regret-LB}
For $\numbid=1$, let $K=\maxbid$ be an even integer. Let 
\begin{align}\label{eq:LB-v}
   \v=[1, \dots, 1, v, \dots, v] \,.
\end{align}
Here, the first $\frac{\maxbid}{2}$ entries are 1 and the remaining $\frac{\maxbid}{2}$ entries are $v$. Let $v=1-\delta$. Define $\delta'=\frac{\delta}{2\maxbid}$,
\begin{align}\label{eq:regret-LB-competing-bids}
    \otherbid{\clubsuit} = [1-\delta', \dots, 1-\delta']\quad\text{and}\quad
    \otherbid{\diamondsuit} = \Big[\frac{1+v}{2}-\delta', \dots, \frac{1+v}{2}-\delta'\Big]
\end{align}
Consider two scenarios: 

\textbf{Scenario 1.} In this scenario, for every $t\in[T]$, the competing bids $\otherbid{t}$ are:
\begin{align*}
    \otherbid{t}=\begin{cases}
        \otherbid{\clubsuit}, &\text{ w.p. } \frac{1}{2}+\delta,\\
        \otherbid{\diamondsuit}, &\text{ w.p. } \frac{1}{2}-\delta,\\
    \end{cases}
\end{align*}

\textbf{Scenario 2.} In this scenario, for every $t\in[T]$, the competing bids $\otherbid{t}$ are:
\begin{align*}
    \otherbid{t}=\begin{cases}
        \otherbid{\clubsuit}, &\text{ w.p. } \frac{1}{2}-\delta,\\
        \otherbid{\diamondsuit}, &\text{ w.p. } \frac{1}{2}+\delta,\\
    \end{cases}
\end{align*}

for some $\delta\in(0, 1/4)$ to be determined shortly. Assume the randomness used in different rounds are independent.

Let $P$ and $Q$ be the distribution of $[\otherbid{t}]_{t\in[T]}$ for scenario 1 and 2 respectively. Then, for $\delta\in(0, \frac{1}{4})$,
\begin{align*}
    \textsc{KL}(P||Q) = T\cdot\textsc{KL}(\textsc{Bern}(0.5+\delta)||\textsc{Bern}(0.5-\delta))=2T\delta\log\Big(\frac{1+2\delta}{1-2\delta}\Big)\leq \frac{8T\delta^2}{1-2\delta}\leq 16T\delta^2
\end{align*}
where the first inequality follows from $\log(\frac{1+x}{1-x})\leq \frac{2x}{1-x}$. By \citet[Lemma 2.6]{tsybakov2009introduction},
\begin{align*}
    1-\textsc{TV}(P, Q) \geq \frac{1}{2}\exp{(-\textsc{KL}(P||Q))}\geq \frac{1}{2}\exp{\left(-16T\delta^2\right)}\,.
\end{align*}

Consider the class of 1-uniform safe strategies. We first show that we need to consider only two strategies out of the $\maxbid$ possible strategies. Recall that any strategy in this class is of the form $(w_q, q)$.

\textbf{Case 1. $1\leq q\leq \frac{\maxbid}{2}$}: Here, $w_q=1, \forall q$. So, $(1, \frac{\maxbid}{2})$ weakly dominates all the strategies of the form $(1, q)$ for $1\leq q\leq \frac{\maxbid}{2}$.

\textbf{Case 2. $\frac{\maxbid}{2}+1\leq q\leq \maxbid$}: In this interval, the highest bid value is for $q=\frac{\maxbid}{2}+1$ due to diminishing marginal returns property. Note that,
\begin{align*}
   w_{\frac{\maxbid}{2}+1}=\frac{\frac{\maxbid}{2}+1-\delta}{\frac{\maxbid}{2}+1}=1-\frac{\delta}{\frac{\maxbid}{2}+1}\leq 1-\frac{\delta}{\maxbid}<1-\delta'\,. 
\end{align*}

Hence, no strategy of form $(w_q, q)$ for $q\geq \frac{\maxbid}{2}+1$ is allocated any unit in when the competing bid profile is $\otherbid{\clubsuit}$. The smallest bid value for $\frac{\maxbid}{2}+1\leq q\leq \maxbid$ is for $q=\maxbid$. Note that, $w_{\maxbid}=\frac{1+v}{2}>\frac{1+v}{2}-\delta'$. Hence, $(w_{\maxbid}, \maxbid)=(\frac{1+v}{2}, \maxbid)=(1-\frac{\delta}{2}, \maxbid)$ is allocated all the units when the competing bid profile is $\otherbid{\diamondsuit}$ and, by definition, weakly dominates all the strategies of the form $(w_q, q)$ for $q\geq \frac{\maxbid}{2}+1$. Hence, we have two undominated~(for the constructed competing bids) strategies: $\ibid_1=(1, \frac{\maxbid}{2})$ and $\ibid_2=(1-\frac{\delta}{2}, \maxbid)$.

Now, we compute the expected value obtained by $\ibid_1$ and $\ibid_2$ for scenarios 1 and 2. 
\begin{align*}
    \E_P[\val(\ibid_1; \otherbid{t})] &= \E_Q[\val(\ibid_1; \otherbid{t})] =\frac{\maxbid}{2}\\
    \E_P[\val(\ibid_2; \otherbid{t})] &= \maxbid\left(\frac{1}{2}-\delta\right)\left(\frac{1+v}{2}\right) = \frac{\maxbid}{2}\left(\frac{1}{2}-\delta\right)(2-\delta)< \frac{\maxbid}{2}\\
    \E_Q[\val(\ibid_2; \otherbid{t})] &=\maxbid\left(\frac{1}{2}+\delta\right)\left(\frac{1+v}{2}\right) = \frac{\maxbid}{2}\left(\frac{1}{2}+\delta\right)(2-\delta)> \frac{\maxbid}{2}
\end{align*}
So, $\ibid_1$~(resp. $\ibid_2$) is optimal for scenario 1~(resp. scenario 2). 
Now, consider the distribution $\frac{P+Q}{2}$, i.e., $\textsc{Bern}(0.5)$.
\begin{align}
   \E_{\frac{P+Q}{2}}[\val(\ibid_1; \otherbid{t})] = \frac{\maxbid}{2} ~~\text{ and }~~&\E_{\frac{P+Q}{2}}[\val(\ibid_2; \otherbid{t})] =\frac{\maxbid}{2} \left(\frac{1+v}{2}\right) \leq \frac{\maxbid}{2}\nonumber\\
   \implies \max_{\ibid\in\optufclass{1}}\E_{\frac{P+Q}{2}}[\val(\ibid; \otherbid{t})]&\leq \frac{\maxbid}{2}\label{eq:p-q-UB}
\end{align}
Hence, for any $\ibid\in\optufclass{1}$, and any round $t\in[T]$,
\begin{align*}
    &\max_{\ibid^*\in\optufclass{1}}\E_P[\val(\ibid^*; \otherbid{t})-\val(\ibid; \otherbid{t})] + \max_{\ibid^*\in\optufclass{1}}\E_Q[\val(\ibid^*; \otherbid{t})-\val(\ibid; \otherbid{t})]\\
    &\geq \max_{\ibid^*\in\optufclass{1}}\E_P[\val(\ibid^*; \otherbid{t})] + \max_{\ibid^*\in\optufclass{1}}\E_Q[\val(\ibid^*; \otherbid{t})] -2\max_{\ibid^*\in\optufclass{1}}\E_{\frac{P+Q}{2}}[\val(\ibid^*; \otherbid{t})]\\
    &=\E_P[\val(\ibid_1; \otherbid{t})] + \E_Q[\val(\ibid_2; \otherbid{t})] -2\max_{\ibid^*\in\optufclass{1}}\E_{\frac{P+Q}{2}}[\val(\ibid^*; \otherbid{t})]\\
    &\stackrel{\eqref{eq:p-q-UB}}{\geq}\frac{\maxbid}{2}+\frac{\maxbid}{2}\left(\frac{1}{2}+\delta\right)(2-\delta)-2\cdot\frac{\maxbid}{2} = \frac{\maxbid}{2}\left(\frac{3\delta}{2}-\delta^2\right)\geq \frac{5\maxbid\delta}{8},
\end{align*}

where the last inequality follows as $\delta\in(0, \frac{1}{4})$. Hence, any strategy $\ibid\in\optufclass{1}$ incurs a total regret of $\frac{5\maxbid T\delta}{16}$ either under $P$ or under $Q$. By two-point method from~\citet[Theorem 2.2]{tsybakov2009introduction},
\begin{align*}
    \E_{\frac{P+Q}{2}}[\textsf{REG}] \geq \frac{5\maxbid T\delta}{16}\cdot(1-\text{TV}(P, Q))\geq \frac{5\maxbid T\delta}{32}\exp{\left(-16T\delta^2\right)}
\end{align*}

Setting $\delta=\frac{1}{4\sqrt{2T}}$, we get $\E_{\frac{P+Q}{2}}[\textsf{REG}] \geq\frac{\maxbid \sqrt{T}}{36\sqrt{e}}$. Hence, $\E_{\frac{P+Q}{2}}[\textsf{REG}]=\Omega(\maxbid\sqrt{T})$.

\section{Omitted Proofs from \cref{sec:learning-rich}}\label{apx:sec:learning-rich}
\subsection{Proof of \cref{lem:sum-to-max-equivalence-UF}}\label{apx:lem:sum-to-max-equivalence}
We state and prove a stronger version of the result in \cref{lem:sum-to-max-equivalence-UF}. Formally,
\begin{lemma} \label{lem:sum-to-max-equivalence} For any $\numbid\in\N$ and competing bid $\otherbid{}$,  
let $\ibid=\langle(b_1, q_1),\dots, (b_\numbid, q_\numbid)\rangle$ be a feasible (not necessarily safe) $\numbid$-uniform strategy for $\otherbid{}$. Then,
    \begin{align*}  \val(\ibid;\otherbid{})=\max_{\l\in[\numbid]}\val((b_\l, Q_\l);\otherbid{}), 
    \end{align*}
    where we recall that $Q_\l=\sum_{j=1}^\l q_j, \forall \l\in[\numbid]$.
\end{lemma}
\cref{lem:sum-to-max-equivalence} states a similar result as \cref{lem:sum-to-max-equivalence-UF}, except a key difference that the bidding strategies are not necessarily safe.

\subsubsection{Proof of \cref{lem:sum-to-max-equivalence}}
We prove the lemma via induction on $\numbid$. The base case is $\numbid=1$ for which  the result is trivially true. Now assume that the result holds for any $\numbid$-uniform bidding strategy for any competing bid $\otherbid{}$. We now show that the result holds for any $(\numbid+1)$-uniform bidding strategy which is feasible for $\otherbid{}$.



Consider any $\numbid+1$-uniform bidding strategy $\ibid=\langle(b_1, q_1),\dots, (b_\numbid, q_\numbid), (b_{\numbid+1}, q_{\numbid+1})\rangle$ feasible for $\otherbid{}$. Then,
\begin{claim}\label{claim:1}
        The bid-quantity pair $(b_{\numbid+1}, Q_{\numbid+1})$ is feasible for $\otherbid{}$ and $\val(\ibid;\otherbid{})\geq \val((b_{\numbid+1}, Q_{\numbid+1});\otherbid{})$.
    \end{claim}

    \begin{proof}
        By assumption, $\ibid$ is feasible, the total demand of $\ibid$ and $(b_{\numbid+1}, Q_{\numbid+1})$ are equal and $b_{(j)}\geq b_{\numbid+1}, \forall j\in[Q_{\numbid+1}]$, where $b_{(j)}$ denotes the bid value in $j$th position in the sorted bid vector. So, by \cref{lem:monotone}, $(b_{\numbid+1}, Q_{\numbid+1})$ is feasible and $\val(\ibid; \otherbid{}) \geq \val((b_{\numbid+1}, Q_{\numbid+1}), \otherbid{})$.
    \end{proof}

Suppose that by bidding $\ibid$, the bidder is allocated $r$ units. There are two cases: (a) $r\leq Q_\numbid$ and (b) $r> Q_\numbid$.

\textbf{Case I.} $r\leq Q_\numbid$. In this case, we have 
    \begin{align}\label{eq:sum-max-1}
    \val(\ibid;\otherbid{})=\val(\ibid[1:\numbid];\otherbid{})\,.
    \end{align}

    Hence, by \cref{claim:1} and \cref{eq:sum-max-1},
    \begin{align}
       \val(\ibid;\otherbid{})&=\max\Big\{\val(\ibid[1:\numbid];\otherbid{}), \val((b_{\numbid+1}, Q_{\numbid+1});\otherbid{}) \Big\}\,.\label{eq:sum-max-2}
    \end{align}

    \textbf{Case II.} $r>Q_\numbid$. As $r>Q_\numbid$, $b_{\numbid+1}$ is the least winning bid which implies $b_{\numbid+1}\geq \ordotherbidnew{r}$, where we recall that $\ordotherbidnew{r}$ is the $r^{th}$ smallest competing bid in $\otherbid{}$. So, $(b_{\numbid+1}, Q_{\numbid+1})$ is allocated at least $r$ units which implies $\val((b_{\numbid+1}, Q_{\numbid+1}), \otherbid{})\geq \val(\ibid; \otherbid{})$. By \cref{claim:1}, we also have $\val(\ibid; \otherbid{}) \geq \val((b_{\numbid+1}, Q_{\numbid+1}), \otherbid{})$. Hence,
    \begin{align}
       \val(\ibid;\otherbid{})=\val((b_{\numbid+1}, Q_{\numbid+1});\otherbid{})\,.
    \end{align}
    
    As $r>Q_\numbid$, $(b_{\numbid+1}, Q_{\numbid+1})$ is allocated at least $Q_\numbid+1$ units, whereas $\ibid[1:\numbid]$ has demand for~(hence, can be allocated) at most $Q_\numbid$ units. So, 
    \begin{align}
        \val(\ibid;\otherbid{})&\geq \val(\ibid[1:\numbid];\otherbid{})\nonumber\\
        \implies \val(\ibid;\otherbid{})&=\max\Big\{\val(\ibid[1:\numbid];\otherbid{}), \val((b_{\numbid+1}, Q_{\numbid+1});\otherbid{})\Big\}\,.\label{eq:sum-max-4}
    \end{align}
    
For both \textbf{Case I} and \textbf{Case II}, we get the same result~(cf. \eqref{eq:sum-max-2} and \eqref{eq:sum-max-4}). Hence,
\begin{align*}
    \val(\ibid;\otherbid{})&=\max\Big\{\val(\ibid[1:\numbid];\otherbid{}), \val((b_{\numbid+1}, Q_{\numbid+1});\otherbid{})\Big\}\\
    &\stackrel{(a)}{=}\max\Big\{\max_{\l\in[\numbid]}\val((b_\l, Q_\l);\otherbid{}), \val((b_{\numbid+1}, Q_{\numbid+1});\otherbid{})\Big\}\\
    &=\max_{\l\in[\numbid+1]}\val((b_\l, Q_\l);\otherbid{})\,.
\end{align*}
Here, $(a)$ holds as $\ibid[1:\numbid]$ is feasible for $\otherbid{}$ allowing us to apply the induction hypothesis for $\numbid$.

\subsection{Proof of \cref{lem:value_restricted}}\label{apx:lem:value_restricted}
To prove this result, we use the following key lemma that measures the value obtained by $\optbid{\numbid}$ in terms of 1-uniform safe bidding strategies. 

\begin{lemma}\label{lem:approx1_multiple} 
Let $\optbid{\numbid}=\langle(b_1^*, q_1^*),\dots, (b_\numbid^*, q_\numbid^*)\rangle$. Suppose  $Q_j^*=\sum_{\l\leq j}q_\l^*, \forall j\in[\numbid]$ and $\optbid{\numbid}$ is allocated $r_t^*$ units in any round $t\in [T]$. Then,

1. For any $t\in [T]$ and $q\le r_t^*$, $(w_{q}, q)$ gets \emph{exactly} $q$ units.

2. For any $j\in[\numbid]$, let $T_j\subseteq[T]$ be the rounds in which $b_j^*$ is the least winning bid, i.e., 
\begin{align*}
    T_j= \big\{t\in [T]: Q_{j-1}^*<r_t^*\leq Q_j^* \big\}\,.
\end{align*}
Suppose that $\exists t\in T_j$ such that $r_t^*<Q_j^*$. Then in any round $t'\in T_j$ in which $\optbid{\numbid}$ is allocated at most $r_t^*$ units, i.e., $r_{t'}^*\leq r_{t}^*$, $(w_{r_t^*}, r_t^*)$ is allocated \emph{at least} $r_{t'}^*$ units.
\end{lemma}

With this lemma, we are ready to prove \cref{lem:value_restricted}. 

\textbf{Case I: $Q_j=\widehat{Q}_j< Q_j^*$.} For any $t\in T_{j, 1}$, invoking \cref{lem:approx1_multiple}~(1) with $q=\widehat{Q}_j$, we conclude that $(w_{\widehat{Q}_j}, \widehat{Q}_j)$ is allocated exactly $\widehat{Q}_j$ units. Summing over all rounds in $T_{j, 1}$,
\begin{align}\label{eq:xj-is-qbar-1}
    \sum_{t\in T_{j, 1}}\val((w_{Q_j}, Q_j); \otherbid{t}) = \psum{\widehat{Q}_j}|T_{j, 1}|\,.
\end{align}
As $\widehat{Q}_j< Q_j^*$, $T_{j, 0}\neq\emptyset$. Let $r_t^*=\widehat{Q}_j$ for some $t\in T_{j, 0}$. By definition, $\widehat{Q}_j < Q_j^*$ and $\widehat{Q}_j\geq r_s^*$ for any $s\in T_{j, 0}$. So, for any $s\in T_{j, 0}$, invoking \cref{lem:approx1_multiple}~(2) with $r_t^*=\widehat{Q}_j$, we conclude that $(w_{\widehat{Q}_j}, \widehat{Q}_j)$ is allocated at least $r_s^*$ units. Hence, summing over all rounds, $(w_{\widehat{Q}_j}, \widehat{Q}_j)$ gets at least the value obtained by $\optbid{\numbid}$ over the rounds in $T_{j, 0}$. So,
\begin{align}\label{eq:xj-is-qbar-2}
    \sum_{t\in T_{j, 0}}\val((w_{Q_j}, Q_j); \otherbid{t})&\geq \sum_{t\in T_{j, 0}}\val\Big(\optbid{\numbid}, \otherbid{t}\Big)=V_{j, 0}|T_{j, 0}|\,.
\end{align}
Combining \cref{eq:xj-is-qbar-1} and \eqref{eq:xj-is-qbar-2}, for $Q_j=\widehat{Q}_j$, 
\begin{align}\label{eq:xj-is-qbar}
    \sum_{t\in T_j}\val((w_{Q_j}, Q_j); \otherbid{t})\geq V_{j, 0}|T_{j, 0}|+\psum{\widehat{Q}_j}|T_{j, 1}|\,.
\end{align}

\textbf{Case II: $Q_j=Q_j^*$.} So, for any $t\in T_{j, 1}$, using \cref{lem:approx1_multiple}~(1) with $q=Q_j^*$, we get that $(w_{Q_j^*}, Q_j^*)$ is allocated exactly $Q_j^*$ units, which is the same as the allocation for $\optbid{\numbid}$. Summing over all rounds in $T_{j, 1}$,
\begin{align}\label{eq:xj-is-qstar-1}
    \sum_{t\in T_{j, 1}}\val((w_{Q_j}, Q_j); \otherbid{t}) = \psum{Q_j^*}|T_{j, 1}|\,.
\end{align}
For the rounds in $T_{j, 0}$, trivially,
\begin{align}\label{eq:xj-is-qstar-2}
    \sum_{t\in T_{j, 0}}\val((w_{Q_j}, Q_j); \otherbid{t})&\geq 0\,.
\end{align}
Combining \cref{eq:xj-is-qstar-1} and \eqref{eq:xj-is-qstar-2}, for $Q_j=Q_j^*$, 
\begin{align}\label{eq:xj-is-qstar-3}
    \sum_{t\in T_j}\val((w_{Q_j}, Q_j); \otherbid{t})\geq \psum{Q_j^*}|T_{j, 1}|\,.
\end{align}

        


\subsubsection{Proof of \cref{lem:approx1_multiple}}\label{apx:lem:approx1_multiple}
(1) First we show that for any $t\in [T]$ and $q\le r_t^*$,  the 1-uniform bid $(w_{q}, q)$ is allocated \emph{exactly} $q$ units. 

As $(w_{q}, q)\in\optufclass{1}$, it is a safe strategy. By assumption, $\optbid{\numbid}$ is allocated $r_t^*$ units in round $t$. Let $\allbids^t=(\optbid{\numbid}, \otherbid{t})$. Recall that, $\ordotherbid{t}{j}$ is the $j^{th}$ 
smallest winning bid in the absence of bids from bidder $i$ for round $t$. If $r_t^*=0$, the result is vacuously true. Suppose $r_t^*>0$, then
\begin{align*}
   \ordotherbid{t}{q}\leq\ordotherbid{t}{r_t^*}\leq p(\allbids^t)\leq w_{r_t^*} \leq w_{q}\,,    
\end{align*}
    where the first inequality holds by definition of $\ordotherbid{t}{j}$ and our assumption that $q\leq r_t^*$. For the second inequality, suppose $Q_{\l-1}^*<r_t^*\leq Q_{\l}^*$ for some $\l\in[\numbid]$ which by definition implies that $b_\l^*$ is the least winning bid. By \cref{lem:price-multiple}, 
    \begin{enumerate}
        \item If $p(\allbids^t)=b^*_\l$, we have  $p(\allbids^t)\geq \ordotherbid{t}{r_t^*}$ as $b_\l^*$ is the least winning bid.
        \item If $p(\allbids^t)=\ordotherbid{t}{Q^*_\l+1}$, we have $r_t^*=Q_\l^*$ and by definition of $\ordotherbid{t}{j}$, $p(\allbids^t) = \ordotherbid{t}{Q^*_\l+1} \geq \ordotherbid{t}{r_t^*}$. 
    \end{enumerate}
    The third inequality is true as RoI constraint is satisfied by $\optbid{\numbid}$ for round $t$, and the fourth is true as $w_j$ is a non-decreasing function of $j$ and $q\leq r_t^*$. From the first and last expressions, $\ordotherbid{t}{q}\leq w_{q}$ which implies that $(w_q, q)$ is allocated at least $q$ units in round $t$. Moreover, $(w_q, q)$ can be allocated at most $q$ units. Hence, the 1-uniform bid $(w_q, q)$ is allocated \emph{exactly} $q$ units.

(2) 
For any $j\in[\numbid]$, let $T_j\subseteq[T]$, defined in \cref{eq:Tj}, be the rounds in which $b_j^
*$ is the least winning bid. Suppose that $\exists t\in T_j$ such that $r_t^*<Q_j^*$. We show that in any round $t'\in T_j$ in which $\optbid{\numbid}$ is allocated at most $r_t^*$ units, i.e., $r_{t'}^*\leq r_{t}^*$, the 1-uniform strategy $(w_{r_t^*}, r_t^*)$ is allocated \emph{at least} $r_{t'}^*$ units.

Observe that $(w_{r_t^*}, r_t^*)\in\optufclass{1}$, so it is a safe strategy. If $r_t^*=0$, the result is trivially true. Hence, suppose $r_t^*>0$ and consider the set of rounds in $T_j$ for any $j\in[\numbid]$. As $r_t^*<Q_{j}^*$, by \cref{lem:price-multiple}, $p(\allbids^t)=b_j^*$. So, for any $t'\in T_j$ such that $r_{t'}^*\leq r_t^*$,
\begin{align*}
    \ordotherbid{t'}{r_{t'}^*} \leq b_j^* \leq w_{r_t^*}\,.
\end{align*}
Here, the first inequality holds as $b_j^*$ is the least winning bid for round $t'\in T_j$ and the second one holds as the RoI constraint is true for $\optbid{\numbid}$ for round $t$. Hence, $(w_{r_t^*}, r_t^*)$ is allocated \emph{at least} $r_{t'}^*$ units in round $t'$.

\subsection{Proof of \cref{lem:tedious-algebra}}
Let $(\mathcal{S}_0^*, \widehat{\mathcal{S}}_0)=\argmax_{(\mathcal{S}^*, \widehat{\mathcal{S}})}\Big\{\sum_{j\in \mathcal{S}^*}\Big(\psum{Q_j^*}|T_{j, 1}|\Big)+\sum_{j\in\widehat{\mathcal{S}}}\Big(V_{j, 0}|T_{j, 0}| + \psum{\widehat{Q}_j}|T_{j, 1}|\Big)\Big\}$. Then,

\begin{align}\label{eq:v-opt}
\max_{\ibid\in\feasclass{\numbid}}\sum_{t=1}^T\val(\ibid; \otherbid{t})&=\sum_{t=1}^T\sum_{\l=0}^1V_{j, \l}|T_{j, \l}|\nonumber\\
&=\sum_{j\in \mathcal{S}_0^*}\Big(V_{j, 0}|T_{j, 0}|+\psum{Q_j^*}|T_{j, 1}|\Big)+\sum_{j\in\widehat{\mathcal{S}}_0}\Big(V_{j, 0}|T_{j, 0}|+\psum{Q_j^*}|T_{j, 1}|\Big)\,.
\end{align}

Consider a partition of $[\numbid]$ that differs from the maximizing partition $(\mathcal{S}_0^*, \widehat{\mathcal{S}}_0)$ by exactly one element, i.e., for any $a\in\widehat{\mathcal{S}}_0$, consider the following partition: $(\mathcal{S}_0^*\cup\{a\}, \widehat{\mathcal{S}}_0\setminus\{a\})$. By definition,
\small
\begin{align}
    \sum_{j\in \mathcal{S}_0^*}\Big(\psum{Q_j^*}|T_{j, 1}|\Big)+\sum_{j\in\widehat{\mathcal{S}}_0}\Big(V_{j, 0}|T_{j, 0}| + \psum{\widehat{Q}_j}|T_{j, 1}|\Big)&\geq \sum_{j\in \mathcal{S}_0^*\cup\{a\}}\Big(\psum{Q_j^*}|T_{j, 1}|\Big)+\sum_{j\in\widehat{\mathcal{S}}_0\setminus\{a\}}\Big(V_{j, 0}|T_{j, 0}| + \psum{\widehat{Q}_j}|T_{j, 1}|\Big)\nonumber\\
    \implies V_{a, 0}|T_{a, 0}| + \psum{\widehat{Q}_a}|T_{a, 1}|&\geq \psum{Q_a^*}|T_{a, 1}|\label{eq:a}\,.
\end{align}
\normalsize
Now, for any $b\in \mathcal{S}^*_0$, consider the following partition: $(\mathcal{S}_0^*\setminus\{b\}, \widehat{\mathcal{S}}_0\cup\{b\})$. By definition,
\small
\begin{align}
    \sum_{j\in \mathcal{S}_0^*}\Big(\psum{Q_j^*}|T_{j, 1}|\Big)+\sum_{j\in\widehat{\mathcal{S}}_0}\Big(V_{j, 0}|T_{j, 0}| + \psum{\widehat{Q}_j}|T_{j, 1}|\Big)&\geq \sum_{j\in \mathcal{S}_0^*\setminus\{b\}}\Big(\psum{Q_j^*}|T_{j, 1}|\Big)+\sum_{j\in\widehat{\mathcal{S}}_0\cup\{b\}}\Big(V_{j, 0}|T_{j, 0}| + \psum{\widehat{Q}_j}|T_{j, 1}|\Big)\nonumber\\
    \implies \psum{Q_b^*}|T_{b, 1}|- \psum{\widehat{Q}_b}|T_{b, 1}| &\geq V_{b, 0}|T_{b, 0}| \,.\label{eq:b}
\end{align}
\normalsize
Plugging in the values,

\begin{align*}
  &\frac{\max_{\ibid\in\feasclass{\numbid}}\sum_{t=1}^T\val(\ibid; \otherbid{t})}{\sum_{j\in \mathcal{S}_0^*}\Big(\psum{Q_j^*}|T_{j, 1}|\Big)+\sum_{j\in\widehat{\mathcal{S}}_0}\Big(V_{j, 0}|T_{j, 0}| + \psum{\widehat{Q}_j}|T_{j, 1}|\Big)}\\
  &= \frac{\sum_{j\in \mathcal{S}_0^*}\Big(V_{j, 0}|T_{j, 0}|+\psum{Q_j^*}|T_{j, 1}|\Big)+\sum_{j\in\widehat{\mathcal{S}}_0}\Big(V_{j, 0}|T_{j, 0}|+\psum{Q_j^*}|T_{j, 1}|\Big)}{\sum_{j\in \mathcal{S}_0^*}\Big(\psum{Q_j^*}|T_{j, 1}|\Big)+\sum_{j\in\widehat{\mathcal{S}}_0}\Big(V_{j, 0}|T_{j, 0}| + \psum{\widehat{Q}_j}|T_{j, 1}|\Big)}\\
  &\stackrel{\eqref{eq:a}}{\leq}\frac{\sum_{j\in \mathcal{S}_0^*}\Big(V_{j, 0}|T_{j, 0}|+\psum{Q_j^*}|T_{j, 1}|\Big)+\sum_{j\in\widehat{\mathcal{S}}_0}\Big(2V_{j, 0}|T_{j, 0}|+\psum{\widehat{Q}_j}|T_{j, 1}|\Big)}{\sum_{j\in \mathcal{S}_0^*}\Big(\psum{Q_j^*}|T_{j, 1}|\Big)+\sum_{j\in\widehat{\mathcal{S}}_0}\Big(V_{j, 0}|T_{j, 0}| + \psum{\widehat{Q}_j}|T_{j, 1}|\Big)}\\
  &\stackrel{\eqref{eq:b}}{\leq}\frac{\sum_{j\in \mathcal{S}_0^*}\Big(2\psum{Q_j^*}|T_{j, 1}|-\psum{\widehat{Q}_j}|T_{j, 1}|\Big)+\sum_{j\in\widehat{\mathcal{S}}_0}\Big(2V_{j, 0}|T_{j, 0}|+\psum{\widehat{Q}_j}|T_{j, 1}|\Big)}{\sum_{j\in \mathcal{S}_0^*}\Big(\psum{Q_j^*}|T_{j, 1}|\Big)+\sum_{j\in\widehat{\mathcal{S}}_0}\Big(V_{j, 0}|T_{j, 0}| + \psum{\widehat{Q}_j}|T_{j, 1}|\Big)}\\
  &=\frac{2\Big\{\sum_{j\in \mathcal{S}_0^*}\psum{Q_j^*}|T_{j, 1}|+\sum_{j\in\widehat{\mathcal{S}}_0}\Big(V_{j, 0}|T_{j, 0}|+\psum{\widehat{Q}_j}|T_{j, 1}|\Big)\Big\}-\sum_{j=1}^\numbid\psum{\widehat{Q}_j}|T_{j, 1}|}{\sum_{j\in \mathcal{S}_0^*}\Big(\psum{Q_j^*}|T_{j, 1}|\Big)+\sum_{j\in\widehat{\mathcal{S}}_0}\Big(V_{j, 0}|T_{j, 0}| + \psum{\widehat{Q}_j}|T_{j, 1}|\Big)}\\
  &=2-\theta_{\hist}\,,
\end{align*}

where

\begin{align*}
    0< \theta_{\hist} &= \frac{\sum_{j=1}^\numbid\psum{\widehat{Q}_j}|T_{j, 1}|}{\sum_{j\in \mathcal{S}_0^*}\Big(\psum{Q_j^*}|T_{j, 1}|\Big)+\sum_{j\in\widehat{\mathcal{S}}_0}\Big(V_{j, 0}|T_{j, 0}| + \psum{\widehat{Q}_j}|T_{j, 1}|\Big)}\\
    &\stackrel{\eqref{eq:a}}{\leq} \frac{\sum_{j=1}^\numbid\psum{\widehat{Q}_j}|T_{j, 1}|}{\sum_{j\in \mathcal{S}_0^*}\Big(\psum{Q_j^*}|T_{j, 1}|\Big)+\sum_{j\in\widehat{\mathcal{S}}_0}\Big(\psum{Q_j^*}|T_{j, 1}|\Big)}=\frac{\sum_{j=1}^\numbid\psum{\widehat{Q}_j}|T_{j, 1}|}{\sum_{j=1}^\numbid\psum{Q_j^*}|T_{j, 1}|}\stackrel{(a)}{\leq} \max_{j\in[\numbid]}\frac{\psum{\widehat{Q}_j}}{\psum{Q_j^*}}\leq 1\,,
\end{align*}

and $(a)$ follows because for positive $a_1, \dots, a_\numbid$ and $b_1, \dots, b_\numbid$, $\frac{\sum_{j=1}^\numbid a_j}{\sum_{j=1}^\numbid b_j}\leq \max_{j\in[\numbid]}\frac{a_j}{b_j}$. 



\subsection{Proof of \cref{thm:m-mbar}}\label{apx:thm:m-mbar}
We state and prove a stronger result which recovers the upper bound in \cref{thm:m-mbar} as a corollary and is also crucial to prove the upper bound in \cref{thm:mbar-non-safe}. Let $\feasclass{\numbid}$ be the class of bidding strategies with at most $\numbid$ bid-quantity pairs that are feasible for the bid history $\hist=[\otherbid{t}]_{t\in[T]}$. Then,
\begin{theorem}\label{thm:m-mbar-strong}
For any $\numbid, \numbid'\in\N$ such that $\numbid'\geq \numbid$, $\Lambda_{\feasclass{\numbid'},\feasclass{\numbid}}\leq\frac{\numbid'}{\numbid}$.
\end{theorem}

\begin{proof}
For $\hist=[\otherbid{t}]_{t\in[T]}$, we define
\begin{align*}
    \Lambda_{\feasclass{\numbid'},\feasclass{\numbid}}(\hist, \v) = \frac{\max_{\ibid'\in\feasclass{\numbid'}}\sum_{t=1}^T\val(\ibid'; \otherbid{t})}{\max_{\ibid\in\feasclass{\numbid}}\sum_{t=1}^T\val(\ibid; \otherbid{t})}\,.
\end{align*}
Then, maximizing over all bid histories and valuation vectors, we get the desired result.

Let $\optbid{\numbid'}=\argmax_{\ibid\in\feasclass{\numbid'}}\sum_{t=1}^T\val(\ibid; \otherbid{t})$. If $\optbid{\numbid'}$ is a $k$-uniform bidding strategy where $k\leq \numbid\leq \numbid'$, by definition, $\optbid{\numbid'}\in\feasclass{\numbid}$ implying $\Lambda_{\feasclass{\numbid'},\feasclass{\numbid}}(\hist, \v)=1\leq \frac{\numbid'}{\numbid}$.

Without loss of generality, let the optimal bidding strategy in $\feasclass{\numbid'}$ be 
\begin{align*}
  \optbid{\numbid'}=\langle(b^*_1, q^*_1),\dots, (b^*_{\numbid'}, q^*_{\numbid'})\rangle  
\end{align*}

and $Q_\l^*=\sum_{j\leq \l}q_j^*$ for all $\l\in[\numbid']$.\footnote{A similar analysis also follows if $\optbid{\numbid'}$ is a $k$-uniform strategy where $\numbid\leq k\leq \numbid'$.} Let $r_t^*$  be the number of units allocated to $\optbid{\numbid'}$ in round $t$. Recall that $T_j$ is the set of rounds when $b_j^*$ is the least winning bid when $\optbid{\numbid'}$ is submitted~(see \cref{eq:Tj}). So,
\begin{align*}
\sum_{t=1}^T\val(\optbid{\numbid'}; \otherbid{t})=\sum_{j=1}^{\numbid'}\sum_{t\in T_j}\val(\optbid{\numbid'}; \otherbid{t})\,.
\end{align*}
Observe that, for any $t\in T_j$,
\begin{align*}
    \val(\optbid{\numbid'};\otherbid{t})=\val(\optbid{\numbid'}[1:j];\otherbid{t})=\max_{\l\in[j]}\val((b_\l^*, Q_\l^*);\otherbid{t})\,,
\end{align*}
where recall that $\ibid[1:j]$ represents the first $j$ bid-quantity pairs and the last equality holds due to \cref{lem:sum-to-max-equivalence}. For any $t\in T_j$, $r_t^*> Q_{j-1}^*$. So, $\forall \l < j$, $\val(\optbid{\numbid'};\otherbid{t})>\val((b_\l^*, Q_\l^*);\otherbid{t})$ as the demand of the strategy $(b_\l^*, Q_\l^*)$ is strictly less than the number of units obtained by $\optbid{\numbid'}$ in that round. Hence, for any $t\in T_j$, 
\begin{align}\label{eq:opt-value-equivalence}
    \val(\optbid{\numbid'};\otherbid{t})&=\val((b_j^*, Q_j^*);\otherbid{t})\implies \sum_{t=1}^T\val(\optbid{\numbid'}; \otherbid{t}) = \sum_{j=1}^{\numbid'}\sum_{t\in T_j} \val((b_j^*, Q_j^*);\otherbid{t})\,.
\end{align}
Let $\onefeasclass{\numbid}\subseteq\feasclass{\numbid}$ be the class of feasible strategies for $\hist$ that contain \textit{exactly} $\numbid$ bid-quantity pairs. Then, 
\begin{align}\label{eq:m-mbar-UB-helper}
   \Lambda_{\feasclass{\numbid'},\feasclass{\numbid}}(\hist, \v)\stackrel{\eqref{eq:opt-value-equivalence}}{=} \frac{\sum_{j=1}^{\numbid'}\sum_{t\in T_j} \val((b_j^*, Q_j^*);\otherbid{t})}{\max_{\ibid\in\feasclass{\numbid}}\sum_{t=1}^T\val(\ibid; \otherbid{t})} \leq \frac{\sum_{j=1}^{\numbid'}\sum_{t\in T_j} \val((b_j^*, Q_j^*);\otherbid{t})}{\max_{\ibid\in\onefeasclass{\numbid}}\sum_{t=1}^T\val(\ibid; \otherbid{t})}\,. 
\end{align}

\textbf{Constructing the class of bidding strategies, $\resonefeasclass{\numbid}$.}
We now construct a class of $\numbid$-uniform bidding strategies, $\resonefeasclass{\numbid}$ using $\optbid{\numbid'}$ as a `parent' bidding strategy, in \cref{alg:construct-bidclass}. 

\begin{algorithm}[!tbh]
\caption{Constructing $\resonefeasclass{\numbid}$}
\label{alg:construct-bidclass}
\small{
\begin{algorithmic}[1]
\Require $\optbid{\numbid'}=\langle(b^*_1, q^*_1),\dots, (b^*_{\numbid'}, q^*_{\numbid'})\rangle$. The collection, $\mathscr{C}$, of all subsets $\mathcal{S}\subseteq[\numbid']$ such that $|\mathcal{S}|=\numbid$. 

\State Initialize $\resonefeasclass{\numbid}\gets\emptyset$. 
\For{$\mathcal{S}\in\mathscr{C}$}
    \State Suppose $\mathcal{S}=\{k_1, \dots, k_{\numbid}\}$ such that ${k_1}<{k_2}<\dots<{k_{\numbid}}$.
    \State Construct a $\numbid$-uniform bidding strategy $\ibid=\langle (b_1, q_1), \dots, (b_\numbid, q_\numbid)\rangle$ where
    \begin{align}\label{eq:construct-b}
    b_j=b_{k_j}^* \quad\text{and}\quad q_j=Q_{k_j}^*-Q_{k_{j-1}}^*, ~&\forall j\in[\numbid]\,.
\end{align}
\State $\resonefeasclass{\numbid}\gets\resonefeasclass{\numbid}\cup\{\ibid\}$
\EndFor

\State \Return The class of bidding strategies, $\resonefeasclass{\numbid}$.
\end{algorithmic}}
\end{algorithm}


For example: suppose $\numbid'=4$ and $\optbid{\numbid'}=\langle (w_3, 3), (w_7, 4), (w_8, 1), (w_{10}, 2)\rangle$. Let $\numbid=2$ and $\mathcal{S}=\{2, 4\}$. Then, $\ibid=\langle (w_7, 7), (w_{10}, 3)\rangle$. 

\begin{lemma}\label{lem:new-class-2}
Let $b_{(j)}^*$ be the $j^{th}$ entry when $\optbid{\numbid'}$ is expressed as a bid vector. Similarly, for any $\ibid\in\resonefeasclass{\numbid}$, let $b_{(j)}$ be the $j^{th}$ entry when $\ibid$ is expressed as a bid vector. Then, for any $j\in[Q_{k_{\numbid}}]$, $b_{(j)}\leq b_{(j)}^*$. Thus, by \cref{lem:monotone}, $\ibid$ is also feasible for $\hist$ and as $\ibid$ was chosen arbitrarily, we conclude that $\resonefeasclass{\numbid}\subseteq\onefeasclass{\numbid}$.
\end{lemma}
\begin{proof}
Let $\mathcal{S}$ contain $k_1<\dots<{k_{\numbid}}$ and $\ibid$ be constructed per \cref{alg:construct-bidclass}. Then, for any $j\in[\numbid]$ and any entry $\l\in(Q_{k_{j-1}}^*, Q_{k_{j}}^*]$, 
\begin{align*}
    b_{(\l)}^*&\geq b_{(Q_{k_j}^*)}^*=b_{k_j}^* = b_{(\l)},
\end{align*}
where the first inequality follows because entries in bid vector are in non increasing order, the first equality follows as $b_{(Q_r)}^*=b_r^*$ for all $r\in[m']$ and second equality follows from \cref{eq:construct-b}. As $j$ and $\l$ were picked arbitrarily, we get that for all $j\in[Q_{k_{\numbid}}]$, $b_{(j)}\leq b_{(j)}^*$.
\end{proof}

By \cref{lem:new-class-2},
\begin{align}
  \Lambda_{\feasclass{\numbid'},\feasclass{\numbid}}(\hist, \v)\stackrel{\eqref{eq:m-mbar-UB-helper}}{\le}\frac{\sum_{j=1}^{\numbid'}\sum_{t\in T_j} \val((b_j^*, Q_j^*);\otherbid{t})}{\max_{\ibid\in\onefeasclass{\numbid}}\sum_{t=1}^T\val(\ibid; \otherbid{t})}\leq \frac{\sum_{j=1}^{\numbid'}\sum_{t\in T_j} \val((b_j^*, Q_j^*);\otherbid{t})}{\max_{\ibid\in\resonefeasclass{\numbid}}\sum_{t=1}^T\val(\ibid; \otherbid{t})}\,.
\end{align}

Consider any $\ibid\in\resonefeasclass{\numbid}$ obtained from the set $\mathcal{S}=\{k_1, \dots, k_\numbid\}\subseteq[\numbid']$. As $\ibid$ is feasible for $\hist$, invoking \cref{lem:sum-to-max-equivalence} for any round $t\in[T]$, we get:
\begin{align}\label{eq:m'-max-lemma}
    \val(\ibid; \otherbid{t}) = \max_{\l\in[\numbid]}\val((b_{k_\l}^*, Q_{k_\l}^*); \otherbid{t})= \max_{\l\in\mathcal{S}} \val((b_\l^*, Q_\l^*); \otherbid{t}),
\end{align}
because for any $\l\in[\numbid]$, $b_\l=b_{k_\l}^*$ and $Q_\l=\sum_{j=1}^jq_j=\sum_{j=1}^jQ_{k_j}^*-Q_{k_{j-1}}^*=Q_{k_\l}^*\,.$

Recall that $T_j$s, the set of rounds in which $b_j^*$ is the least winning bid, partition the set $[T]$. So, 
\begin{align}\label{eq:m'-LB}
    \sum_{t=1}^T\val(\ibid; \otherbid{t}) &= \sum_{t=1}^{\numbid'}\sum_{t\in T_j}\val(\ibid; \otherbid{t})\nonumber\\
    &=\sum_{j\in \mathcal{S}}\sum_{t\in T_j}\val(\ibid; \otherbid{t}) + \sum_{j\in [\numbid']\setminus \mathcal{S}}\sum_{t\in T_j}\val(\ibid; \otherbid{t}) \\
    &\ge \sum_{j\in \mathcal{S}}\sum_{t\in T_j}\val(\ibid; \otherbid{t})
\end{align}
For any $j\in\mathcal{S}$, consider any round $t\in T_j$: 
\begin{align}\label{eq:m-mbar-LB1}
  \val(\ibid; \otherbid{t}) &\stackrel{\eqref{eq:m'-max-lemma}}{=} \max_{\l\in\mathcal{S}}\val((b_\l^*, Q_\l^*);\otherbid{t}) \geq \val((b_j^*, Q_j^*);\otherbid{t})\nonumber\\
  \implies \sum_{t\in T_j}\val(\ibid; \otherbid{t})&\geq \sum_{t\in T_j}\val((b_j^*, Q_j^*);\otherbid{t})
\end{align}
which implies,
\begin{align}\label{eq:m-mbar-LB2}
   \sum_{t=1}^T\val(\ibid; \otherbid{t}) &\stackrel{\eqref{eq:m'-LB}}{\ge} \sum_{j\in \mathcal{S}}\sum_{t\in T_j}\val(\ibid; \otherbid{t}) \stackrel{\eqref{eq:m-mbar-LB1}}{\geq} \sum_{j\in \mathcal{S}}\sum_{t\in T_j}\val((b_j^*, Q_j^*);\otherbid{t})\,. 
\end{align}

Note that a bidding strategy $\ibid\in\resonefeasclass{\numbid}$ is uniquely determined by the set $\mathcal{S}\subseteq[\numbid']$ that generates $\ibid$. Hence, 
\begin{align}
    \max_{\ibid\in\resonefeasclass{\numbid}}\sum_{t=1}^T\val(\ibid; \otherbid{t}) &= \max_{\mathcal{S}\subseteq[\numbid']:|\mathcal{S}|= \numbid}\sum_{t=1}^T\val(\ibid; \otherbid{t})\nonumber\\
&\stackrel{\eqref{eq:m-mbar-LB2}}{\geq}\max_{\mathcal{S}\subseteq[\numbid']:|\mathcal{S}|= \numbid}\sum_{j\in \mathcal{S}}\sum_{t\in T_j}\val((b_j^*, Q_j^*);\otherbid{t})\nonumber\\
&\geq \frac{\numbid}{\numbid'}\sum_{j=1}^{\numbid'} \sum_{t\in T_j}\val((b_j^*, Q_j^*);\otherbid{t})\,.\label{eq:m'-LB-final}
\end{align}
Hence,  
\begin{align*}
   \Lambda_{\feasclass{\numbid'},\feasclass{\numbid}}(\hist, \v)\leq \frac{\sum_{j=1}^{\numbid'}\sum_{t\in T_j} \val((b_j^*, Q_j^*);\otherbid{t})}{\max_{\ibid\in\resonefeasclass{\numbid}}\sum_{t=1}^T\val(\ibid; \otherbid{t})}\stackrel{\eqref{eq:m'-LB-final}}{\leq} \frac{\numbid'}{\numbid}\,.
\end{align*}
Maximizing over all bid histories and valuation vectors, we get that $\Lambda_{\feasclass{\numbid'},\feasclass{\numbid}}\leq \frac{\numbid'}{\numbid}$.
    
\end{proof}
Now, we prove that the result also holds for safe bidding strategies. Formally, we show that for $\numbid, \numbid'\in\N$ and $\numbid'\geq \numbid$, $\Lambda_{\optufclass{\numbid'}, \optufclass{\numbid}}\leq \frac{\numbid'}{\numbid}$. To this end, we consider the quantity $\Lambda_{\optufclass{\numbid'}, \optufclass{\numbid}}(\hist, \v)$ analogous to $\Lambda_{\feasclass{\numbid'}, \feasclass{\numbid}}(\hist, \v)$, upper bound it and maximize over all bid histories and valuation vectors to get the desired result. 

For any bid history, $\hist=[\otherbid{t}]_{t\in[T]}$, let $\safebid{\numbid'}:=\argmax_{\ibid'\in\optufclass{\numbid'}}\sum_{t=1}^T\val(\ibid'; \otherbid{t})$. Without loss of generality, let $\safebid{\numbid'}=\langle (w_{Q_1^*}, q_1^*), \dots, (w_{Q_{\numbid'}^*}, q_{\numbid'}^*)\rangle$. The proof to bound $\Lambda_{\optufclass{\numbid'}, \optufclass{\numbid}}(\hist)$ is similar to that of \cref{thm:m-mbar-strong} till \cref{eq:m-mbar-UB-helper}, where we have
\begin{align*}
    \Lambda_{\optufclass{\numbid'}, \optufclass{\numbid}}(\hist) \leq \frac{\sum_{j=1}^{\numbid'}\sum_{t\in T_j} \val((w_{Q_j}^*, Q_j^*);\otherbid{t})}{\max_{\ibid\in\oneufclass{\numbid}}\sum_{t=1}^T\val(\ibid; \otherbid{t})}\,.
\end{align*}
Recall that $\oneufclass{\numbid}$ is the class of $\numbid$-uniform safe bidding strategies. Suppose $\oneufclass{\numbid}(\hist)$ be the class that is constructed analogous to the class $\resonefeasclass{\numbid}$ obtained in \cref{alg:construct-bidclass}. We now show that $\oneufclass{\numbid}(\hist)\subseteq\oneufclass{\numbid}$. 

\begin{lemma}\label{lem:new-class-1}
    For any subset $\mathcal{S}\subseteq[\numbid']$ such that $|\mathcal{S}|=\numbid$, suppose $\ibid$ is the bidding strategy obtained from $\mathcal{S}$ by applying \cref{alg:construct-bidclass}. Then, $\ibid$ is a safe bidding strategy. As $\mathcal{S}$ was chosen arbitrarily, $\oneufclass{\numbid}(\hist)\subseteq\oneufclass{\numbid}$.    
\end{lemma}

\begin{proof}
Let $\mathcal{S}$ contain $k_1<\dots<{k_{\numbid}}$ and $\ibid$ be constructed per \cref{alg:construct-bidclass}. Then, for any $j\in[\numbid]$, 
\begin{align*}
   b_j=b_{k_j}^* = w_{Q_{k_j}^*},\quad\text{and}\quad Q_j=Q_{k_j}^*\,.
\end{align*}
Thus, $\ibid$ is a safe bidding strategy.
\end{proof}
The rest of the proof follows similar to that of \cref{thm:m-mbar-strong} which gives the desired result.

\subsection{Proof of \cref{thm:mbar-non-safe}}\label{apx:thm:mbar-non-safe}

We need to establish upper bounds on $\Lambda_{\feasclass{\numbid'}, \optufclass{\numbid}}$. From \cref{thm:m-mbar-strong} and \cref{thm:Price_universal}, we have:
\begin{align*}
   \Lambda_{\feasclass{\numbid'}, \optufclass{\numbid}} &= \frac{\max_{\ibid'\in\feasclass{\numbid'}}\sum_{t=1}^T\val(\ibid'; \otherbid{t})}{\max_{\ibid\in\optufclass{\numbid}}\sum_{t=1}^T\val(\ibid; \otherbid{t})}\\
   &=\frac{\max_{\ibid'\in\feasclass{\numbid'}}\sum_{t=1}^T\val(\ibid'; \otherbid{t})}{\max_{\ibid''\in\feasclass{\numbid}}\sum_{t=1}^T\val(\ibid''; \otherbid{t})}\cdot\frac{\max_{\ibid''\in\feasclass{\numbid}}\sum_{t=1}^T\val(\ibid''; \otherbid{t})}{\max_{\ibid\in\optufclass{\numbid}}\sum_{t=1}^T\val(\ibid; \otherbid{t})}\\
   &=\Lambda_{\feasclass{\numbid'},\feasclass{\numbid}}\cdot\Lambda_{\feasclass{\numbid}, \optufclass{\numbid}} \leq \frac{2\numbid'}{\numbid}\,.
\end{align*}
\section{Tight Lower Bounds for Results in \cref{sec:learning-rich}}\label{apx:sec:LB}
\subsection{Tight Lower Bound for \cref{thm:Price_universal} (For $\numbid\geq 2$)}\label{apx:LB-2-m-gen}

In this section, we design a bid history $\hist$ and valuation vector, $\mathbf{v}$ such that for any $\delta\in(0, 1/2]$, $\Lambda_{\feasclass{\numbid}, \optufclass{\numbid}}(\hist, \v)\geq 2-\delta$. By definition, for any $\hist$ and $\numbid\geq 2$,
\begin{align}\label{eq:pouf-loose-bound}
    \Lambda_{\feasclass{\numbid}, \optufclass{\numbid}}(\hist, \v)= \frac{\optvalue{\numbid}}{\safevalue{\numbid}}\geq\frac{\optvalue{\numbid}}{\numbid \safevalue{1}}\,,
\end{align}
where the second inequality follows because $\Lambda_{\optufclass{\numbid}, \optufclass{1}}\leq \numbid$~(by \cref{thm:m-mbar}). Surprisingly, we show that instead of computing $\safevalue{\numbid}$ directly, computing the upper bound of $\numbid\safevalue{1}$ also gives a tight bound.


\subsubsection{Construction of $\hist$.} We first decide all the parameters. 
\begin{itemize}
    \item Fix $\numbid\geq 2$ and any integer $N\geq 2\ceil{\frac{1}{\delta}}$. 
    \item Let $\maxbid=N^{2\numbid-1}$. Consider $T=N^{2\numbid-1}$ rounds and $K=N^{2\numbid-1}+1$ units in each auction.
    \item Let $\epsilon'=\frac{\numbid\delta/(2\numbid-1)-1/N}{2(1-1/N)}<\frac{\delta}{2}\leq \frac{1}{4}$. Set $\epsilon$ such that $\epsilon'=\epsilon N^{2\numbid-1}(N^{2\numbid-1}+1)$. 
\end{itemize}
Consider a valuation vector $\mathbf{v}=[1, v, \cdots, v]$ such that $v=1-2\epsilon'$, and target RoI $\gamma=0$. Partition the $N^{2\numbid-1}$ rounds into $2\numbid$ partitions such that the first partition has $1$ round and the $j^{th}$ partition has $N^{j-1}-N^{j-2}$ rounds for $2\leq j\leq 2\numbid$. Each partition has identical competing bid profile submitted by other bidders. In particular,
\begin{enumerate}
    \item The first partition (containing one round) has all the bids submitted by others as $w_{N^{2\numbid-1}+1}+\epsilon$.
    \item If $j>1$ and $j$ is odd, for the $j^{th}$ partition~(of size $N^{j-1}-N^{j-2}$), the smallest $N^{2\numbid-j}+1$ winning competing bids are $w_{N^{2\numbid-j}+1}+\epsilon$ and the remaining bids are $C\gg w_1$. 
    \item If $j>1$ and $j$ is even, for the $j^{th}$ partition~(of size $N^{j-1}-N^{j-2}$), the smallest $N^{2\numbid-j}$ winning competing bids are $w_{N^{2\numbid-j}+1}+\epsilon$ and the remaining bids are $C\gg w_1$. 
\end{enumerate}
We present an example for such a bid history in \cref{tab:lb-gen-m}.
\begin{table}[!tbh]
    \centering
    \caption{Bid history achieving tight lower bound for $\numbid=2$. Each round in the same partition has identical competing bid profile. Total number of units in each auction is $K=N^3+1$.}
    \small
    \begin{tabular}{cccc}
    \toprule
 Partition 1& Partition 2& Partition 3&Partition 4\\
 
         $t=1$&  $t\in[2, N]$&  $t\in[N+1,N^2]$&  $t\in[N^2+1,N^3]$\\
         \midrule
         $0$ bids are $C$&  $N^3-N^2+1$ bids are $C$ &  $N^3-N$ bids are $C$&  $N^3$ bids are $C$\\
         $N^3+1$ bids are $w_{N^3+1}+\epsilon$&  $N^2$ bids are $w_{N^2+1}+\epsilon$ &  $N+1$ bids are $w_{N+1}+\epsilon$&  1 bid is $w_2+\epsilon$\\
         \bottomrule
    \end{tabular}
    
    \label{tab:lb-gen-m}
\end{table}
\normalsize

\subsubsection{Computing $\optbid{\numbid}$.} 

\begin{lemma}\label{lem:\numbid-bids-tight-2}
    For the aforementioned $\hist$, $\optbid{\numbid}=\langle(b_1,q_1), \dots, (b_\numbid, q_\numbid)\rangle$ where
\begin{align}
    (b_j, q_j) = \begin{cases}
        \left(1, N \right), &~\text{ if } j=1\\
        \left(w_{N^{2j-2}}, N^{2j-1}-N^{2j-3} \right), &~\text{ if } 2\leq j \leq \numbid\,.
    \end{cases}
\end{align}

Furthermore,
\begin{align*}
\optvalue{\numbid} =N^{2\numbid-1} + (2\numbid-1)(N^{2\numbid-1}-N^{2\numbid-2})v\,.    
\end{align*}
\end{lemma}

\begin{proof}
We begin by a crucial observation that the bid history does not allow obtaining more than $N^{2\numbid-j}$ units in the $j^{th}$ partition while satisfying the RoI constraint, \textit{irrespective of the number of bids} submitted by the bidder. To verify this, note that, the maximum number of units that can be allocated to any bidding strategy in the $j^{th}$ partition is either $N^{2\numbid-j}$ or $N^{2\numbid-j}+1$~(depending on if $j$ is even or odd). Suppose contrary to our claim, the bidder is allocated $N^{2\numbid-j}+1$ units in the some round $t$ in the $j^{th}$ partition by bidding some $\ibid$. Let $\allbids^t=(\ibid, \otherbid{t})$ be the complete bid profile. So, $p(\allbids^t)\geq w_{N^{2\numbid-j}+1}+\epsilon$ but $x(\allbids^t)=N^{2\numbid-j}+1$ indicating that the RoI constraint is violated, which verifies our claim.

So, the total number of units, $N_{\text{total}}$, that can be obtained by the bidder over all the rounds is:
    \begin{align*}
        N_{\text{total}} \leq  N^{2\numbid-1}+\sum_{j=2}^{2\numbid}N^{2\numbid-j}(N^{j-1}-N^{j-2})=2\numbid N^{2\numbid-1}-(2\numbid-1)N^{2\numbid-2}=:N_{\max}.
    \end{align*}

 Now, we compute the the number of units obtained by bidding $\optbid{\numbid}$ and show that it is allocated $N_{\max}$ units for the constructed bid history, demonstrating that it is the optimal bidding strategy.

 Consider any auction in the $j^{th}$ partition. The lowest winning bid in the bid profile is $w_{N^{2\numbid-j}+1}+\epsilon$. Note that the unique bid values~(ignoring the quantity for the sake of brevity) in $\optbid{\numbid}$ are $\ibid=\left\{1, w_{N^2}, \dots, w_{N^{2\numbid-2}} \right\}$. We claim that the winning bid values of $\optbid{\numbid}$ in the $j^{th}$ partition are $\widehat{\ibid}=\left\{1, w_{N^2}, \dots, w_{N^{2\numbid+2\floor{-j/2}}} \right\}$. This is true because the least bid value in $\widehat{\ibid}$ is greater than $w_{N^{\numbid-j}+1}+\epsilon$, i.e., 
 \begin{align*}
    w_{N^{2\numbid+2\floor{-j/2}}} - (w_{N^{2\numbid-j}+1}+\epsilon)&\geq w_{N^{2\numbid-j}} - (w_{N^{2\numbid-j}+1}+\epsilon)\\
    &=\frac{1-v}{N^{2\numbid-j}(N^{2\numbid-j}+1)}-\epsilon=\frac{2\epsilon N^{2\numbid-1}(N^{2\numbid-1}+1)}{N^{2\numbid-j}(N^{2\numbid-j}+1)}-\epsilon\geq \epsilon>0\,.
 \end{align*}
 
Let $N_j$ denote the number of units allocated to $\optbid{\numbid}$ in each auction in the $j^{th}$ partition. There are two cases:
 \begin{enumerate}
     \item [(a)] for $j$ odd, recall that 
    for the $j^{th}$ partition~(of size $N^{j-1}-N^{j-2}$), the smallest $N^{2\numbid-j}+1$ winning competing bids are $w_{N^{2\numbid-j}+1}+\epsilon$ and the remaining bids are $C\gg w_1$. Then,  
     \begin{align*}
     N_j = N+\sum_{\l=2}^{\numbid-\frac{j-1}{2}}(N^{2\l-1}-N^{2\l-3}) = N^{2\numbid-j}\,.
 \end{align*}
 \item [(b)] For $j$ even, recall that  for the $j^{th}$ partition~(of size $N^{j-1}-N^{j-2}$), the smallest $N^{2\numbid-j}$ winning competing bids are $w_{N^{2\numbid-j}+1}+\epsilon$ and the remaining bids are $C\gg w_1$. 
 \begin{align*}
     N_j = \min\left\{N^{2\numbid-j}, N+\sum_{\l=2}^{\numbid+1-\frac{j}{2}}(N^{2\l-1}-N^{2\l-3}) \right\}=\min\{N^{2\numbid-j}, N^{2\numbid-j+1}\}= N^{2\numbid-j}\,.
 \end{align*}
 \end{enumerate}
  Here, the minimum is taken over two quantities where the first quantity is the number of competing bids less than $C$ in any round $t$ in the $j^{th}$ partition and the second quantity represents the total demand of the winning bids in $\optbid{\numbid}$ for that round. So, the total number of units obtained across all rounds is
 \begin{align*}
     N^{2\numbid-1}+\sum_{j=2}^{2\numbid}N^{2\numbid-j}(N^{j-1}-N^{j-2})=2\numbid N^{2\numbid-1}-(2\numbid-1)N^{2\numbid-2}.
 \end{align*}
 As this is the maximum number of units that can be obtained by the bidder, $\optbid{\numbid}$ is optimal. The total value obtained by bidding $\optbid{\numbid}$ is
 \begin{align*}
     \optvalue{\numbid} &= 1+(N^{2\numbid-1}-1)v+\sum_{j=2}^{2\numbid}  (N^{j-1}-N^{j-2})(1+(N^{2\numbid-j}-1)v)\\
     &=N^{2\numbid-1} + (2\numbid-1)(N^{2\numbid-1}-N^{2\numbid-2})v\,.
\end{align*}
\end{proof}

\subsubsection{Computing $\safevalue{1}$} Recall that we compute $\safevalue{1}$ and invoke the bounds on $\Lambda_{\optufclass{\numbid}, \optufclass{1}}$, instead of directly evaluating $\safevalue{\numbid}$.

\begin{lemma}\label{lem:1-bid-optimal-2}
  For the aforementioned $\hist$, $\safebid{1}=(1, 1)$ and $\safevalue{1}=N^{2\numbid-1}$.  
\end{lemma}

\begin{proof}
   The basic idea is to enumerate the total units~(value) that can be obtained by bidding $(w_q, q)$ for $q\in[N^{2\numbid-1}]$ and then finding the maximum of those values. As $q$ can be exponential in $\numbid$, we exploit the structure of the bid history to compute the objective in an efficient manner. 
   
   Suppose $q=1$. The maximum number of units~(value) that can be obtained by bidding $(1, 1)$ is trivially $N^{2\numbid-1}$, So, $(1, 1)$ obtains a total value $N^{2\numbid-1}$.
   
   Suppose $q\geq 2$. Furthermore, assume $N^{2\numbid-j}<q\leq N^{2\numbid-j+1}$, for some $2\leq j\leq 2\numbid$. Consider any bid of the form $(w_q, q)$. Bidding $(w_q, q)$ does not obtain any units in the partitions indexed by $j, j+1, \dots, 2\numbid$ as $w_q$ is strictly less than the least winning competing bids in those partitions, i.e., $w_q \leq w_{N^{2\numbid-j}+1} < w_{N^{2\numbid-\l}+1}+\epsilon$, for any $\l\in\{j, j+1,\dots, 2\numbid\}$. So, if $N^{2\numbid-j}<q\leq N^{2\numbid-j+1}$, $(w_q, q)$ gets no units in
   \begin{align*}
       \sum_{\l=j}^{2\numbid} N^{\l-1}-N^{\l-2} = N^{2\numbid-1}-N^{j-2} \text{ auctions.}
   \end{align*}
   
   In the remaining $N^{j-2}$ auctions it can win at most $q$ units. So, the maximum value obtained by $(w_q, q)$ for any $q\geq 2$ is
   \begin{align*}
       N^{j-2}(1+(q-1)v) &\leq N^{j-2}(1+(N^{2\numbid-j+1}-1)v)=N^{j-2}+(N^{2\numbid-1}-N^{j-2})v<N^{2\numbid-1},
   \end{align*}
   
where the last inequality holds as $v<1$. Hence, $\safebid{1}=(1, 1)$ and $\safevalue{1}=N^{2\numbid-1}$. 
\end{proof}
Hence, from \cref{lem:\numbid-bids-tight-2}, \cref{lem:1-bid-optimal-2} and \cref{eq:pouf-loose-bound}
\begin{align*}
     \Lambda_{\feasclass{\numbid}, \optufclass{\numbid}}(\hist, \v)\geq\frac{\optvalue{\numbid}}{\numbid \safevalue{1}}&=\frac{N^{2\numbid-1} + (2\numbid-1)(N^{2\numbid-1}-N^{2\numbid-2})v}{\numbid N^{2\numbid-1}}\\
    &=\frac{2\numbid N^{2\numbid-1}-(2\numbid-1)N^{2\numbid-2} - 2\epsilon'(2\numbid-1)(N^{2\numbid-1}-N^{2\numbid-2})}{\numbid N^{2\numbid-1}}\\
    &=2-\frac{2\numbid-1}{\numbid}\Big(\frac{1}{N}+2\epsilon'\Big(1-\frac{1}{N}\Big)\Big)\\
    &=2-\delta\,,
\end{align*}
where the last step follows by substituting the value of $\epsilon'$ and thus concludes the proof.

\subsection{Tight Lower Bound for \cref{thm:m-mbar}}\label{apx:LB-m}
In this section, we construct a bid history $\hist$ and valuation vector, $\mathbf{v}$ such that for any $\delta\in(0, \frac{1}{2}]$, $\Lambda_{\optufclass{\numbid'}, \optufclass{\numbid}}(\hist, \v)\geq \frac{\numbid'}{\numbid}-\delta$, i.e., the upper bound is tight for $\numbid'\geq \numbid+1$~(If $\numbid'=\numbid$, $\Lambda_{\optufclass{\numbid'}, \optufclass{\numbid}}=1$ implying the bound is tight).  

Recall that \cref{thm:m-mbar} states that $\Lambda_{\optufclass{\numbid}, \optufclass{1}}\leq \numbid$ for any $\numbid\geq1$. So, 
\begin{align}\label{eq:m-approx-trick}
   \Lambda_{\optufclass{\numbid'}, \optufclass{\numbid}}(\hist, \v)= \frac{\safevalue{\numbid'}}{\safevalue{\numbid}}\geq\frac{\safevalue{\numbid'}}{\numbid \safevalue{1}}=\frac{1}{\numbid}\cdot\Lambda_{\optufclass{\numbid'}, \optufclass{1}}(\hist, \v)\,. 
\end{align}
We now establish a lower bound on $\Lambda_{\optufclass{\numbid'}, \optufclass{1}}(\hist, \v)$.
\subsubsection{Construction of $\hist$.} We first decide all the parameters. 
\begin{itemize}
    \item Fix $\numbid'\geq 2$ and any integer $N\geq \ceil{\frac{\numbid'}{\delta}}$. 
    \item Let $\maxbid=N^{\numbid'-1}$. Consider $T=N^{\numbid'-1}$ rounds and $K=N^{\numbid'-1}$ units in each auction.
    \item Let $\epsilon'=\frac{\numbid\delta/(\numbid'-1)-1/N}{2(1-1/N)}<\frac{\numbid\delta}{2(\numbid'-1)}\leq \frac{1}{4}$. Set $\epsilon$ such that $\epsilon'=\epsilon N^{\numbid'-1}(N^{\numbid'-1}+1)$. 
\end{itemize}

 Let the valuation vector be $\mathbf{v}=[1, v, \cdots, v]$ such that $v=1-2\epsilon'$, and target RoI, $\gamma=0$.  Partition the $N^{\numbid'-1}$ rounds into $\numbid'$ partitions such that the first partition has $1$ round and the $j^{th}$ partition has $N^{j-1}-N^{j-2}$ rounds for $2\leq j\leq \numbid'$. Each partition has identical competing bid profile submitted by the other bidders. In particular,
\begin{enumerate}
    \item The first partition (containing one round) has all the competing bids to be $w_{N^{\numbid'-1}+1}+\epsilon$.
    \item For $2\leq j\leq \numbid'$, the $j^{th}$ partition~(of size $N^{j-1}-N^{j-2}$); the smallest $N^{\numbid'-j}+1$ competing winning bids are $w_{N^{\numbid'-j}+1}+\epsilon$ and the remaining bids are $C\gg w_1$. 
\end{enumerate}
We present an example for such a bid history in \cref{tab:LB-m}.
\begin{table}[!tbh]
    \centering
    \caption{Bid history attaining tight lower bound for $\numbid'=4$~(and $\numbid=1$). Each round in the same partition has identical competing bid profile. Total number of units in each auction is $K=N^3$.}
    \small
    \begin{tabular}{cccc}
    \toprule
 Partition 1& Partition 2& Partition 3&Partition 4\\
 
         $t=1$&  $t\in[2, N]$&  $t\in[N+1,N^2]$&  $t\in[N^2+1,N^3]$\\
         \midrule
         $0$ bids are $C$&  $N^3-N^2-1$ bids are $C$ &  $N^3-N-1$ bids are $C$&  $N^3-2$ bids are $C$\\
         $N^3$ bids are $w_{N^3+1}+\epsilon$&  $N^2+1$ bids are $w_{N^2+1}+\epsilon$ &  $N+1$ bids are $w_{N+1}+\epsilon$&  2 bids are $w_2+\epsilon$\\
         \bottomrule
    \end{tabular}
    
    \label{tab:LB-m}
\end{table}
\normalsize
\subsubsection{Computing $\safebid{\numbid'}$} 

\begin{lemma}\label{lem:numbid'-bids-tight}
    For the aforementioned $\hist$, $\safebid{\numbid'}=\langle(b_1, q_1),\dots, (b_{\numbid'}, q_{\numbid'})\rangle$ where
\begin{align}\label{eq:tight_example}
    (b_j, q_j) = \begin{cases}
        \left(1, 1\right), &~\text{if } j=1\\
        \left(w_{N^{j-1}}, N^{j-1}-N^{j-2}\right), &~\text{if } 2\leq j\leq \numbid'
    \end{cases}\,.
\end{align}

Furthermore,
\begin{align*}
\safevalue{\numbid'}=N^{\numbid'-1} + (\numbid'-1)(N^{\numbid'-1}-N^{\numbid'-2})v \,.   
\end{align*}
\end{lemma}
\begin{proof}
We begin by a crucial observation that the bid history does not allow obtaining more than $N^{\numbid'-j}$ units in the $j^{th}$ partition, while satisfying the RoI constraint,  \textit{irrespective of the number of bids} submitted by the bidder. To verify this, suppose contrary to our claim, that the bidder is allocated $N^{\numbid'-j}+1$ units in the some round $t$ in the $j^{th}$ partition by bidding some $\ibid$. Let $\allbids^t=(\ibid, \otherbid{t})$ be the complete bid profile. Hence, $p(\allbids^t)\geq w_{N^{\numbid'-j}+1}+\epsilon$ but $x(\allbids^t)=N^{\numbid'-j}+1$ indicating that the RoI constraint is violated which verifies our claim.



So, the total number of units, $N_{\text{total}}$, that can be obtained by the bidder across all the rounds is:
    \begin{align*}
        N_{\text{total}} \leq  N^{\numbid'-1}+\sum_{j=2}^{\numbid'}N^{\numbid'-j}(N^{j-1}-N^{j-2})=\numbid' N^{\numbid'-1}-(\numbid'-1)N^{\numbid'-2}:=N_{\max}.
    \end{align*}
    
 Now, we compute the number of units obtained by bidding $\safebid{\numbid'}$ and show that it is allocated $N_{\max}$ units across all the auctions, implying that it is the optimal bidding strategy. 
 
 To show $\safebid{\numbid'}$ is optimal, consider any auction in the $j^{th}$ partition. The lowest winning competing bid is $w_{N^{\numbid'-j}+1}+\epsilon$. Note that the unique bid values~(ignoring the quantity for the sake of brevity) in $\safebid{\numbid'}$, provided in \cref{eq:tight_example},  are $\ibid=\left\{1, w_N, \dots, w_{N^{\numbid'-1}}\right\}$. We claim that the winning bid values of $\safebid{\numbid'}$ in the $j^{th}$ partition are $\widetilde{\ibid}=\left\{1, w_N, \dots, w_{N^{\numbid'-j}} \right\}$. Again recall that for $2\leq j\leq \numbid'$, for the $j^{th}$ partition~(of size $N^{j-1}-N^{j-2}$),
the smallest $N^{\numbid'-j}+1$ competing winning bids are $w_{N^{\numbid'-j}+1}+\epsilon$ and the remaining bids are $C\gg w_1$. And here, the least bid value in $\widetilde{\ibid}$ is greater than $w_{N^{\numbid'-j}+1}+\epsilon$, i.e., 
 \begin{align*}
    w_{N^{\numbid'-j}} - (w_{N^{\numbid'-j}+1}+\epsilon)=\frac{1-v}{N^{\numbid'-j}(N^{\numbid'-j}+1)}-\epsilon=\frac{2\epsilon N^{\numbid'-1}(N^{\numbid'-1}+1)}{N^{\numbid'-j}(N^{\numbid'-j}+1)}-\epsilon\geq \epsilon>0\,.
 \end{align*}

 Moreover, observe that the bidder is allocated the maximum number of units demanded for each of the bid value in $\widetilde{\ibid}$. So, the number of units in each auction in the $j^{th}$ partition by bidding $\safebid{\numbid'}$ is
 \begin{align*}
     N_j = 1+\sum_{\l=2}^{\numbid'-j+1}(N^{\l-1}-N^{\l-2}) = N^{\numbid'-j}\,.
 \end{align*}
 Hence, the total number of units obtained across all rounds is
 \begin{align*}
    N^{\numbid'-1}+\sum_{j=2}^{\numbid'}N^{\numbid'-j}(N^{j-1}-N^{j-2})=\numbid' N^{\numbid'-1}-(\numbid'-1)N^{\numbid'-2} = N_{\max}.
 \end{align*}
 As this is the maximum number of units that can be obtained by the bidder, $\safebid{\numbid'}$ is optimal. The total value obtained by bidding $\safebid{\numbid'}$ is given by
 \begin{align*}
     \safevalue{\numbid'} &= 1+(N^{\numbid'-1}-1)v+\sum_{j=2}^{\numbid'}  (N^{j-1}-N^{j-2})(1+(N^{\numbid'-j}-1)v)\\
     &=N^{\numbid'-1} + (\numbid'-1)(N^{\numbid'-1}-N^{\numbid'-2})v\,.
\end{align*}
\end{proof}

\subsubsection{Computing $\safebid{1}$} 

\begin{lemma}\label{lem:1-bid-optimal}
  For the constructed $\hist$, $\safebid{1}=\left(1, 1\right)$ and $\safevalue{1}=N^{\numbid-1}$.  
\end{lemma}
\begin{proof}
    The proof is similar to that of \cref{lem:1-bid-optimal-2} in the sense that it involves enumerating the total units obtaining by submitting $(w_q, q)$ for all $q\in[N^{\numbid-1}]$ in an efficient manner by leveraging the structure in the bid history.
\end{proof}
Substituting values from \cref{lem:numbid'-bids-tight}, \cref{lem:1-bid-optimal}
\begin{align*}
\Lambda_{\optufclass{\numbid'}, \optufclass{\numbid}}(\hist, \v)&\stackrel{\eqref{eq:m-approx-trick}}{\geq} \frac{1}{\numbid}\cdot\Lambda_{\optufclass{\numbid'}, \optufclass{1}}(\hist, \v)\\
&=\frac{\safevalue{\numbid'}}{\numbid\safevalue{1}}\\
&=\frac{N^{\numbid'-1} + (\numbid'-1)(N^{\numbid'-1}-N^{\numbid'-2})v}{\numbid N^{\numbid'-1}}\\
    &=\frac{\numbid' N^{\numbid'-1}-(\numbid'-1)N^{\numbid'-2} - 2\epsilon'(\numbid'-1)(N^{\numbid'-1}-N^{\numbid'-2})}{\numbid N^{\numbid'-1}}\\
    &=\frac{\numbid'}{\numbid}-\frac{(\numbid'-1)}{\numbid}\Big(\frac{1}{N}+2\epsilon'\Big(1-\frac{1}{N}\Big)\Big)\\
    &= \frac{\numbid'}{\numbid}-\delta\,,
\end{align*}
where the last step follows by substituting the value of $\epsilon'$ and thus concludes the proof. 

\begin{corollary}
    As $\optufclass{\numbid}\subseteq\feasclass{\numbid}$ for all $\numbid\in\N$, the upper bound in \cref{thm:m-mbar-strong} is also tight, i.e., for any $\numbid'\geq\numbid$ and $\delta\in(0, \frac{1}{2}]$, there exists a bid history $\hist$ and valuation curve $\v$ such that $\Lambda_{\feasclass{\numbid'}, \feasclass{\numbid}}(\hist, \v)\geq \frac{\numbid'}{\numbid}-\delta$.
\end{corollary}

\subsection{Tight Lower Bound for \cref{thm:mbar-non-safe}}\label{apx:LB-2m}

In this section, we present a bid history $\hist$ and valuation vector, $\mathbf{v}$ such that for any $\delta\in(0, \frac{1}{2}]$, $\Lambda_{\feasclass{\numbid'}, \optufclass{\numbid}}(\hist, \v)\geq\frac{2\numbid'}{\numbid}-\delta$, i.e., the upper bound is tight for $\numbid'\geq \numbid$. 

Recall that \cref{thm:m-mbar} states that $\Lambda_{\optufclass{\numbid}, \optufclass{1}}\leq \numbid$ for any $\numbid\geq1$. So, 
\begin{align}\label{eq:m-approx-trick-2}
   \Lambda_{\feasclass{\numbid'}, \optufclass{\numbid}}(\hist, \v)= \frac{\optvalue{\numbid'}}{\safevalue{\numbid}}\geq\frac{\optvalue{\numbid'}}{\numbid \safevalue{1}}=\frac{1}{\numbid}\cdot\Lambda_{\feasclass{\numbid'}, \optufclass{1}}(\hist, \v)\,. 
\end{align}
Recall that for showing tightness of the bound for \cref{thm:Price_universal} for any $\numbid\geq2$, we computed a lower bound on $\Lambda_{\feasclass{\numbid}, \optufclass{1}}(\hist, \v)$~(cf. \cref{eq:pouf-loose-bound}). 
Thus, we consider a valuation curve and a bid history that has a structure identical to the one presented in \cref{apx:LB-2-m-gen} but with the following modified parameters, 
\begin{itemize}
    \item Fix $\numbid'\geq 1$ and any integer $N\geq 2\ceil{\frac{2\numbid'}{\delta}}$. 
    \item Let $\maxbid=N^{2\numbid'-1}$. Consider $T=N^{2\numbid'-1}$ rounds and $K=N^{2\numbid'-1}+1$ units in each auction.
    \item Let $\epsilon'=\frac{\numbid\delta/(2\numbid'-1)-1/N}{2(1-1/N)}<\frac{\numbid\delta}{2(2\numbid'-1)}\leq \frac{1}{4}$. Set $\epsilon$ such that $\epsilon'=\epsilon N^{2\numbid'-1}(N^{2\numbid'-1}+1)$. 
\end{itemize}

Consider a valuation vector $\mathbf{v}=[1, v, \cdots, v]$ such that $v=1-2\epsilon'$, and target RoI $\gamma=0$. Substituting the values from \cref{lem:\numbid-bids-tight-2} and \cref{lem:1-bid-optimal-2},
\begin{align*}
    \Lambda_{\feasclass{\numbid'}, \optufclass{\numbid}}(\hist, \v)&\stackrel{\eqref{eq:m-approx-trick-2}}{\geq} \frac{1}{\numbid}\cdot \Lambda_{\feasclass{\numbid'}, \optufclass{1}}(\hist, \v)\\
    &=\frac{N^{2\numbid'-1} + (2\numbid'-1)(N^{2\numbid'-1}-N^{2\numbid'-2})v}{\numbid N^{2\numbid'-1}}\\
    &=\frac{2\numbid' N^{2\numbid'-1}-(2\numbid'-1)N^{2\numbid'-2} - 2\epsilon'(2\numbid'-1)(N^{2\numbid'-1}-N^{2\numbid'-2})}{\numbid N^{2\numbid'-1}}\\
    &=\frac{2\numbid'}{\numbid}-\frac{(2\numbid'-1)}{\numbid}\Big(\frac{1}{N}+2\epsilon'\Big(1-\frac{1}{N}\Big)\Big)\\
    &=\frac{2\numbid'}{\numbid} -\delta\,,
\end{align*}
where the last step follows by substituting the value of $\epsilon'$ and thus concludes the proof.
\section{Omitted Details from \cref{sec:extensions}}\label{apx:sec:extensions}
\subsection{Proof of \cref{lem:forgetful}}
We introduce a few additional notations to make the proof easier to parse. Specifically, we denote random variables with a hat, e.g. $\widehat{\ibid}^t$, and their realizations without hats, e.g. $\ibid^t$. Let $\E_t[\cdot]$ denote the expectation taken only with respect to randomness in selecting $\widehat{\ibid}^t$. For an adaptive adversary and any fixed $\ibid\in\optufclass{\numbid}$,
\begin{align}\label{eq:adaptive-UB}
    \sum_{t=1}^T\val(\ibid; \widehat{\otherbid{t}}) - \sum_{t=1}^T\E_t\Big[\val(\widehat{\ibid}^t; \widehat{\otherbid{t}})\Big]&\leq \sum_{t=1}^T\sup_{\otherbid{t}}\E_t\Big[\val(\ibid; \otherbid{t}) - \val(\widehat{\ibid}^t; \otherbid{t})\Big]\nonumber\\
    &=\sup_{\otherbid{1}, \dots, \otherbid{T}}\sum_{t=1}^T\E_t\Big[\val(\ibid; \otherbid{t}) - \val(\widehat{\ibid}^t; \otherbid{t})\Big]\,.
\end{align}
To see why \cref{eq:adaptive-UB} holds, let $T=2$. Then, the right hand side in the first equation becomes
\begin{align*}
   \sup_{\otherbid{1}}\E_1\Big[\val(\ibid; \otherbid{1}) - \val(\widehat{\ibid}^1; \otherbid{1})\Big] + \sup_{\otherbid{2}}\E_2\Big[\val(\ibid; \otherbid{2}) - \val(\widehat{\ibid}^2; \otherbid{2})\Big] 
\end{align*}
Note that $\E_1\Big[\val(\ibid; \otherbid{1}) - \val(\widehat{\ibid}^1; \otherbid{1})\Big]$ does not depend on $\otherbid{2}$ and $\otherbid{2}$, due to the forgetful property of the learning algorithm, does not depend on the realization of $\widehat{\otherbid{1}}$. Hence, 
\begin{align*}
   &\sup_{\otherbid{1}}\E_1\Big[\val(\ibid; \otherbid{1}) - \val(\widehat{\ibid}^1; \otherbid{1})\Big] + \sup_{\otherbid{2}}\E_2\Big[\val(\ibid; \otherbid{2}) - \val(\widehat{\ibid}^2; \otherbid{2})\Big] \\
   &= \sup_{\otherbid{1}, \otherbid{2}}\Big\{\E_1\Big[\val(\ibid; \otherbid{1}) - \val(\widehat{\ibid}^1; \otherbid{1})\Big] + \E_2\Big[\val(\ibid; \otherbid{2}) - \val(\widehat{\ibid}^2; \otherbid{2})\Big]\Big\}
\end{align*}
Repeating the same argument for any $T\geq 2$, we get the equality in \cref{eq:adaptive-UB}. Hence,
\begin{align}
   \max_{\ibid\in\optufclass{\numbid}}\sum_{t=1}^T\val(\ibid; \widehat{\otherbid{t}}) - \sum_{t=1}^T\E_t\Big[\val(\widehat{\ibid}^t; \widehat{\otherbid{t}})\Big]&\leq \max_{\ibid\in\optufclass{\numbid}}\sup_{\otherbid{1}, \dots, \otherbid{T}}\Big\{\sum_{t=1}^T\val(\ibid; \otherbid{t}) - \sum_{t=1}^T\E_t[\val(\widehat{\ibid}^t; \otherbid{t})]\Big\}\nonumber \\
   &=\sup_{\otherbid{1}, \dots, \otherbid{T}}\max_{\ibid\in\optufclass{\numbid}}\Big\{\sum_{t=1}^T\val(\ibid; \otherbid{t}) - \sum_{t=1}^T\E_t[\val(\widehat{\ibid}^t; \otherbid{t})]\Big\}\leq B,\nonumber   
\end{align}
where the last inequality holds as $\textsf{REG}\leq B$ for any sequence of competing bids $\otherbid{1}, \dots, \otherbid{T}$ by an \textit{oblivious} adversary. Taking expectations on both sides gives $\textsf{REG}_{\text{adap}}\leq B$.

\subsection{Proof of \cref{lem:whp-to-expected}}
We have that $\widehat{\textsf{REG}}_{\text{adap}} \le B(\delta)$ w. p. at least $1-\delta$. Then, 
\begin{align*}
   \textsf{REG}_{\text{adap}}=\E[\widehat{\textsf{REG}}_{\text{adap}}] &= \E\left[\widehat{\textsf{REG}}_{\text{adap}}\cdot\ind{\widehat{\textsf{REG}}_{\text{adap}} \geq B(\delta)}\right] + \E\left[\widehat{\textsf{REG}}_{\text{adap}}\cdot\ind{\widehat{\textsf{REG}}_{\text{adap}} \leq B(\delta)}\right]\\
   &\leq \maxbid T \cdot\P[\widehat{\textsf{REG}}_{\text{adap}} \geq B(\delta)] + B(\delta) \leq \maxbid T\delta + B(\delta)\,,
\end{align*}
where the first inequality holds as $\textsf{REG}_{\text{adap}}\leq \maxbid T$ w.p. 1. 

\subsection{Proof of \cref{thm:bandit-adaptive}}\label{apx:bandit-feedback-adaptive}

\textbf{The Algorithm}. Before proving \cref{thm:bandit-adaptive}, we present \cref{alg:weight-pushing-adaptive-whp}, which is designed for this setting. As stated earlier, our approach builds on ideas from \cite{gyorgy2007line}. We first define a covering path set $\mathcal{C}$ as a collection of $s$-$d$ paths $\path$ such that every edge $e \in E$ appears in at least one path $\path \in \mathcal{C}$. Such a set can be constructed with size at most $|E|$ by including, for each edge $e$, at least one path that contains $e$. Hence, $|\mathcal{C}| \leq |E|$.

In each round $t$, similar to \cref{alg:weight-pushing}, the bidder constructs a DAG $\mathcal{G}^t(V, E)$ without edge weights. With probability $\lambda > 0$, the bidder selects an $s$-$d$ path uniformly at random from the set $\mathcal{C}$. Otherwise, they update edge weights and sample a path as a Markov chain, following \cref{alg:weight-pushing}. Once a path is selected, the bidder maps it to a safe bidding strategy with at most $\numbid$ bid-quantity pairs.

Since the true weights cannot be computed $\forall e\in E$ under the bandit feedback structure, we therefore use the following \textit{biased} estimator $\widehat{\textsf{w}}^t(e)$ of $\textsf{w}^t(e)$: 
\begin{align}\label{eq:w-hat-adaptive}
        \widehat{\textsf{w}}^t(e) &= \frac{\textsf{w}^t(e)}{p^t(e)}\cdot\ind{e\in\path^t} + \frac{\theta}{p^t(e)}\quad\text{where}\quad p^t(e)=\sum_{\path:e\in\path}\P^t(\path)\,.
    \end{align}
    
   Here, $\textsf{w}^t(e)$ is as per \cref{eq:full-info-edge-weight} and $\theta \in [0, \maxbid]$ is a parameter that will be determined shortly. As before, $p^t(e)$ denotes the probability of selecting edge $e$ in round $t$, and $\P^t(\cdot)$ denotes the distribution over all $s$–$d$ paths in the graph $\mathcal{G}^t(V, E)$. Following the description of the algorithm, the distribution $\P^t(\cdot)$ can be expressed as a convex combination of two distributions: the uniform distribution over the paths in the covering path set $\mathcal{C}$, given by $\frac{\ind{\cdot \in \mathcal{C}}}{|\mathcal{C}|}$, and $\widehat{\P}^t(\cdot)$, which is the distribution over $s$–$d$ paths in round $t$ obtained using the path selection step from \cref{alg:weight-pushing}. Formally,
    \begin{align}\label{eq:def-P}
       \P^t(\path)&:= (1-\lambda) \cdot                                  \widehat{\P}^t(\path) + \lambda\cdot\frac{\ind{\path\in\mathcal{C}}}{|\mathcal{C}|},
    \quad\text{and}\quad\widehat{\P}^t(\path)=\prod_{e\in\path} \varphi^t(e) = \frac{\exp(\eta\sum_{\tau=1}^{t-1}\widehat{\textsf{w}}^{\tau}(\path) )}{\sum_{\path'}\exp(\eta\sum_{\tau=1}^{t-1}\widehat{\textsf{w}}^{\tau}(\path'))}\,.
    \end{align}

    Here, the last equality holds because the $s$-$d$ path distribution induced by the path selection step \cref{alg:weight-pushing} matches that of the Hedge algorithm over $s$-$d$ paths (see proof of \cref{thm:full-info}). The estimator, $\widehat{\textsf{w}}^t(e)$, is well defined as 
    \begin{align*}
       p^t(e)\stackrel{\eqref{eq:w-hat-adaptive},\eqref{eq:def-P}}{\geq} \lambda\sum_{\path:e\in\path} \frac{\ind{\path\in\mathcal{C}}}{|\mathcal{C}|}\geq \frac{\lambda}{|\mathcal{C}|} >0,
    \end{align*}
    
    where the second inequality follows as every edge $e$ is contained in at least one path $\path\in\mathcal{C}$.

\begin{algorithm}[!tbh]
\caption{Learning Safe Bidding Strategies~(Bandit Setting, Adaptive Adversary)}
\label{alg:weight-pushing-adaptive-whp}
\small{
\begin{algorithmic}[1]
\Require valuation curve $\mathbf{v}$, time horizon $T$, learning rate $\eta>0$, parameters $\lambda$, $\theta$, covering path set $\mathcal{C}$. Initialize $\varphi^0(e)=1$ and $\widehat{\textsf{w}}^0(e)=0, \forall e\in E$.
\For{$t = 1, 2, \dots, T$}
    \State Construct $\mathcal{G}^t(V, E)$ similar to $\mathcal{G}(V, E)$ without weights.
    \State Sample $Z_t\sim\text{Unif}[0, 1]$. 
    \If{$Z_t\leq \lambda$}
    \State Select path $\path^t$ uniformly at random from the set $\mathcal{C}$. 
    \Else
    \State Set $\Gamma^{t-1}(d)=1$ and recursively compute in bottom-to-top fashion for every node $u$ in $\mathcal{G}^t(V, E)$:
    \begin{align*}
        \Gamma^{t-1}(u)=\sum_{v:u\to v=e\ni E}\Gamma^{t-1}(v)\cdot\varphi^{t-1}(e)\cdot\exp(\eta \widehat{\textsf{w}}^{t-1}(e))\,.
    \end{align*}

    \State For edge $e=u\to v$ in $\mathcal{G}^t(V, E)$, update edge probability:
       \begin{align*}
           \varphi^t(e)=\varphi^{t-1}(e)\cdot\exp(\eta \widehat{\textsf{w}}^{t-1}(e))\cdot\frac{\Gamma^{t-1}(v)}{\Gamma^{t-1}(u)}.
       \end{align*}
    
    \quad Define initial node $u=s$ and path $\path^t=s$. 
    \While{$u\neq d$}
    \State Sample $v$ with probability $\varphi^t(u\to v)$.
    \State Append $v$ to the path $\path^t$; set $u\gets v$.
    \EndWhile
    \EndIf
    \State If $\path^t=s\to (1, z_1)\to\dots\to(k, z_k)\to d$ for some $k\in[\numbid]$, submit $\ibid^t=\langle(b_1, q_1), \dots, (b_k, q_k)\rangle$ where
    \begin{align*}
      b_\l=w_{z_\l} \quad\text{and}\quad q_\l=z_\l-z_{\l-1}, ~\forall  \l\in [k]\,.
    \end{align*}
    \State The bidder sets edge weights $\widehat{\textsf{w}}^t(e)$ per \cref{eq:w-hat-adaptive}.
    
\EndFor
\end{algorithmic}}
\end{algorithm}

\textbf{Regret Analysis.} We begin by stating the following result which is a modification of \cite[Lemma 2]{gyorgy2007line}.

\begin{lemma}\label{lem:gyrogy}
    For any $\delta\in(0, 1)$ and $\theta\in(0, \maxbid]$ and $e\in E$, 
    \begin{align*}
        \P\left[ \sum_{t=1}^T \textsf{w}^t(e)> \sum_{t=1}^T \widehat{\textsf{w}}^t(e) + \frac{\maxbid^2}{\theta}\log \frac{|E|}{\delta}\right] \leq \frac{\delta}{|E|}\,.
    \end{align*}
\end{lemma}

To upper bound the regret of \cref{alg:weight-pushing-adaptive-whp}, define $\Phi_t = \sum_{\path} \exp(\eta\sum_{\tau=1}^t \widehat{\textsf{w}}^\tau(\path))$ where $\widehat{\textsf{w}}^\tau(\path)=\sum_{e\in\path}\widehat{\textsf{w}}^\tau(e)$ for any path $\path$ in the DAG. Suppose
\begin{align}\label{def:eta-lambda}
  \eta=\frac{\lambda}{2(\numbid+1)\maxbid|E|}\quad\text{and}\quad 0< \lambda\leq \frac{1}{2}\,.
\end{align}

Observe that $\widehat{\textsf{w}}^t(e)=\frac{\textsf{w}^t(e)}{p^t(e)}\cdot\ind{e\in\path^t} + \frac{\theta}{p^t(e)}\implies \widehat{\textsf{w}}^t(e)\leq \frac{2\maxbid |E|}{\lambda}$ as $\textsf{w}^t(e)\leq \maxbid, \theta\leq \maxbid$ and $p^t(e)\geq \frac{\lambda}{|\mathcal{C}|}\geq \frac{\lambda}{|E|}$. So,
\begin{align}\label{eq:less-than-one}
    \eta\widehat{\textsf{w}}^\tau(\path) = \eta\sum_{e\in\path}\widehat{\textsf{w}}^\tau(e)  \leq \frac{2\eta(\numbid+1)\maxbid|E|}{\lambda} \stackrel{\eqref{def:eta-lambda}}{=} 1,
\end{align}

where the inequality holds as any path $\path$ in the DAG has at most $\numbid+1$ edges. Now,
\begin{align*}
    \frac{\Phi_t}{\Phi_{t-1}} &= \sum_{\path}\frac{ \exp(\eta\sum_{\tau=1}^{t-1} \widehat{\textsf{w}}^\tau(\path))}{\Phi_{t-1}}\cdot\exp(\eta\widehat{\textsf{w}}^t(\path))\\
    &\stackrel{\eqref{eq:def-P}}{=}\sum_\path\widehat{\P}^t(\path)\cdot\exp(\eta\widehat{\textsf{w}}^t(\path))\\
    &\stackrel{\eqref{eq:less-than-one}}{\leq} \sum_\path\widehat{\P}^t(\path)\cdot(1+\eta\widehat{\textsf{w}}^t(\path)+\eta^2\widehat{\textsf{w}}^t(\path)^2)\\
    &=1+ \sum_\path\widehat{\P}^t(\path)\cdot(\eta\widehat{\textsf{w}}^t(\path)+\eta^2\widehat{\textsf{w}}^t(\path)^2)\\
    &\leq 1+\frac{\eta}{1-\lambda}\sum_\path\P^t(\path)\widehat{\textsf{w}}^t(\path)+\frac{\eta^2}{1-\lambda}\sum_\path\P^t(\path)\widehat{\textsf{w}}^t(\path)^2\\
    &\leq \exp\left(\frac{\eta}{1-\lambda}\sum_\path\P^t(\path)\widehat{\textsf{w}}^t(\path)+\frac{\eta^2}{1-\lambda}\sum_\path\P^t(\path)\widehat{\textsf{w}}^t(\path)^2\right)
\end{align*}

where the first inequality uses $e^x\leq 1+x+x^2$ for $x\leq1$, second one follows as $\widehat{\P}^t(\path) = \frac{\P^t(\path)-\frac{\lambda}{|\mathcal{C}|}\cdot\ind{\path\in\mathcal{C}}}{1-\lambda} \leq \frac{\P^t(\path)}{1-\lambda}$ and third one is true as $1+x\leq e^x, \forall x$.
Now, 
\begin{align*}
    \sum_\path\P^t(\path)\widehat{\textsf{w}}^t(\path) &= \sum_\path\P^t(\path)\sum_{e\in\path}\widehat{\textsf{w}}^t(e)\\
    &= \sum_{e\in E}\widehat{\textsf{w}}^t(e)\sum_{\path:e\in\path}\P^t(\path)\\
    &\stackrel{\eqref{eq:w-hat-adaptive}}{=} \sum_{e\in E}\widehat{\textsf{w}}^t(e)p^t(e)\\
     &\stackrel{\eqref{eq:w-hat-adaptive}}{=}\sum_{e\in E}\textsf{w}^t(e)\cdot\ind{e\in\path^t} + \theta \\
     &= \textsf{w}^t(\path^t) + |E|\theta\,.
\end{align*}

Similarly,
\begin{align*}
    \sum_\path\P^t(\path)\widehat{\textsf{w}}^t(\path)^2 &\leq (\numbid+1)\sum_\path\P^t(\path)\sum_{e\in\path}\widehat{\textsf{w}}^t(e)^2\\
    &= (\numbid+1)\sum_{e\in E}\widehat{\textsf{w}}^t(e)^2p^t(e)\\
    &\stackrel{\eqref{eq:w-hat-adaptive}}{=}(\numbid+1)\sum_{e\in E}\widehat{\textsf{w}}^t(e)(\textsf{w}^t(e)\cdot\ind{e\in\path^t}+\theta)\\
    &\leq 2(\numbid+1)\maxbid\sum_{e\in E}\widehat{\textsf{w}}^t(e),
\end{align*}

where the first inequality holds due to Cauchy-Schwarz inequality. Combining everything,
\begin{align*}
  \frac{\Phi_t}{\Phi_{t-1}} &\leq  \exp\left(\frac{\eta}{1-\lambda}(\textsf{w}^t(\path^t) + |E|\theta)+\frac{2(\numbid+1)\maxbid\eta^2}{1-\lambda}\sum_{e\in E}\widehat{\textsf{w}}^t(e)\right) 
\end{align*}

Taking logarithms on both sides and adding from $t=1$ till $t=T$, we get
\begin{align*}
    \log \Phi_T - \log \Phi_0 &\leq \frac{\eta}{1-\lambda}\sum_{t=1}^T(\textsf{w}^t(\path^t) + |E|\theta)+\frac{2(\numbid+1)\maxbid\eta^2}{1-\lambda}\sum_{t=1}^T\sum_{e\in E}\widehat{\textsf{w}}^t(e)\\
    &\leq \frac{\eta}{1-\lambda}\sum_{t=1}^T(\textsf{w}^t(\path^t) + |E|\theta)+\frac{2(\numbid+1)\maxbid\eta^2|E|}{1-\lambda}\max_\path\sum_{t=1}^T\widehat{\textsf{w}}^t(\path),
\end{align*}

where the last inequality follows from the fact that \(\sum_{e\in E} \widehat{\textsf{w}}^t(e) \leq \sum_{\path \in \mathcal{C}} \widehat{\textsf{w}}^t(\path)\), since each edge $e$ is contained in at least one path \(\path \in \mathcal{C}\). Hence,
\begin{align*}
  \sum_{t=1}^T\sum_{e\in E}\widehat{\textsf{w}}^t(e) \leq \sum_{t=1}^T\sum_{\path\in\mathcal{C}}\widehat{\textsf{w}}^t(\path)=\sum_{\path\in\mathcal{C}}\sum_{t=1}^T\widehat{\textsf{w}}^t(\path) \leq |E|\max_\path\sum_{t=1}^T\widehat{\textsf{w}}^t(\path),
\end{align*}

where the last inequality holds as $|\mathcal{C}|\leq |E|$. Observe that, $\Phi_0 \leq \maxbid^\numbid$ and $\log \Phi_T \geq \eta\max_\path\sum_{t=1}^T\widehat{\textsf{w}}^t(\path)$. Thus,
\begin{align*}
    \max_\path\sum_{t=1}^T\widehat{\textsf{w}}^t(\path) - \frac{\numbid\log \maxbid}{\eta}
    &\leq \frac{1}{1-\lambda}\sum_{t=1}^T(\textsf{w}^t(\path^t) + |E|\theta)+\frac{2(\numbid+1)\maxbid\eta|E|}{1-\lambda}\max_\path\sum_{t=1}^T\widehat{\textsf{w}}^t(\path)\,.
\end{align*}

Rearranging the terms, we get
\begin{align*}
   (1-\lambda-2\eta(\numbid+1)\maxbid|E|)\max_\path\sum_{t=1}^T\widehat{\textsf{w}}^t(\path)-\sum_{t=1}^T\textsf{w}^t(\path^t) 
    &\leq \frac{\numbid\log \maxbid}{\eta}+ T|E|\theta\,. 
\end{align*}

Note that $1-\lambda-2\eta(\numbid+1)\maxbid|E| = 1-2\lambda\geq0$. By \cref{lem:gyrogy} and union bound over all edges $e\in E$, we get that with probability at least $1-\delta$, 
\begin{align*}
    \sum_{t=1}^T \widehat{\textsf{w}}^t(e) \geq \sum_{t=1}^T \textsf{w}^t(e) - \frac{\maxbid^2}{\theta}\log \frac{|E|}{\delta}, \forall e\in E\,.
\end{align*}

For any $\path$, summing over the edges, we get that with probability at least $1-\delta$,
\begin{align*}
    \sum_{t=1}^T\widehat{\textsf{w}}^t(\path) \geq \sum_{t=1}^T\textsf{w}^t(\path) - \frac{(m+1)\maxbid^2}{\theta}\log \frac{|E|}{\delta}\,.
\end{align*}

Hence,
\begin{align*}
   (1-\lambda-2\eta(\numbid+1)\maxbid|E|)\left(\max_\path\sum_{t=1}^T\textsf{w}^t(\path)-\frac{(m+1)\maxbid^2}{\theta}\log \frac{|E|}{\delta} 
    \right)- \sum_{t=1}^T\textsf{w}^t(\path^t)&\leq \frac{\numbid\log \maxbid}{\eta} + T|E|\theta\,.
\end{align*}

Using the fact that $\lambda = 2\eta(\numbid+1)\maxbid|E|$ and $\max_\path\sum_{t=1}^T\textsf{w}^t(\path)\leq \maxbid T$, we get
\begin{align*}
    \max_\path\sum_{t=1}^T\textsf{w}^t(\path)
    - \sum_{t=1}^T\textsf{w}^t(\path^t)&\leq \frac{\numbid\log \maxbid}{\eta} + 4\eta(\numbid+1)\maxbid^2|E| T + T|E|\theta + \frac{(m+1)\maxbid^2}{\theta}\log \frac{|E|}{\delta} 
\end{align*}

Setting
\begin{align*}
    \eta = \frac{1}{2\maxbid}\sqrt{\frac{\numbid \log \maxbid}{(\numbid+1)|E| T}}, ~~\theta = \maxbid\sqrt{\frac{(m+1)}{T|E|}\log \frac{|E|}{\delta}}, \quad\text{and}\quad \lambda = \sqrt{\frac{\numbid(\numbid+1)|E|\log \maxbid}{T}}
\end{align*}

and $T\geq \max\{4\numbid(\numbid+1)|E|\log \maxbid, \log (|E|/\delta)\}$~(to ensure $\theta\in(0,  \maxbid]$ and $\lambda\in(0, \frac{1}{2}]$) we obtain,
\begin{align*}
   \max_\path\sum_{t=1}^T\textsf{w}^t(\path)
    - \sum_{t=1}^T\textsf{w}^t(\path^t)&\leq 4\maxbid\sqrt{\numbid(\numbid+1)|E| T \log \maxbid} + 2\maxbid\sqrt{(\numbid+1)|E|T \log (|E|/\delta)}\,. 
\end{align*}

By \cref{thm:DAG-base-strategy} and \cref{alg:weight-pushing-adaptive-whp}, $\max_\path\sum_{t=1}^T\textsf{w}^t(\path) = \max_{\ibid\in\optufclass{\numbid}}\sum_{t=1}^T\val(\ibid; \otherbid{t})$ and $\sum_{t=1}^T\textsf{w}^t(\path^t)=\sum_{t=1}^T\val(\ibid^t; \otherbid{t})$. Furthermore, $|E|=O(\numbid\maxbid^2)$. Hence, with probability at least $1-\delta$,
\begin{align*}
   \widehat{\textsf{REG}}_{\text{adap}}  = \max_{\ibid\in\optufclass{\numbid}}\sum_{t=1}^T\val(\ibid; \otherbid{t})
    - \sum_{t=1}^T\val(\ibid^t; \otherbid{t})&\leq O(\numbid^{3/2}\maxbid^2\sqrt{T \log \maxbid} + \numbid\maxbid^2\sqrt{T\log (\maxbid/\delta)})\,. 
\end{align*}

\textbf{Time Complexity.} The computational bottleneck is in computing the edge weight estimator in \cref{eq:w-hat-adaptive}. Specifically, we need to compute the marginal probabilities $p^t(e)$ efficiently. From the proof of \cref{thm:bandit}, we know that this can be done efficiently in $O(\numbid\maxbid^2)$ time. Thus, the per-round running time of \cref{alg:weight-pushing-adaptive-whp} is $O(\numbid\maxbid^2)$ time per-round.

\subsubsection{Proof of \cref{lem:gyrogy}}
    Fix any $e\in E$. By Chernoff bound, for any $z>0$ and $\zeta>0$, 
    \begin{align*}
       \P\left[\sum_{t=1}^T \textsf{w}^t(e)> \sum_{t=1}^T \widehat{\textsf{w}}^t(e) + z\right] \leq e^{-\zeta z}\E\left[\exp\left(\zeta\left(\sum_{t=1}^T \textsf{w}^t(e)- \sum_{t=1}^T \widehat{\textsf{w}}^t(e)\right)\right)\right] 
    \end{align*}
    
    Setting $z=\frac{\maxbid^2}{\theta}\log \frac{|E|}{\delta}$ and $\zeta=\frac{\theta}{\maxbid^2}$, we get 
    \begin{align*}
        \P\left[ \sum_{t=1}^T \textsf{w}^t(e)> \sum_{t=1}^T \widehat{\textsf{w}}^t(e) + \frac{\maxbid^2}{\theta}\log \frac{|E|}{\delta}\right] \leq \frac{\delta}{|E|}\E\left[\exp\left(\zeta\left(\sum_{t=1}^T \textsf{w}^t(e)- \sum_{t=1}^T \widehat{\textsf{w}}^t(e)\right)\right)\right]\,.
    \end{align*}

    To complete the proof, we need to show that   $\E[\exp(\zeta(\sum_{t=1}^T \textsf{w}^t(e)- \sum_{t=1}^T \widehat{\textsf{w}}^t(e)))] \leq 1$. Define 
    \begin{align*}
       X_t = \exp\left(\zeta\left(\sum_{\tau=1}^t \textsf{w}^\tau(e)- \sum_{\tau=1}^t\widehat{\textsf{w}}^\tau(e)\right)\right)\quad \text{and}\quad X_0=1\,. 
    \end{align*}
    
    We show that for any $t$, $\E[X_t|\ibid_1, \dots, \ibid_{t-1}]:=\E_t[X_t]\leq X_{t-1}$. Computing conditional expectations recursively results in $\E[X_t]\leq 1, \forall t$ which proves the result. To prove $\E_t[X_t]\leq X_{t-1}$, note that $X_t = X_{t-1}\exp(\zeta(\textsf{w}^t(e)- \widehat{\textsf{w}}^t(e)))$. So,
    \begin{align*}
        \E_t[X_t] &= X_{t-1}\E_t\left[\exp(\zeta(\textsf{w}^t(e)- \widehat{\textsf{w}}^t(e)))\right]\\
        &\stackrel{\eqref{eq:w-hat-adaptive}}{=} X_{t-1}\E_t\left[\exp\left(\frac{\theta}{\maxbid^2}\left(\textsf{w}^t(e)- \frac{\textsf{w}^t(e)}{p^t(e)}\cdot\ind{e\in\path^t} - \frac{\theta}{p^t(e)}\right)\right)\right]\\
        &=X_{t-1}e^{-\frac{\theta^2}{\maxbid^2p^t(e)}}\E_t\left[\exp\left(\frac{\theta}{\maxbid^2}\left(\textsf{w}^t(e)- \frac{\textsf{w}^t(e)}{p^t(e)}\cdot\ind{e\in\path^t}\right)\right)\right]\\
        &\leq X_{t-1}e^{-\frac{\theta^2}{\maxbid^2 p^t(e)}}\E_t\left[1 + \frac{\theta}{\maxbid^2}\left(\textsf{w}^t(e)- \frac{\textsf{w}^t(e)}{p^t(e)}\cdot\ind{e\in\path^t}\right)+\frac{\theta^2}{\maxbid^4}\left(\textsf{w}^t(e)- \frac{\textsf{w}^t(e)}{p^t(e)}\cdot\ind{e\in\path^t}\right)^2\right] \\
        &=X_{t-1}e^{-\frac{\theta^2}{\maxbid^2 p^t(e)}}\E_t\left[1 + \frac{\theta^2}{\maxbid^4}\left(\textsf{w}^t(e)- \frac{\textsf{w}^t(e)}{p^t(e)}\cdot\ind{e\in\path^t}\right)^2\right] \\
        &\leq X_{t-1}e^{-\frac{\theta^2}{\maxbid^2 p^t(e)}}\E_t\left[1 + \frac{\theta^2}{\maxbid^4}\left( \frac{\textsf{w}^t(e)}{p^t(e)}\cdot\ind{e\in\path^t}\right)^2\right]\\
        &\leq X_{t-1}e^{-\frac{\theta^2}{\maxbid^2 p^t(e)}}\left(1 + \frac{\theta^2}{ \maxbid^2p^t(e)}\right)\\
        &\leq X_{t-1}, 
    \end{align*}
    which is the desired result. Here, the first inequality follows as $\frac{\theta}{\maxbid^2}\left(\textsf{w}^t(e)- \frac{\textsf{w}^t(e)}{p^t(e)}\cdot\ind{e\in\path^t}\right) \leq \theta/\maxbid \leq 1$ and $e^x\leq 1+x+x^2, \forall x\leq 1$, the fourth equality is true as $\E_t[\frac{\textsf{w}^t(e)}{p^t(e)}\cdot\ind{e\in\path^t}]=\textsf{w}^t(e)$, the third inequality follows as $\textsf{w}^t(e)\leq \maxbid$ and the final inequality holds true as $1+x\leq e^x, \forall x$.

\subsection{Proof of \cref{thm:cumulative-impossible}}
Let $\numbid=\maxbid=1$ and $K=2$.  Fix any $c\in(0, 1]$ and choose $\epsilon\in(0, \frac{c}{2})$. Assume $\epsilon T \in \N$. For the first $(1-\epsilon) T$ rounds, the top $K=2$ competing bids are $[3T, 1+\frac{\epsilon}{2}]$. For the remaining $\epsilon T$ rounds, the competing bids are $[\epsilon, \epsilon]$. Define $\v=1$. 

Here, $\ibid^*=1+\epsilon$ is the fixed hindsight optimal strategy that satisfies the cumulative RoI constraint. The total value obtained by $\ibid^*$ is $T$~(1 in each round). Similarly, the total payment is $(1+\epsilon)\cdot(1-\epsilon)T+\epsilon\cdot\epsilon T=T$. 

    By \cref{assumpt:1}, the bidder follows safe strategies for the first $(1-\epsilon)T$ rounds and does not win any units. In the remaining $\epsilon T$ rounds, it can win at most $\maxbid=1$ unit per round. Thus, the maximum possible value obtained by the bidder is $\epsilon T$. Hence, $c$-$\textsf{REG}_{\text{cumul}}\geq c\cdot T-\epsilon T\geq \frac{cT}{2}=\Omega(T)$. 
    
\subsection{Proof of \cref{thm:regret-UB-stoc}}\label{apx:ssec:time-varying}
Before proving the result, we formally present the learning algorithm for stochastic contexts in \cref{alg:weight-pushing-stoc}. Note that the learner does not need to know the distribution $F_\v$ from which the contexts are sampled. The space complexity and per-round time complexity of \cref{alg:weight-pushing-stoc} is $O(\numbid\maxbid^2|\mathcal{V}|)$.

\begin{algorithm}[!tbh]
\caption{Learning Safe Bidding Strategies~(Full Information, Stochastic Contexts)}
\label{alg:weight-pushing-stoc}
\small{
\begin{algorithmic}[1]
\Require Set of valuation vectors, $\mathcal{V}$, time horizon $T$, learning rate $0< \eta\leq \frac{1}{\maxbid}$. Initialize $\varphi^0(e; \v)=1$ and $\textsf{w}^0(e; \v)=0, \forall e\in E, \forall \v\in\mathcal{V}$.
\For{$t = 1, 2, \dots, T$}
    \State Observe valuation curve $\v^t\in\mathcal{V}$ which is sampled from $F_\v$.
    \State Construct $\mathcal{G}^t(V, E, \v)$ similar to $\mathcal{G}(V, E)$ without weights for all $\v\in\mathcal{V}$.
    \State $\textsf{UPDATE}:$ Obtain edge probabilities $\varphi^t(\cdot; \v)$ for all $\v\in\mathcal{V}$ as follows:
    \For{$\v\in\mathcal{V}$}
    
    
    
    \State Set $\Gamma^{t-1}(d; \v)=1$ and recursively compute in bottom-to-top fashion for every node $u$ in $\mathcal{G}^t(V, E, \v)$:
    \begin{align*}
        \Gamma^{t-1}(u; \v)=\sum_{v:u\to v=e\ni E}\Gamma^{t-1}(v; \v)\cdot\varphi^{t-1}(e; \v)\cdot\exp(\eta \textsf{w}^{t-1}(e; \v ))
    \end{align*}

    \State For edge $e=u\to v$ in $\mathcal{G}^t(V, E, \v)$, update:
       $ \varphi^t(e; \v)=\varphi^{t-1}(e; \v)\cdot\exp(\eta \textsf{w}^{t-1}(e; \v))\cdot\frac{\Gamma^{t-1}(v; \v)}{\Gamma^{t-1}(u; \v)}.$
\EndFor
    \State $\textsf{SAMPLE}$: Define initial node $u=s$ and path $\path^t=s$. 
    \While{$u\neq d$}
    \State Sample $v$ with probability $\varphi^t(u\to v; \v^t)$.
    \State Append $v$ to the path $\path^t$; set $u\gets v$.
    \EndWhile
    \State $\textsf{MAP}$: If $\path^t=s\to (1, z_1)\to\dots\to(k, z_k)\to d$ for some $k\in[\numbid]$, submit $\ibid^t=\langle(b_1, q_1), \dots, (b_k, q_k)\rangle$ where
    \begin{align*}
      b_\l=w_{z_\l}^t \quad\text{and}\quad q_\l=z_\l-z_{\l-1}, ~\forall  \l\in [k]\,.
    \end{align*}
    
    where $w_{z_\l}^t$ is the $z_\l^{th}$ entry of the average cumulative valuation vector for all $\l\in[k]$ corresponding to $\v^t$.
    \State The bidder observes $\otherbid{t}$ and sets the edge weights for $e=x\to y$ as follows for each $\v=[v_1, \dots, v_\maxbid]\in\mathcal{V}$: 

(i) If \(x=(\l-1, j)\) and \(y=(\l, j')\) with \(\l\in[\numbid]\) and \(j < j'\),
\begin{align*}
    \textsf{w}^t(e; \v) = \sum_{k=j+1}^{j'}v_k\cdot\ind{w_{j'}\geq \ordotherbid{t}{k}}\,.
\end{align*}

(ii) If \(x=(\l, j)\) and \(y=d\), then \(\textsf{w}^t(e, \v) = 0, \forall \l \in [\numbid], j \in [\maxbid]\).
\EndFor
\end{algorithmic}}
\end{algorithm}

\textbf{Regret Analysis.} For any fixed $\v\in\mathcal{V}$, define $\Phi_t(\v) = \sum_{\ibid\in\optufclass{\numbid}(\v)}\exp(\eta\sum_{\tau=1}^t\val(\ibid; \otherbid{\tau}))$. Then, 
\begin{align*}
    \frac{\Phi_{t}(\v)}{\Phi_{t-1}(\v)} &= \sum_{\ibid\in\optufclass{\numbid}(\v)}\frac{\exp(\eta\sum_{\tau=1}^{t-1}\val(\ibid; \otherbid{\tau}))}{\Phi_{t-1}(\v)}\cdot\exp(\eta\val(\ibid; \otherbid{t})) \\
    &= \sum_{\ibid\in\optufclass{\numbid}(\v)}\P[\ibid^t=\ibid|\v^t=\v]\cdot\exp(\eta\val(\ibid; \otherbid{t})),
\end{align*}
where the last line follows because for a fixed context, \cref{alg:weight-pushing-stoc} is equivalent to \cref{alg:weight-pushing}, which is an equivalent implementation of the Hedge algorithm~(see proof of \cref{thm:full-info}).

For $x\leq 1$, $e^x\leq 1+x+x^2$ and $\log(1+x)\leq x$ for all $x\geq0$. As $\eta\val(\ibid; \otherbid{t}) \leq \eta \maxbid\leq 1$, we have:
\begin{align*}
   \frac{\Phi_{t}(\v)}{\Phi_{t-1}(\v)} &\leq  \sum_{\ibid\in\optufclass{\numbid}(\v)}\P[\ibid^t=\ibid|\v^t=\v]\cdot(1+\eta\val(\ibid; \otherbid{t}) + \eta^2\val(\ibid; \otherbid{t})^2)\\
   &\leq \exp\Big(\sum_{\ibid\in\optufclass{\numbid}(\v)}\P[\ibid^t=\ibid|\v^t=\v](\eta\val(\ibid; \otherbid{t}) + \eta^2\val(\ibid; \otherbid{t})^2)\Big)
\end{align*}

Taking logarithms on both sides and summing over $t=1$ to $t=T$, we get
\begin{align*}
    \log \Phi_T(\v)-\log \Phi_0(\v) \leq \sum_{t=1}^T\sum_{\ibid\in\optufclass{\numbid}(\v)}\P[\ibid^t=\ibid|\v^t=\v](\eta\val(\ibid; \otherbid{t}) + \eta^2\val(\ibid; \otherbid{t})^2)
\end{align*}

Define $\Phi_0(\v)=\sum_{\ibid\in\optufclass{\numbid}(\v)} 1\leq \maxbid^\numbid$. Moreover, for any stationary policy $\pi\in\Pi$, per \cref{eq:Pi}, 
\begin{align*}
  \log \Phi_T(\v) \geq \eta\sum_{t=1}^T\val(\pi(\v); \otherbid{t}) \,. 
\end{align*}

Hence, 
\begin{align}\label{eq:abs-UB-con-regret-stoc}
&\sum_{t=1}^T\val(\pi(\v); \otherbid{t}) - \sum_{t=1}^T\sum_{\ibid\in\optufclass{\numbid}(\v)}\P[\ibid^t=\ibid|\v^t=\v]\val(\ibid; \otherbid{t})\nonumber\\
&\leq \frac{\numbid\log\maxbid}{\eta} +\eta \sum_{t=1}^T\sum_{\ibid\in\optufclass{\numbid}(\v)}\P[\ibid^t=\ibid|\v^t=\v]\val(\ibid; \otherbid{t})^2
\end{align}
Taking expectations with respect to the randomness of the context~(valuation vectors) and using the fact that the choice of $\pi$ was arbitrary, we get:
\begin{align*}
   \textsf{REG}_{sto}&=\max_{\pi\in\Pi}\sum_{t=1}^T\E_{\v_t\sim F_\v}[\val(\pi(\v^t); \otherbid{t})] - \sum_{t=1}^T\E[\val(\ibid^t; \otherbid{t})] \\
   &\stackrel{\eqref{eq:abs-UB-con-regret-stoc}}{\leq} \frac{\numbid\log\maxbid}{\eta} +\eta \E_{\v^t\sim F_\v}\left[\sum_{t=1}^T\sum_{\ibid\in\optufclass{\numbid}(\v)}\P[\ibid^t=\ibid|\v^t=\v]\val(\ibid; \otherbid{t})^2\right]\\
   &\leq \frac{\numbid\log\maxbid}{\eta} +\eta \maxbid^2\sum_{t=1}^T\sum_{\v}\P[\v^t=\v]\sum_{\ibid\in\optufclass{\numbid}(\v)}\P[\ibid^t=\ibid|\v^t=\v]=\frac{\numbid\log\maxbid}{\eta} +\eta \maxbid^2T,
\end{align*}
where the second inequality follows as $\val(\ibid; \otherbid{t})\leq \maxbid$. Setting, $\eta=\frac{1}{\maxbid}\sqrt{\frac{\numbid\log\maxbid}{T}}$, we get
\begin{align*}
  \textsf{REG}_{sto}&\leq O(\maxbid\sqrt{\numbid T\log \maxbid})\,.  
\end{align*}

\subsection{Adversarial Contexts}\label{apx:ssec:adv-contexts}
In this section, we consider the case when the contexts~(valuation vectors) are generated adversarially in advance (akin to an oblivious adversary). The regret in this setting is defined as
\begin{align*}
    \textsf{REG}_{adv} = \max_{\pi^*\in\Pi}\sum_{t=1}^T \val(\pi^*({\v^t}); \otherbid{t}) - \sum_{t=1}^T\E[\val(\ibid^t; \otherbid{t})],
\end{align*}
where $\ibid^t\in \optufclass{\numbid}(\v^t), \forall t\in[T]$ and $\Pi$ is defined in \cref{eq:Pi}. Observe that for any $\pi$, $\sum_{t=1}^T \val(\pi({\v^t}); \otherbid{t}) = \sum_{\v\in\mathcal{V}}\sum_{t=1}^T\val(\pi({\v}); \otherbid{t})\cdot\ind{\v^t=\v}$. Hence, the optimal stationary policy in this setting is 
\begin{align}\label{eq:pi-star-adv}
   \pi^*(\v) = \argmax_{\ibid\in\optufclass{\numbid}(\v)}\sum_{t=1}^T\val(\ibid; \otherbid{t})\cdot\ind{\v^t=\v}, \quad\forall \v\in\mathcal{V}.
\end{align}
This benchmark is the natural choice for adversarial rewards and contexts~\citep[Chapter 18.1]{lattimore2020bandit}. Note that the optimal policy $\pi^*(\cdot)$ differs between the stochastic and adversarial settings. In the stochastic case, the benchmark assumes that the clairvoyant knows the competing bids and distribution over the contexts \textit{a priori} but does not observe the realized contexts in each round in advance. Hence, their goal is to maximize the expected total value over the $T$ rounds \textit{ex ante}. In contrast, under adversarial contexts, the bidder is assumed to know both the contexts and competing bids \textit{a priori}, and thus aims to maximize the total value \textit{ex post} over the $T$ rounds. 

In this setting, the bidder maintains $|\mathcal{V}|$ copies of \cref{alg:weight-pushing}, each corresponding to a context in $\mathcal{V}$. In each round $t\in[T]$, after observing $\v^t$, the bidder updates the edge probabilities only in the algorithm corresponding to $\v^t$ and submits a bidding strategy in $\optufclass{\numbid}(\v^t)$. The space complexity of \cref{alg:weight-pushing-adv} is $O(\numbid\maxbid^2|\mathcal{V}|)$ while the per-round time complexity is $O(\numbid\maxbid^2)$.

\begin{algorithm}[!tbh]
\caption{Learning Safe Bidding Strategies~(Full Information, Adversarial Contexts)}
\label{alg:weight-pushing-adv}
\small{
\begin{algorithmic}[1]
\Require Initialize $|\mathcal{V}|$ copies of \cref{alg:weight-pushing}, each corresponding to a context in $\mathcal{V}$.
\For{$t = 1, 2, \dots, T$}
    \State Observe valuation curve $\v^t\in\mathcal{V}$.
    \State Obtain bidding strategy and update edge probabilities by running \cref{alg:weight-pushing} corresponding to the context $\v^t$.
\EndFor
\end{algorithmic}}
\end{algorithm}

Observe that for stochastic as well as adversarial contexts, a separate DAG is maintained for each context. However, in any given round $t$, while the edge probabilities in each DAG are updated under stochastic contexts the edge probabilities in only the DAG corresponding to the context $\v^t$ are updated for adversarial contexts. 

\begin{theorem}\label{thm:regret-UB-adv}
    When the contexts are adversarial, there exists an algorithm that runs in \(\text{poly}(\numbid, \maxbid)\) space and per-round time and achieves \(\textsf{REG}_{adv} \leq O(\maxbid\sqrt{\numbid |\mathcal{V}| T\log \maxbid})\) under full information feedback.
\end{theorem}
\begin{proof}
    Recall that when the contexts are adversarial, regret is defined as
\begin{align*}
\textsf{REG}_{adv} = \sum_{t=1}^T \val(\pi^*({\v^t}); \otherbid{t}) - \sum_{t=1}^T\E[\val(\ibid^t; \otherbid{t})]
\end{align*}
where $\pi^*(\cdot)$ is defined in \cref{eq:pi-star-adv}. Let $T_\v$ be the set of rounds in which the context is $\v\in\mathcal{V}$. Then, 
\begin{align*}
    \textsf{REG}_{adv} &= \sum_{\v\in\mathcal{V}}\left\{\sum_{t\in T_\v} \val(\pi^*({\v}); \otherbid{t}) - \sum_{t\in T_\v}\E[\val(\ibid^t; \otherbid{t})]\right\}\\
     &\stackrel{\eqref{eq:pi-star-adv}}{=} \sum_{\v\in\mathcal{V}}\left\{\max_{\ibid\in\optufclass{\numbid}(\v)}\sum_{t\in T_\v} \val(\ibid; \otherbid{t}) - \sum_{t\in T_\v}\E[\val(\ibid^t; \otherbid{t})]\right\}\\
     &\leq \sum_{\v\in\mathcal{V}}O(\maxbid\sqrt{\numbid |T_\v|\log \maxbid})\,,
\end{align*}
where the last inequality follows as a separate instance of \cref{alg:weight-pushing} is utilized for each context and using the regret upper bound of \cref{alg:weight-pushing} in the full information setting.\footnote{Although $T_\v$ is not known in advance~(which is required for tuning the learning rate), the bidder can use the standard \textit{doubling trick}~\citep{cesa2006prediction} to implement \cref{alg:weight-pushing} which yields the same regret guarantee up to constant factors.} Finally, by Jensen's inequality, we get that $\textsf{REG}_{adv} \leq O(\maxbid\sqrt{\numbid |\mathcal{V}|T\log \maxbid})$.
\end{proof}

We complement the regret upper bound for adversarial contexts with the following lower bound:

\begin{theorem}\label{thm:regret-LB-adv}
Let $\maxbid\geq2$, $\numbid=1$ and $|\mathcal{V}|<\infty$. There exist competing bids, $[\otherbid{t}]_{t\in[T]}$, such that for adversarial contexts, under any learning algorithm, $\E[\textsf{REG}_{adv}]=\Omega(\maxbid\sqrt{|\mathcal{V}|T})$ in the full information setting.
\end{theorem}


\begin{proof}
We begin by presenting a corollary of \cref{thm:regret-LB}:
\begin{corollary}\label{cor:regret-LB}
    Fix any $\theta>0$ and assume that $v_j\in[0, \theta], \forall j\in[\maxbid]$. Then, for $\maxbid\geq2$ and $\numbid=1$, there exist competing bids, $[\otherbid{t}]_{t\in[T]}$, such that, under any learning algorithm, $\E[\textsf{REG}]\geq c\theta\maxbid\sqrt{T}$ for some $c>0$, in the full information setting where $\textsf{REG}$ is defined in \cref{eq:safe-regret}.
\end{corollary}
\begin{proof}[Proof of \cref{cor:regret-LB}]
    The result follows directly by considering the scaled valuation vector $\theta\v$ (as defined in \cref{eq:LB-v}), the scaled competing bid profiles $\theta\otherbid{\clubsuit}$ and $\theta\otherbid{\diamondsuit}$ (as defined in \cref{eq:regret-LB-competing-bids}), and applying the analysis of \cref{thm:regret-LB}.
\end{proof}

Fix any $N\in\N$ and define $\mathcal{V}=\{\alpha_i\v: \alpha_i=\frac{i}{N},~~\text{for}~~i\in[N]\}$ where $\v$ is defined in \cref{eq:LB-v}. Note that $\v'\in[0, 1]^\maxbid$ for all $\v'\in\mathcal{V}$. Assume that $T$ is a multiple of $N$. Partition the $T$ rounds into $N$ epochs with the $i^{th}$ epoch spanning the rounds $t=\frac{(i-1)T}{N}+1,\dots,  \frac{iT}{N}$. For all the rounds in the $i^{th}$ epoch, set the context as $\alpha_i\v$. Then, for some $c'>0$,
\begin{align*}
    \E[\textsf{REG}_{adv}] = \sum_{i=1}^{N}\E\left[\max_{\ibid\in\optufclass{\numbid}(\alpha_i\v)}\sum_{t=\frac{(i-1)T}{N}+1}^{\frac{iT}{N}} \val(\ibid; \otherbid{t}) - \sum_{t=\frac{(i-1)T}{N}+1}^{\frac{iT}{N}}\val(\ibid^t; \otherbid{t})\right]\geq c'\maxbid\sqrt{\frac{T}{N}}\sum_{i=1}^{N}\alpha_i\,.
\end{align*}
The inequality follows from \cref{cor:regret-LB}. As $\sum_{i=1}^{N}\alpha_i\geq N/2$, $\E[\textsf{REG}_{adv}]\geq \frac{1}{2}c'\maxbid
N\sqrt{\frac{T}{N}}=\Omega(\maxbid\sqrt{|\mathcal{V}|T})$ where the last equality holds as $|\mathcal{V}|=N$.
\end{proof}
\section{Simulation Details}\label{apx:exp_details}



We sample the values from the $\text{Unif}[0, 1]$ distribution.  In each simulation, we sample $T\sim\text{Unif}[100, 300]$ auctions and let $\maxbid\sim\text{Unif}[10, 80]$. We vary $\numbid=1$ to $\numbid=10$ and average over 100 simulations to obtain plots in \cref{fig:sims}. As computing $\optvalue{\numbid}$~(the value obtained by the optimal bidding strategy in $\feasclass{\numbid}$) can be non-trivial, we obtain a uniform upper bound for $\optvalue{\numbid}$ that is \textit{independent of $\numbid$}~(see details in \cref{apx:ilp}). 

\subsection{Reconstructing Individual Bid Data}\label{apx:data}

We obtained the publicly available auction data for $T_{\max}=443$ EU ETS emission permit auctions held in 2022 and 2023~\citep{eex-eua-primary-auction}. For each auction indexed by $t\in[T_{\max}]$, we have the following relevant information: the minimum bid~($b_{\min}^t$), the maximum bid~($b_{\max}^t$), the average of the bids~($b_{\text{avg}}^t$), the median of the  bids~($b_{\text{med}}^t$), and the number of bid-quantity pairs submitted~($n_{\text{sub}}^t$). We normalized the bids to be in $[0, 1]$. For all rounds $t$, $b_{\text{avg}}^t\approx b_{\text{med}}^t$~(linear regression yields coefficient 1.01 and intercept $-0.008$).

Upon further investigation, we observed that, except a few, a significant number of auctions had either $b_{\min}^t\approx b_{\text{avg}}^t\ll b_{\max}^t$~(Type I) or $b_{\min}^t\ll b_{\text{avg}}^t\approx b_{\max}^t$~(Type II). As $b_{\text{avg}}^t\approx b_{\text{med}}^t, \forall t$, we deduce that for Type I, most of the bids are concentrated in the interval $[b_{\min}^t, 2b_{\text{avg}}^t-b_{\min}^t]$ whereas for Type II, most of the bids are in the interval $[2b_{\text{avg}}^t-b_{\max}^t, b_{\max}^t]$. We posit that for Type I~(resp. Type II) auctions, $f\in(0, 1)$ fraction of the bids~($n_{\text{sub}}^t$) are in $[b_{\min}^t, 2b_{\text{avg}}^t-b_{\min}^t]$~(resp. $[2b_{\text{avg}}^t-b_{\max}^t, b_{\max}^t]$)  and the $1-f$ fraction of bids are in $[2b_{\text{avg}}^t-b_{\min}^t, b_{\max}^t]$~(resp. $[b_{\min}^t,2b_{\text{avg}}^t-b_{\max}^t]$). If for Type I (resp. Type II) auctions, $2b_{\text{avg}}^t>b_{\min}^t+ b_{\max}^t$~(resp. $2b_{\text{avg}}^t<b_{\min}^t+ b_{\max}^t$), we assume that all the bids are uniformly present in the interval $[b_{\min}^t, b_{\max}^t]$. With these assumptions, we generate individual bid data for each auction by sampling uniformly from these intervals.

After generating the individual bid data, we compute the metrics for the reconstructed bids (say $\widehat{b}_{\text{avg}}^t$) for each auction and reject those with a relative error of at least $\delta$~(tolerance). For our simulations, we set $\delta = 0.05$. We perform a grid search for $f$ to maximize the number of auctions where the metrics of the reconstructed data are within $\delta$ relative error of the actual metrics, and obtain that $f = 0.97$. Following this pre-processing, we have reconstructed individual bid data for $T = 341$ auctions. The bids are normalized to be in $[0, 1]$. 
\subsection{An Uniform Upper Bound for $\optvalue{\numbid}.$}\label{apx:ilp}
For any bid history, $\hist=[\otherbid{t}]_{t\in[T]}$, suppose $\optbid{\numbid}$ is allocated $r_t$ units in any round $t$. Then, by \cref{lem:approx1_multiple}~(1), we know that $(w_{r_t}, r_t)$ also obtains $r_t$ units in round $t$. So,
\begin{align*}
    \val(\optbid{\numbid}; \otherbid{t}) &= \val((w_{r_t}, r_t); \otherbid{t}) \leq \max_{\ibid\in\optufclass{1}}\val(\ibid; \otherbid{t})\\
    \implies\optvalue{\numbid}&=\sum_{t=1}^T\val(\optbid{\numbid}; \otherbid{t})\leq \sum_{t=1}^T\max_{\ibid\in\optufclass{1}}\val(\ibid; \otherbid{t})\,.
\end{align*}
\section{Resolving Ties}\label{apx:sec:ties}
We first emphasize that the effect of resolution of ties on the objective function under the value maximization behavioral model is fundamentally different from the quasilinear utility maximization model. To illustrate this, consider a second price auction with a single indivisible item with two bidders each valuing the item equally at $v$. Assume that both the bidders bid $v$ and there exists a definitive tie breaking rule~(either randomized or deterministic). If the bidders are considered quasilinear utility maximizers, the objective function value~(value obtained minus the payments) for both the bidders is 0. However, under the value maximization model~(assume $\gamma=0$ for both the bidders), the objective function value for the winning bidder is $v$ and 0 for the losing bidder.\footnote{Although, the total (liquid) welfare is $v$ in both the cases.} So, we need to carefully analyze the key results of our work in the event of ties.

In the context of multi-unit auctions, we can classify ties into two informal types: (a) `good ties'---ties occurring at any bid other than the last accepted bid (LAB) and (b) `bad ties', which are the ties occurring at the LAB. All the results in this work are unaffected in case only `good ties' occur. Thus, we focus on `bad ties' in this paper. For this discussion, we consider a public, deterministic tie breaking rule under which ties are always broken in favor of lower indexed bidder. In the presence of ties, the undominated class of safe bidding strategies is still the class of nested strategies~(\cref{thm:opt-bid}). The tie breaking rule is incorporated into the decomposition in \cref{lem:decomp} and computing the edge weights in the DAG in the offline and online settings in \cref{sec:learning-safe}, ensuring that the maximum weight path in the DAG gives the optimal offline solution and \cref{alg:weight-pushing} achieves sublinear regret in the online setting.


\end{document}